%% file: main.tex
\documentclass[12pt]{article}

\input{Z_packages.tex}
\input{Z_commands.tex}
\DeclareRobustCommand{\VAN}[2]{#1}  

\begin{document}
\doparttoc 
\faketableofcontents 

\def\spacingset#1{\renewcommand{\baselinestretch}%
{#1}\small\normalsize} \spacingset{1}


  \title{\bf Causal Inference with Corrupted Data: \\
Measurement Error, Missing Values, \\Discretization, and Differential Privacy}
  \author{
  Anish Agarwal \\
   Columbia University
  \and Rahul Singh\thanks{We thank Alberto Abadie, Isaiah Andrews, Joshua Angrist, David Autor, Abhijit Banerjee, David Bruns-Smith, Victor Chernozhukov, Avi Feller, Guido Imbens, Anna Mikusheva, Ismael Mourifie, Sendhil Mullainathan, Whitney Newey, Elizabeth Ogburn, James Robins, Devavrat Shah, Vasilis Syrgkanis, Panos Toulis, Suhas Vijaykumar, and many seminar participants for helpful comments, particularly the UC Berkeley community. We thank Leon Deng, Ajinkya Gundaria, William Liu, and Caleb Rollins for excellent research assistance. Rahul Singh thanks the Jerry Hausman Dissertation Fellowship. Part of this work was done while both authors visited the Simons Institute for the Theory of Computing. 
  }\\
  Harvard University}
  \date{Original draft: July 2021. This draft: February 2024.}
\maketitle
%
\vspace{-10pt}
\input{0_abstract}

\noindent%

\noindent%
\vfill




\newpage

\cleardoublepage
\pagenumbering{arabic}

\spacingset{1.5}

\input{1_intro}
\input{2_related}
\input{3_model}
\input{4_estimation}

\input{5_theory}
\input{6_application}

\input{7_conclusion}

\spacingset{1.5}
\bibliographystyle{apalike}
\typeout{}
\DeclareRobustCommand{\VAN}[2]{#2}  
\DeclareRobustCommand{\VAN}[2]{#1}  

\input{main.bbl}
\spacingset{1.5}
\newpage

\appendix
\addcontentsline{toc}{section}{Appendix} 
\part{Appendix} 

\spacingset{1}
\parttoc 
\newpage 
\spacingset{1.5}

\input{A_matrix}

\input{B_regression}
\input{C_riesz}
\input{D_target}

\input{E_examples}
\input{F_dictionary}
\input{G_matrix}
\input{H_regression}
\input{I_riesz}
\input{J_target}
\input{K_factor}

\input{L_sim}
\input{M_privacy}

\end{document}

%% file: Z_packages.tex
\usepackage[utf8]{inputenc} 
\usepackage[T1]{fontenc}    
\usepackage{hyperref}       
\usepackage{xurl}           
\usepackage{booktabs}       
\usepackage{amsfonts}       
\usepackage{nicefrac}       
\usepackage{microtype}      
\usepackage[pdftex]{graphicx}  
\usepackage{float}
\usepackage{amsthm,amsmath,amssymb}
\usepackage{mathtools}
\usepackage{amsfonts}
\usepackage{amssymb}
\floatstyle{ruled}
\usepackage{xcolor}
\usepackage{bbm}

\usepackage{adjustbox}
\usepackage{wrapfig}
\usepackage{caption}
\usepackage{subcaption}
\usepackage{cprotect}
\usepackage{etoc}
\usepackage{titlesec}
\usepackage{enumitem}
\usepackage{minitoc}

\usepackage{array}  

\usepackage{tikz}
\usetikzlibrary{arrows,automata, calc, positioning}
\usepackage{adjustbox}
\usepackage{wrapfig}

\usepackage{multirow}
\usepackage{accents}

\usepackage{epigraph}
\usepackage{etoolbox}

%% file: Z_commands.tex
\newcommand{\indep}{\raisebox{0.05em}{\rotatebox[origin=c]{90}{$\models$}}}

\newtheorem{proposition}{Proposition}[section]
\newtheorem{definition}{Definition}[section]
\newtheorem{remark}{Remark}[section]
\newtheorem{lemma}{Lemma}[section]
\newtheorem{theorem}{Theorem}[section]
\newtheorem{algorithm}{Algorithm}[section]
\newtheorem{corollary}{Corollary}[section]

\newtheorem{assumption}{Assumption}[section]
\newtheorem{example}{Example}[section]

\DeclareMathOperator*{\argmin}{\arg\!\min}

\DeclareMathOperator{\1}{\mathbbm 1}

\DeclareRobustCommand{\Ec}{\mathcal{E}}

\DeclareRobustCommand{\Hc}{\mathcal{H}}
\DeclareRobustCommand{\Rc}{\mathcal{R}}
\DeclareRobustCommand{\Wc}{\mathcal{W}}

\DeclareRobustCommand{\bE}{\boldsymbol{E}}
\DeclareRobustCommand{\bZ}{\boldsymbol{Z}}
\DeclareRobustCommand{\bA}{\boldsymbol{A}}
\DeclareRobustCommand{\bH}{\boldsymbol{H}}

\DeclareMathOperator{\bhrho}{\hat{\boldsymbol{\rho}}}

\DeclareRobustCommand{\bW}{\boldsymbol{W}}
\DeclareRobustCommand{\bX}{\boldsymbol{X}}
\DeclareRobustCommand{\bU}{\boldsymbol{U}}
\DeclareRobustCommand{\bS}{\boldsymbol{\Sigma}}
\DeclareRobustCommand{\bV}{\boldsymbol{V}}
\DeclareRobustCommand{\bhX}{\hat{\boldsymbol{X}}}
\DeclareRobustCommand{\bhG}{\hat{\boldsymbol{G}}}
\DeclareRobustCommand{\bhA}{\hat{\boldsymbol{A}}}
\DeclareRobustCommand{\bhU}{\hat{\boldsymbol{U}}}
\DeclareRobustCommand{\bhV}{\hat{\boldsymbol{V}}}
\DeclareRobustCommand{\bhS}{\hat{\boldsymbol{\Sigma}}}
\DeclareRobustCommand{\bB}{\boldsymbol{B}}

\DeclareRobustCommand{\bI}{\boldsymbol{I}}
\DeclareRobustCommand{\bG}{\boldsymbol{G}}
\DeclareRobustCommand{\bhM}{\hat{\boldsymbol{M}}}
\DeclareRobustCommand{\btM}{\tilde{\boldsymbol{M}}}

\newcommand{\tgamma}{\tilde{\gamma}}
\newcommand{\talpha}{\tilde{\alpha}}

\newcommand{\hbeta}{\hat{\beta}}
\newcommand{\hgamma}{\hat{\gamma}}
\newcommand{\halpha}{\hat{\alpha}}
\newcommand{\heta}{\hat{\eta}}
\newcommand{\htheta}{\hat{\theta}}
\newcommand{\hsigma}{\hat{\sigma}}
\newcommand{\hpsi}{\hat{\psi}}

\DeclareRobustCommand{\bhG}{\hat{\boldsymbol{G}}}
\DeclareRobustCommand{\bM}{\boldsymbol{M}}

\newcommand{\E}{\mathbb{E}}
\newcommand{\R}{\mathbb{R}}
\newcommand{\p}{\mathbb{P}}
\newcommand{\V}{\mathbb{V}}

\DeclareRobustCommand{\brho}{\boldsymbol{\rho}}
\DeclareRobustCommand{\norm}[1]{\left\| #1 \right\|}
\DeclareRobustCommand{\hrho}{\hat{\rho}}

\DeclareRobustCommand{\tbU}{\tilde{\boldsymbol{U}}}
\DeclareRobustCommand{\tbV}{\tilde{\boldsymbol{V}}}
\DeclareRobustCommand{\tbS}{\tilde{\boldsymbol{\Sigma}}}

\DeclareRobustCommand{\bVp}{\boldsymbol{V}^{\prime}}

\newcommand{\bgamma}{\boldsymbol{\gamma}}

\newcommand{\NA}{\texttt{NA}}
\newcommand{\tr}{\textsc{train}}
\newcommand{\err}{\textsc{ error}}
\newcommand{\te}{\textsc{test}}
\newcommand{\fil}{\textsc{fill}}
\newcommand{\treat}{\textsc{treat}}
\newcommand{\untreat}{\textsc{untreat}}
\newcommand{\tbeta}{\tilde{\beta}}

\newcommand{\homega}{\hat{\omega}}
\newcommand{\lr}{\textsc{(lr)}}
\newcommand{\row}{\textsc{row}}

\newcommand{\hs}{\hat{s}}

\newcommand{\bias}{\textsc{bias}}
\newcommand{\out}{\textsc{out}}
\newcommand{\poly}{\textsc{poly}}
\newcommand{\signal}{\textsc{signal}}
\newcommand{\noise}{\textsc{noise}}

\newenvironment{subsubfigure}[2][]{%
    \begin{subfigure}[#1]{#2}%
    \stepcounter{subsubfigure}%
}{%
    \addtocounter{subfigure}{-1}%
    \end{subfigure}%
}
\newcounter{subsubfigure}

\newcolumntype{L}[1]{>{\raggedright\let\newline\\\arraybackslash\hspace{0pt}}m{#1}}
\newcolumntype{C}[1]{>{\centering\let\newline\\\arraybackslash\hspace{0pt}}m{#1}}
\newcolumntype{R}[1]{>{\raggedleft\let\newline\\\arraybackslash\hspace{0pt}}m{#1}}

\let\originalleft\left
\let\originalright\right
\renewcommand{\left}{\mathopen{}\mathclose\bgroup\originalleft}
\renewcommand{\right}{\aftergroup\egroup\originalright}


%% file: 0_abstract.tex
\begin{abstract}
The US Census Bureau will deliberately corrupt data sets derived from the 2020 US Census, enhancing the privacy of respondents while potentially reducing the precision of economic analysis. To investigate whether this trade-off is inevitable, we formulate a semiparametric model of causal inference with high dimensional corrupted data. We propose a procedure for data cleaning, estimation, and inference with data cleaning-adjusted confidence intervals. We prove consistency and Gaussian approximation by finite sample arguments, with a rate of $n^{-1/2}$ for semiparametric estimands that degrades gracefully for nonparametric estimands. Our key assumption is that the true covariates are approximately low rank, which we interpret as approximate repeated measurements and empirically validate. Our analysis provides nonasymptotic theoretical contributions to matrix completion, statistical learning, and semiparametric statistics. Calibrated simulations verify the coverage of our data cleaning-adjusted confidence intervals and demonstrate the relevance of our results for Census-derived data.
\end{abstract}

%% file: 1_intro.tex
\section{Introduction}


The 2010 US Census inadvertently revealed too much information. In a simulated hack, researchers at the Census Bureau could re-identify between 52 and 179 million respondents from anonymous summary tables \cite{hawes2021}. To protect privacy, the Bureau will inject synthetic noise into summary tables of the 2020 Census 
and coarsen wage microdata in the Current Population Survey (CPS). 
Techniques like these, called privacy mechanisms in computer science, guarantee a property called \textit{differential privacy} via deliberate data corruption \cite{dwork2006calibrating}. Differential privacy is widely implemented in the technology sector, e.g. Apple iOS and Google Chrome data. Due to its recent adoption in the government sector, several economists have warned of a looming trade-off: the privacy of respondents versus the precision of economic analysis \cite{abowd2019economic,hotz2022balancing}. 

We study differential privacy and discretization as modern challenges for causal inference. Economic data continue to suffer from classical data corruptions such as missing values and measurement error. Therefore, we analyze a class of data corruptions that encompasses both modern and classical issues \textit{simultaneously}, while remaining agnostic about their relative magnitudes. Our research question is how (and when) it is possible to estimate typical causal parameters using high dimensional economic data that suffer from measurement error, missing values, discretization, and differential privacy mechanisms.
An answer requires nonasymptotic analysis because differential privacy is defined as a finite sample property.


We study a broad class of causal parameters, including semiparametric scalars such as the average treatment effect, the local average treatment effect, and the average elasticity, as well as nonparametric functions such as heterogeneous treatment effects, in a nonlinear and high dimensional setting. Our main contribution is a procedure for automatic data cleaning, causal estimation, and inference with confidence intervals that account for the bias and variance consequences of data cleaning. The procedure is simple. It essentially combines principal component analysis, ordinary least squares, and sample splitting in new ways.

Our key assumption is that the true covariates are approximately low rank, which we validate for US Census-derived data and interpret from a causal perspective. We argue that covariates collected from the Census include approximate repeated measurements---e.g. disability benefits, medical benefits, and unemployment benefits---which implies that they are approximately low rank. 
There are three key aspects of our contribution.

First, our simple procedure adapts to the \textit{type} and \textit{level} of data corruption. The same code works in a variety of settings, allowing for classical types such as measurement error and missing values as well as modern types such as discretization and differential privacy mechanisms, across variance levels. Crucially, the researcher does not need exact knowledge of the corruption distribution, e.g. its parametric form or covariance structure, and in this way we depart from the error-in-variable Lasso and Dantzig literatures; see Section~\ref{sec:related}.  
We depart from previous work on principal component regression 
by proposing new variants that involve ``implicit'' data cleaning---i.e. prediction on a test observation without cleaning it---and inference in nonlinear, heterogeneous causal models. We propose an error-in-variable balancing weight that adapts to the causal parameter of interest, which is a natural yet original solution for cross sectional data. Our theory of implicit data cleaning and our error-in-variable balancing weight appear to be new. The former is of independent interest. 

Second, our theoretical analysis allows the rate of data cleaning to be slower than the rate of causal inference, so an analyst can use matrix completion 
for automatic data cleaning of covariates. We extend the classic semiparametric framework, where the goal is to obtain $n^{-1/2}$ convergence for the causal parameter despite a slower rate of convergence for a nonparametric ``nuisance'' regression. Our goal is to obtain $n^{-1/2}$ convergence for the causal parameter despite a slower rate of convergence for high dimensional data cleaning, which is a ``nuisance'' task.
Since our data cleaning guarantees only hold on-average, we are unable to use previous semiparametric results; instead, we generalize semiparametric and nonparametric debiased machine learning theory to i.n.i.d. corrupted data, with new results on nominal and conservative variance estimation. 
Altogether, our framework translates slow, on-average data cleaning guarantees into fast causal estimation and inference guarantees. 

Third, our empirical results suggest that there exist scenarios in which the trade-off between privacy and precision can be overcome, and others in which it cannot. We replicate and extend \cite{autor2013china}'s seminal paper on the effect of import competition in US labor markets. To begin, we demonstrate the plausibility of our key assumption: Census data products contain many variables that are approximate repeated measurements. Next, we corrupt the data, injecting synthetic noise calibrated to the privacy level mandated for the 2020 US Census. We implement differential privacy and discretization in a way that belongs to our class of data corruptions, which can therefore be cleaned and adjusted for in the confidence interval. We recover the main results of \cite{autor2013china} without losing statistical precision. In this representative setting for economic research, it appears to be possible to achieve both privacy at the individual level and precision at the population level.

Section~\ref{sec:related} situates our contributions within the context of related work. Section~\ref{sec:model} formalizes our class of data corruptions and our key assumption. Section~\ref{sec:algo} proposes our procedure and demonstrates its performance in simulations. Section~\ref{sec:theory} theoretically justifies our procedure, and verifies the key assumption for nonlinear factor models. Section~\ref{sec:application} presents the semi-synthetic exercise and discusses limitations. Section~\ref{sec:conclusion} concludes. 

%% file: 2_related.tex
\section{Related work}\label{sec:related}

\textbf{Semiparametrics}.
We use two insights from classic \cite{hasminskii1979nonparametric,klaassen1987consistent,robinson1988root,bickel1993efficient,andrews1994asymptotics,newey1994asymptotic,robins1995semiparametric,ai2003efficient,van2006targeted,hahn2013asymptotic} and modern \cite{zheng2011cross,athey2018approximate,chernozhukov2018original,hirshberg2021augmented,chernozhukov2018learning,chernozhukov2021simple} semiparametric theory. First, a causal parameter typically has regression and balancing weight representations, and both appear in the semiparametrically efficient asymptotic variance. We directly build on this insight: an error-in-variable regression and an error-in-variable balancing weight appear in our data cleaning-adjusted confidence intervals. Second, sample splitting eliminates restrictive conditions on the data generating process and estimation procedure. We combine these two classic ideas with implicit data cleaning, which appears to be a new idea.

\textbf{Error-in-variable regression}. We provide a framework to repurpose error-in-variable regression estimators for downstream causal inference. Error-in-variable regression has a vast literature spanning econometrics, statistics, and computer science studying the model
\begin{equation}\label{eq:eiv}
    Y_i=\gamma_0(X_{i,\cdot})+\varepsilon_i,\quad 
    Z_{i,\cdot}=X_{i,\cdot}+H_{i,\cdot} \quad \text{($X_{i,\cdot}$ is the $i$-th row of matrix $\bX$ and so on)}
\end{equation}
where $(X_{i,\cdot},\varepsilon_i,H_{i,\cdot})$ are mutually independent and $(\varepsilon_i,H_{i,\cdot})$ are mean zero.  We consider a generalization of this setting with missingness, and we define our causal parameter as a scalar summary of nonlinear $\gamma_0$.
Methods in econometrics typically assume auxiliary information for identification: repeated measurements \cite{hausman1991identification,li1998nonparametric,schennach2004estimation}, instrumental variables \cite{schennach2007instrumental,hu2008instrumental}, and negative controls \cite{miao2018identifying,deaner2018nonparametric}. Similar in spirit to repeated measurements, we assume $\bX$ is approximately low rank.
%
Methods in statistics extend the Lasso and Dantzig selector to high dimensional error-in-variable regression \cite{loh2012high,rosenbaum2013improved,datta2017cocolasso}. However, these methods require linearity and exact sparsity of $\gamma_0$, as well as knowledge of the covariance of measurement error $H_{i,\cdot}$. By contrast, we assume $H_{i,\cdot}$ are subexponential; the analyst does not need to know the measurement error covariance, and therefore can be agnostic about the type and level of corruption. 
We propose new variants of principal component regression (PCR) for the error-in-variable regression and balancing weight. Previous work studies PCR for error-in-variable regression only, explicitly cleaning all observations \cite{stock2002forecasting,bai2006confidence,agarwal2020principal}. We develop a technique of implicit data cleaning that avoids mixing together signal and noise across observations, which aids with downstream statistical inference of nonlinear models. Moreover, our error-in-variable balancing weight for cross sectional data appears to be new.

\textbf{PCA for large factor models}. The initial step of PCR is PCA. A vast literature studies the identification, estimation, and inference of latent factors $(\lambda_i,\mu_j)$ in models of the form 
\begin{equation}\label{eq:factor}
    Z_{i,\cdot}=X_{i,\cdot}+H_{i,\cdot},\quad X_{ij}=\lambda_i^T\mu_j\quad \text{($X_{i,\cdot}$ is the $i$-th row of matrix $\bX$ and so on)}
\end{equation}
where $Z_{i,\cdot}$ is observed, the ambient dimension $dim(X_{i,\cdot})$ is high, and the latent dimension $dim(\lambda_i)$ is fixed \cite{bai2003inferential,bai2013principal}. Our interest in downstream causal inference allows us to bypass the issue of identifying latent factors, and to relax the linear factor model; instead, we require that the approximate rank of $\bX$ diverges more slowly than $dim(X_{i,\cdot})$. The nonlinear factor model $X_{ij}=g(\lambda_i,\mu_j)$, where $dim(\lambda_j)$ may slowly diverge and $g$ is smooth, is \textit{sufficient} but \textit{unnecessary} for our analysis. Like the factor model literature, we allow weak correlation and heteroscedastity of measurement errors within units. 

\textbf{Low rank causal models}. Whereas we study treatment effects, policy effects, and elasticities in cross sectional data, a rich literature studies treatment effects, in panel data, via a low rank factor model for potential outcomes \cite{athey2021matrix,bai2019matrix,xiong2019large,fernandez2020low, agarwal2020synthetic,feng2020causal}. 
By contrast, we study a more general class of causal parameters, in cross sectional data, when covariates are approximately low rank.
The only previous work to consider both measurement error and missingness in cross sectional treatment effects appears to be \cite{kallus2018causal}. The authors study average treatment effect and prove consistency, without inference, for a parametric linear model where the true covariates are low dimensional Gaussians and the measurement error distribution is correctly specified. 
By contrast, we study a broad class of semiparametric and nonparametric causal parameters and provide inference, with data cleaning-adjusted confidence intervals. We do not require exact distributional knowledge of (high dimensional) true covariates or measurement error.

\textbf{Privacy in econometrics}. Our research question complements others in a recent literature on private econometrics. One strand considers how to disclose estimates obtained with private data access \cite{dwork2009differential,smith2011privacy,komarova2020identification}.
 Another considers how to conduct estimation after linking public and private records, where privacy considerations constrain linkage \cite{komarova2018identification}. We ask a complementary question motivated by empirical economic research using public, Census-derived data products: how to conduct causal inference in the presence of simultaneous data corruptions, including canonical privacy mechanisms applied to the data before estimation.

%% file: 3_model.tex
\section{Model overview}\label{sec:model}


\textbf{Causal parameter}.
For readability, we focus on one causal parameter in the main text: the average treatment effect (ATE) with i.n.i.d. data 
$
\theta_0=\frac{1}{n}\sum_{i=1}^n\theta_i,$ where $\theta_i=\E[Y_i^{(1)}-Y_i^{(0)}].
$
Here, $Y_i^{(d)}$ is the potential outcome for unit $i$ under intervention $D=d$. $\theta_0$ is a sample average because different units may be drawn from different distributions---a challenging yet plausible scenario when data are corrupted. With i.i.d. data, $\theta_0$ simplifies to the familiar ATE. Appendix~\ref{sec:examples} considers a general class of semi- and nonparametric causal parameters e.g. the local average treatment effect, average elasticity, and heterogeneous treatment effects.

We denote the actual outcome by $Y_i\in \R$, the assigned treatment by $D_i\in\{0,1\}$, and the covariates that determine treatment assignment by $X_{i,\cdot}\in \R^p$. In order to express $\theta_0$ in terms of $(Y_i,D_i,X_{i,\cdot})$, we impose some additional structure on the problem. Generalizing a classic assumption in the literature on distribution shift, we assume that the conditional distributions $\p(Y_i|D_i,X_{i,\cdot})$ and $\p(D_i|X_{i,\cdot})$ are common across units; distribution shift is only in the marginal distributions of covariates $\p_i(X_{i,\cdot})$. 

Imposing these conditions as well as selection on $X_{i,\cdot}$, we recover two classic formulations of the treatment effect. The outcome formulation is in terms of the outcome mechanism $\gamma_0(D_i,X_{i,\cdot})=\E[Y_i|D_i,X_{i,\cdot}]$, also called the regression, which is common across units:
  $
        \theta_i=\E[\gamma_0(1,X_{i,\cdot})-\gamma_0(0,X_{i,\cdot})].
        $
The treatment formulation is in terms of the treatment mechanism $\E[D_i|X_{i,\cdot}]$, which is also common across units, and which appears in the denominator of the balancing weight $\alpha_0(D_i,X_{i,\cdot})=\frac{D_i}{\E[D_i|X_{i,\cdot}]}-\frac{1-D_i}{1-\E[D_i|X_{i,\cdot}]}$: here, 
        $
        \theta_i=\E[Y_i \cdot  \alpha_0(D_i,X_{i,\cdot})].
        $
Our estimation and analysis combine both classic formulations.

\textbf{Data corruption}. The crux of our problem is that we observe $(Y_i,D_i,Z_{i,\cdot})$ instead:
 \begin{equation}\label{eq:corruption}
     Y_i=\gamma_0(D_i,X_{i,\cdot})+\varepsilon_i,\quad 
      Z_{i,\cdot}=(X_{i,\cdot}+H_{i,\cdot})\odot \pi_{i,\cdot}.
 \end{equation}
 Though the outcome $Y_i$ is generated from treatment $D_i$ and true covariates $X_{i,\cdot}$, we do not observe $X_{i,\cdot}$; instead, we observe the corrupted covariates $Z_{i,\cdot}$, which are the true covariates $X_{i,\cdot}$ plus conditionally mean zero corruption $H_{i,\cdot}$, multiplied elementwise by an independent masking vector $\pi_{i,\cdot}\in\{\NA,1\}^p$. Our concise model~\eqref{eq:corruption} generalizes the models~\eqref{eq:eiv} and~\eqref{eq:factor}, and it encompasses all four types of corruption.
For example, to encode classical measurement error, let $Z_{i,\cdot}$ equal $X_{i,\cdot}$ plus a vector of Gaussian noise. To encode missing values, let $Z_{i,\cdot}=X_{i,\cdot} \odot \pi_{i,\cdot}$. In Appendix~\ref{sec:examples}, we accommodate corruption of the outcome $Y_i$ and treatment $D_i$, under restrictions.
 
 Discretization is a process by which a continuous vector $X_{i,\cdot}$ maps to a discrete vector  $Z_{i,\cdot}$, and our class encodes variants where $\E[Z_{i,\cdot}|X_{i,\cdot}]=X_{i,\cdot}$. For example, the covariate of interest may be a vector of probabilities $X_{i,\cdot}$, yet we observe actual occurrences $Z_{i,\cdot} \sim \text{Bernoulli}(X_{i,\cdot})$. Another example is randomized rounding, where continuous values are randomly rounded to nearby integers, e.g.  $
Z_{i,\cdot}=sign(X_{i,\cdot})\text{Poisson}(|X_{i,\cdot}|)$. Our class does not include deterministic rounding. Instead, it provides guidance on which types of rounding can be handled well in downstream causal inference. As such, it suggests alternative types of discretization for wage data in the CPS which are more favorable for economic research.

Differential privacy is a concept from computer science which means plausible deniability that any individual contributed their data to tabular summaries. The canonical mechanism that ensures differential privacy is to release $Z_{i,\cdot}$ equal to $X_{i,\cdot}$ plus a vector of Laplacian noise, calibrating the variance of the Laplacian to a priori bounds on the true values and other properties of the tabular summary statistics \cite{dwork2006calibrating}. Our framework allows for other canonical privacy mechanisms where $\E[Z_{i,\cdot}|X_{i,\cdot}]=X_{i,\cdot}$, e.g. discrete Gaussian, piece wise uniform, and bounded mechanisms. In the context of the Census, we consider adding Laplacian noise to data on aggregate units, which we formalize in Section~\ref{sec:application}. Injecting synthetic noise in this way helps to prevent the kind of attack simulated on the 2010 Census.

Across examples, $H_{i,\cdot}$ is subexponential, i.e. its tails are no worse than an exponential distribution's. So are compositions of various types of data corruption since the class of subexponential distributions is closed under addition. Therefore our class of data corruptions includes classical and modern issues \textit{simultaneously}. It allows us to address the trade-off between privacy and precision in the context of heteroscedastic measurement error---a major aspect of the problem often overlooked \cite{chetty2019practical,steed2022policy}.

What corruptions do we exclude? Our definition of $H_{i,\cdot}$ rules out nonseparable or endogeneous measurement error. Our definition of $\pi_{i,\cdot}$ rules out endogenous missingness, i.e. sample selection. We also rule out deterministic rounding, i.e. interval censoring. However, our class includes canonical privacy mechanisms as well as randomized rounding, itself a privacy mechanism. Therefore we study this class, which we formalize in Section~\ref{sec:theory}. Importantly, we allow heteroscedastic corruptions that are dependent within a unit.

\textbf{Key assumption: Approximate repeated measurements}. Our key assumption is that the true covariates are approximately low rank: the rank of the matrix $\bX \in\R^{n\times p}$ is approximately $r \ll (n,p)$. 
Among the $p$ covariates in the data set, there are approximately only $r$ latent types of covariates. For intuition, consider repeated measurements. In the classic repeated measurement model, we have two noisy measurements of one signal. In our model, we have $p$ noisy measurements $(Z_{i1},...,Z_{ip})$ that are approximately repeated measurements of only $r$ signals, where both $(r,p)$ grow with sample size $n$, yet $r\ll (n,p)$.


We place this assumption because it seems plausible in Census-derived data. Consider the commuting zone (CZ) level data set of \cite{autor2013china}. Each CZ is a local economy with a vector of covariates $X_{i,\cdot}\in \R^{30}$ if we use variables from the authors' preferred specification as well as additional variables from their appendix. The variables include average disability, unemployment, and medical benefits, which are not precisely repeated measurements but approximately so. We compute the singular value decomposition of $\bX$ then visualize its singular values, also called its principal components, in Figure~\ref{fig:scree2}. We see that only about five principal components are significantly positive; $r=5$ while $p=30$.

\begin{wrapfigure}{r}{0.5\textwidth}
    \vspace{-30pt}
    \begin{center}
        \includegraphics[width=0.48\textwidth]{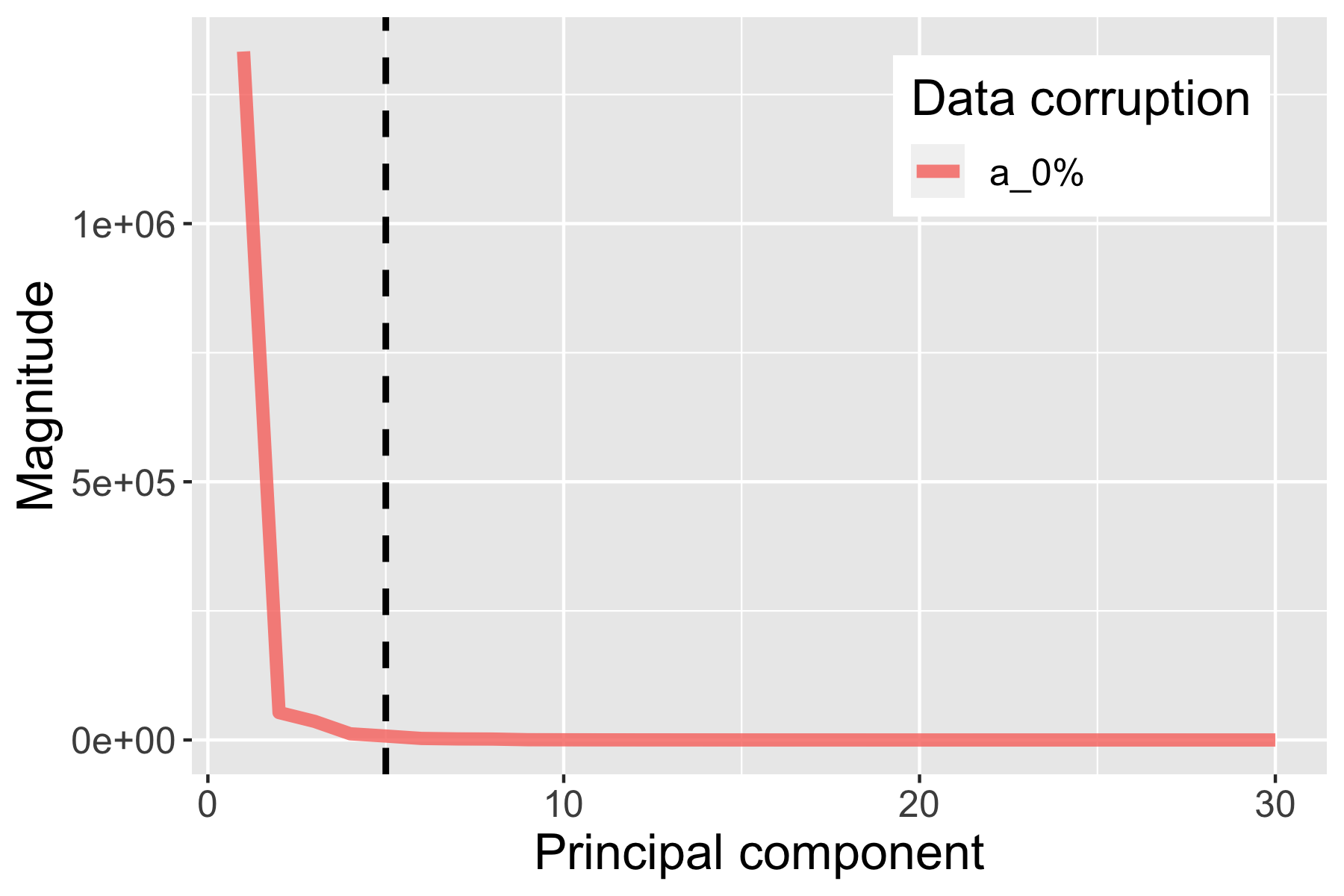}
    \end{center}
    \vspace{-26pt}
    \caption{\label{fig:scree2} Key assumption in Census data}
    \vspace{-20pt}
\end{wrapfigure}

Our key assumption admits a causal interpretation in the running example of ATE.
Consider the special case where the true covariates are exactly low rank, i.e. $r=rank(\bX)$. The singular value decomposition is $\bX=\bU\bS\bV^T$ where $\bU\in \R^{n\times r}$, $\bS\in \R^{r\times r}$, and $\bV\in \R^{p\times r}$. $\bV$ consists of $r$ vectors in $\R^p$, called the right singular vectors of $\bX$, which are also the eigenvectors of the empirical covariance $n^{-1}\bX^T \bX$. The span of these vectors is an $r$ dimensional subspace of $\R^p$, i.e. a low dimensional subset of a high dimensional ambient space. In this scenario, we assume that treatment assignment is determined by the subspace. More generally, when covariates are approximately low rank, $\bX=\bX^{\lr}+\bE^{\lr}$, where $\bX^{\lr}=\bU\bS\bV^T$ is a rank $r$ approximation to $\bX$, and $\bE^{\lr}$ is the approximation residual. We can either assume (i) selection is determined by $\bX^{\lr}$ only, i.e. the treatment assignment for unit $i$ depends on the \textit{projection} of $X_{i,\cdot}$ onto the subspace spanned by $\bV$; or (ii) selection is determined by both $\bX^{\lr}$ and $\bE^{\lr}$. To handle the latter, we keep track of $\Delta_E=\|\bE^{\lr}\|_{\max}$ in our theoretical analysis. Our analysis is robust to small violations of the exactly low rank assumption from statistical and causal perspectives.

%% file: 4_estimation.tex
\section{Data cleaning-adjusted confidence interval}\label{sec:algo}

We would like a procedure that estimates parameters in nonlinear, heterogeneous causal models as if data were uncorrupted, yet adjusts for data cleaning in the confidence interval. Moreover, we would like a procedure that does not require knowledge of the corruption covariance structure in advance, departing from previous work. If such a procedure were to exist, it would in some sense preempt the looming trade-off between privacy and precision.

\begin{wrapfigure}{r}{0.5\textwidth}
    \vspace{-30pt}
    \begin{center}
        \includegraphics[width=0.48\textwidth]{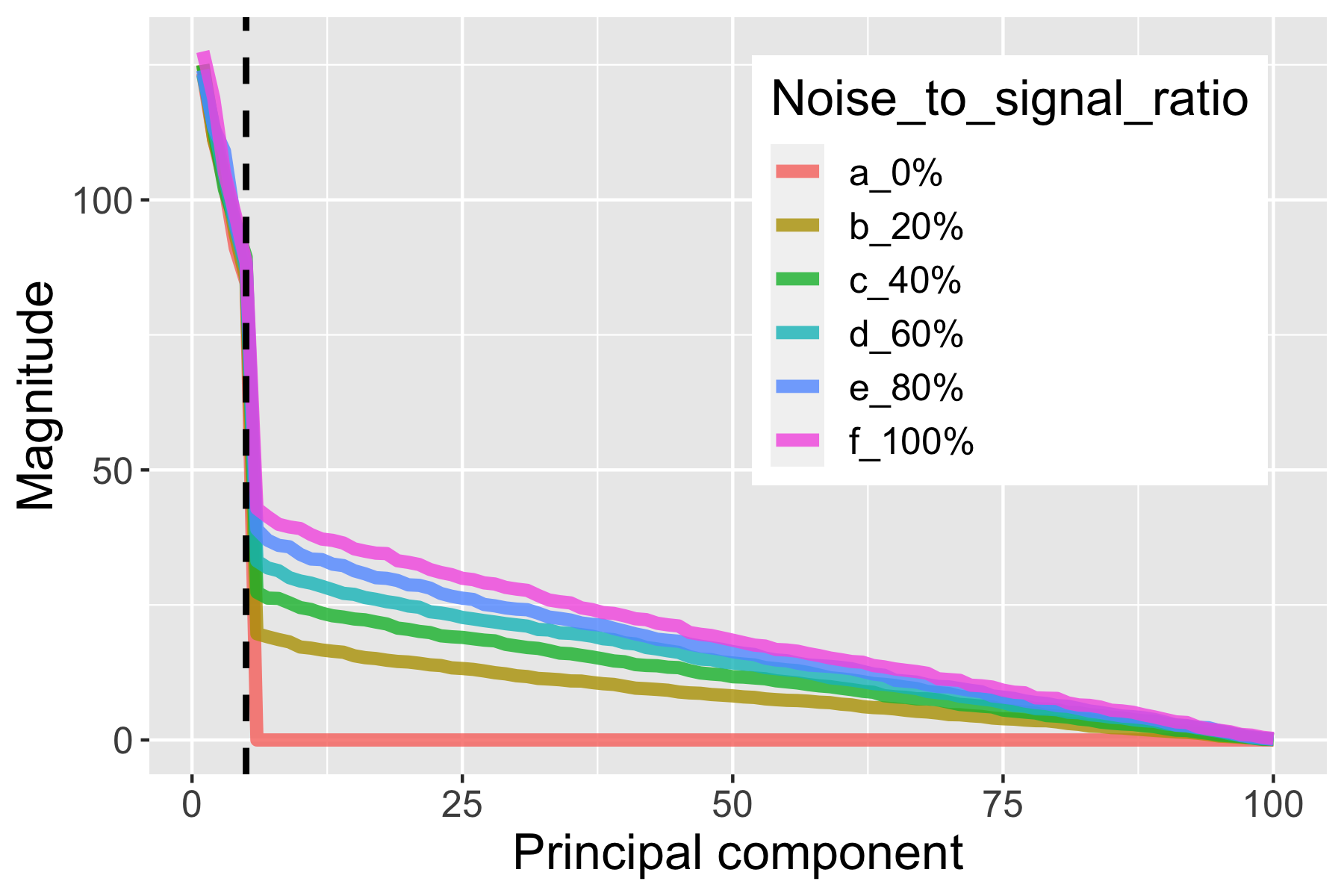}
    \end{center}
    \vspace{-26pt}
    \caption{\label{fig:spectrum} Key assumption in simulated data}
    \vspace{-20pt}
\end{wrapfigure}

\textbf{Why is inference hard}?
 We illustrate our procedure with an average treatment effect simulation. By construction, the treatment effect is $\theta_0=2.2$. We consider a data generating process (DGP) which satisfies our key assumption: one sample involves a matrix of covariates $\bX\in\R^{100\times 100}$ with rank $r=5$. See Appendix~\ref{sec:sim} for details and for similar results using alternative dimensions of $\bX$. The DGP has nonlinear outcome and treatment mechanisms. Figure~\ref{fig:spectrum} plots the principal components of true covariates $\bX$ in red. As expected, five principal components are nonzero and the rest are zero since $rank(\bX)=5$.

As a first pass, we implement ordinary least squares (OLS) of $Y_i$ on $(D_i,X_{i,\cdot})$. Running OLS on clean data 1000 times, the point estimates $\htheta$ (Figure~\ref{fig:OLS_est}) center around the true value $2.2$, and appear Gaussian. 
OLS works well in the absence of data corruption; there is nothing hidden in the DGP for clean data. We repeat this exercise introducing measurement error with variance that is 20\% of the variance of the covariates. Inversion of the empirical covariance matrix $n^{-1}\bZ^T\bZ$ becomes numerically unstable, manifesting in point estimates that are erratic (Figure~\ref{fig:OLS_corrupted}) and standard errors that are typically \NA's. Notably, data corruption flips the sign about a quarter of the time, a phenomenon we verify for 2SLS and for settings closer to \cite{autor2013china} in Appendix~\ref{sec:sim}. OLS is not well-suited to the combination of high dimensional covariates, (approximate) low rank, and measurement error. Indeed, any estimator that ignores covariate measurement error in a nonlinear, heterogeneous causal model suffers from bias of a complicated form \cite{battistin2014treatment}. 

\begin{figure}[ht]
\vspace{-20pt}
\begin{centering}
     \begin{subfigure}[b]{0.48\textwidth}
         \centering
         \includegraphics[width=\textwidth]{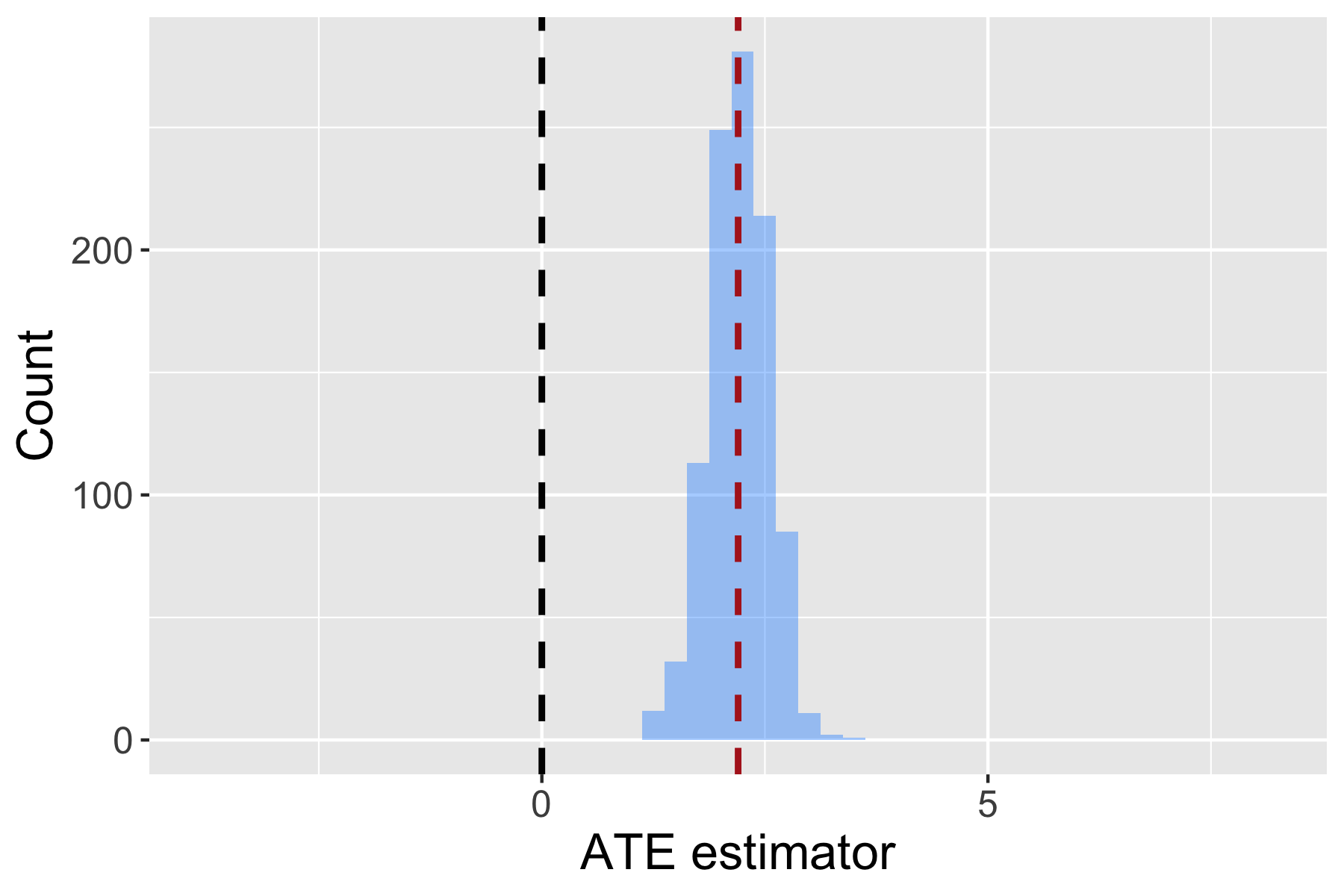}
         \vspace{-26pt}
         \caption{\label{fig:OLS_est} OLS succeeds on clean data}
     \end{subfigure}
     \hfill
     \begin{subfigure}[b]{0.48\textwidth}
         \centering
         \includegraphics[width=\textwidth]{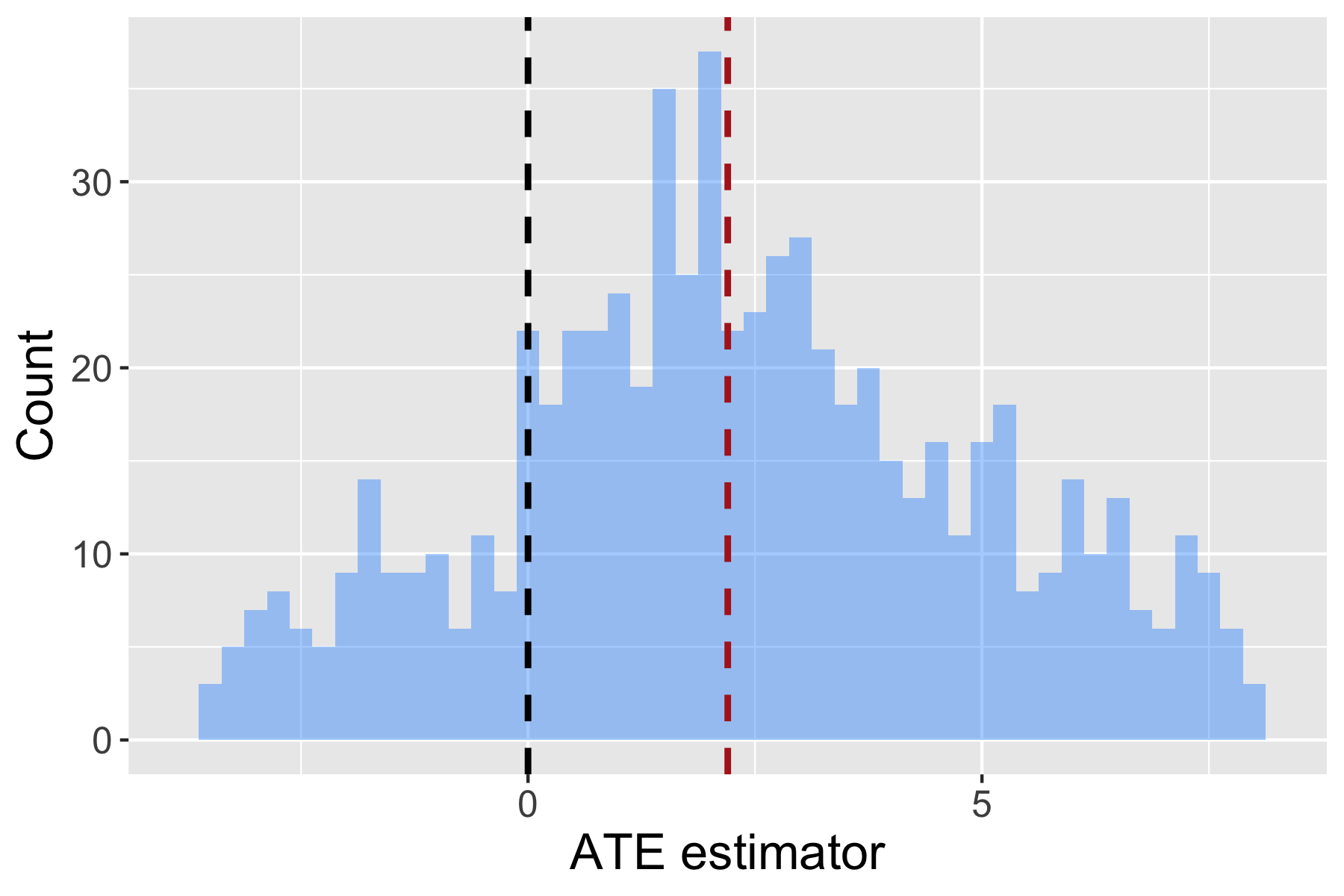}
         \vspace{-26pt}
         \caption{\label{fig:OLS_corrupted} OLS fails on corrupted data}
     \end{subfigure}
\par
\vspace{-10pt}
\caption{First pass --- OLS}\label{fig:OLS}
\vspace{-15pt}
\end{centering}
\end{figure}

Data corruption can derail causal inference, which motivates filling the \NA's, reigning in the extremes, and otherwise de-noising the values in $\bZ$ in hopes of recovering $\bX$. These are precisely the goals of matrix completion applied to the matrix $\bZ$ \cite{candes2009exact,candes2010power,keshavan2009matrix,hastie2015matrix,chatterjee2015matrix} . Our goal is to automate data cleaning via matrix completion, then to adjust for data cleaning in the confidence interval. To select an appropriate matrix completion method, we return to Figure~\ref{fig:spectrum} to visualize the principal components of the corrupted covariates $\bZ=\bX+\bH$ for various noise-to-signal ratios (i.e. the noise variance divided by the signal variance). The initial five principal components remain virtually unchanged, while the lower principal components are amplified; signal remains spectrally concentrated while noise is spectrally diffuse. Therefore a natural way to clean the covariates is to discard the lower principal components---in essence, to perform principal component analysis (PCA).\footnote{Alternative choices include canonical correlation analysis and partial least squares, which clean $\bZ$ using $Y$. We leave these directions to future research.}

Why is inference hard after data cleaning? Several challenges arise. First, data cleaning may mix together signal and noise across observations; yet, we wish to prove Gaussian approximation via a central limit theorem. Our solution is to break dependence via both sample splitting, which is a classic idea \cite{klaassen1987consistent}, and implicit data cleaning, which is a new idea. Second, if we turn to automated data cleaning, the best rates of convergence to the true matrix $\bX$ are slower than $n^{-1/2}$; yet, we wish to obtain a standard error of order $\hsigma n^{-1/2}$ for $\theta_0$. Our solution is to use a doubly robust estimating equation and to generalize double rate robustness \cite{chernozhukov2018original,van2018targeted,rotnitzky2021characterization}. The third issue is a theoretical one to which we will return in Section~\ref{sec:theory}: the best rates of matrix completion are not for recovering specific matrix entries but rather averages across matrix entries; yet, we wish to obtain downstream semiparametric inference. Our solution is to develop an algorithmic and analytic framework that forges a connection.

\textbf{Overview of the procedure}. Split the observations $(Y_i,D_i,Z_{i,\cdot})$ into equally sized \tr\, and \te\, sets, each with $m=n/2$ observations. Our procedure consists of four steps, which we state at a high level before filling in the details: (i) data cleaning: $\bhX$ using \tr; (ii) error-in-variable regression: $\hgamma$ using \tr; (iii) error-in-variable balancing: $\halpha$ using \tr; (iv) causal parameter: $\htheta \pm 1.96 \hsigma n^{-1/2}$ using \te. We opt for simplicity at each step, essentially combining PCA and OLS (albeit in new ways). We view these high level steps a template for more complex procedures in future work.

\textit{Step 1: Data cleaning}. The automated data cleaning procedure is extremely simple: fill in missing values as zeros, scale appropriately, then perform PCA. 

For any mathematical operation to be well defined, the $\NA$'s must be filled in somehow. To begin, we tally the likelihood of non-missingness for each covariate $j\in[p]$ in \tr:
$
\hrho_j=\max\left\{ \frac{1}{m} \sum_{i\in \tr} \1(Z_{ij} \neq \NA ), \frac{1}{m}\right\},$ and $\bhrho=diag(\hrho_1,...,\hrho_p)\in\R^{p\times p}.
$
Next, we fill in missing values with a \fil\, operator defined such that $\fil(Z_{ij})=\frac{Z_{ij}}{\hrho_j}$ if $Z_{ij} \neq \NA$ and $\fil(Z_{ij})=0$ if $ Z_{ij}= \NA$.
Let $\bZ^{\tr}$ be rows of $\bZ$ where $i\in \tr$. The \fil\, operator may act on $\bZ^{\tr}$ or $\bZ^{\te}$, but it always uses the likelihoods $\bhrho$ calculated from  $\bZ^{\tr}$.

\begin{proposition}[Filling with zeros is unbiased and simple]\label{prop:fill}
For $i\in\te$,
$$
\E[\fil(Z_{ij})|X_{ij},\tr]=X_{ij}\frac{\rho_j}{\hrho_j}.
$$
The alternative procedure of filling missing values with averages from \tr, denoted by $\bar{Z}_j^{\tr}$, gives  
$$
\E[\textsc{fill-as-means}(Z_{ij})|X_{ij},\tr]=X_{ij}\rho_j+ \bar{Z}_j^{\tr}(1-\rho_j).
$$
\end{proposition}
\textsc{fill-as-means} gives a convex combination of the signal $X_{ij}$ and of the noisy average $\bar{Z}_j^{\tr}$. The noisy average introduces additional correlations that our procedure avoids.

After filling \tr, we project it onto its own principal subspace to calculate the cleaned training covariates $\bhX$:
$
\fil(\bZ^{\tr})=\bhU\bhS\bhV^T,$ and $\bhX=\bhU_k\bhS_k\bhV_k^T.
$
We truncate the SVD of $\fil(\bZ^{\tr})$ to include only the top $k$ principal components, where $k$ is a hyperparameter. 
Figure~\ref{fig:spectrum} suggests a choice of $k$. 
Below, we empirically verify that our results are robust to different choices of $k>r$. Future work may derive a data driven procedure $k=\hat{r}$ \cite{stock1998diffusion,bai2002determining,onatski2009formal}.
We preserve the ambient dimension $p$.

\textit{Step 2: Error-in-variable regression}. Our error-in-variable regression is also simple: after cleaning \tr, perform ordinary least squares (OLS) on \tr, then use this OLS coefficient on the filled \te\, for prediction. We only fill, and do not clean, the test set. 

We introduce nonlinearity into the regression to allow treatment effect heterogeneity, which is crucial for causal inference. Appendix~\ref{sec:dictionary} characterizes what nonlinearity is allowed. Here, we focus on the interacted dictionary
$
b(D_i,\hat{X}_{i,\cdot})=(D_i \hat{X}_{i,\cdot},(1-D_i) \hat{X}_{i,\cdot} ).
$ 
Then the OLS coefficient
$
\hbeta=\{(\hbeta^{\treat})^T, (\hbeta^{\untreat})^T\}^T$ equals $[\{b(D^{\tr},\bhX)\}^Tb(D^{\tr},\bhX)]^{\dagger}$ $[\{b(D^{\tr},\bhX)\}^T Y^{\tr}],
$
where $\dagger$ means pseudoinverse.

The subtlety is in how predictions are constructed from $\hbeta$. Out of sample prediction does \textit{not} involve cleaning the test set: for $i\in \te$,
$
\hgamma(D_i,Z_{i,\cdot})=b\{D_{i},\fil(Z_{i,\cdot})\}\hbeta. 
$

\begin{proposition}[Implicit data cleaning preserves independence]\label{prop:indep}
For $i\in \te$,
$
\hgamma(D_i,Z_{i,\cdot})=b(D_i,Z_{i,\cdot})\tbeta, $
where
$ \tbeta=\{(\bhrho^{-1}\hbeta^{\treat})^T, (\bhrho^{-1}\hbeta^{\untreat})^T\}^T$
and we replace $\NA$ with $0$ in $Z_{i,\cdot}$. Therefore for $(i,j)\in \te$,
$
\hgamma(D_i,Z_{i,\cdot}) \indep \hgamma(D_j,Z_{j,\cdot}) | \tr.
$
\end{proposition}
Remarkably, post-multiplying $b(D_i,Z_{i,\cdot})$ by $\tbeta$ handles the measurement error, missingness, discretization, and differential privacy of $Z_{i,\cdot}$ while also producing high quality nonlinear predictions of $Y_i$. We call this phenomenon ``implicit'' data cleaning. Moreover, since $\tbeta$ is learned exclusively from $\tr$, it is deterministic conditional on $\tr$, so predictions for observations $(i,j)\in\te$ preserve their independence. This property of implicit data cleaning will be essential for our inferential theory.

Our new variant of PCR has broader use outside of causal inference. In online learning, a corrupted test observation $Z_{i,\cdot}$ does not need to be explicitly cleaned with respect to $\te$ or even $\tr$. Instead, it may be implicitly cleaned by post multiplying it with the coefficient $\tbeta$. For test observations, data cleaning and prediction can be combined into one step. 

\textit{Step 3: Error-in-variable balancing}. Our error-in-variable balancing weight generalizes our error-in-variable regression. It avoids the estimation and inversion of propensity scores, which may be numerically unstable in high dimensions. Pleasingly, it achieves exact balance for any finite sample size, in a sense that we formalize below. Moreover, it adapts to the causal parameter of interest, as we explain in Appendix~\ref{sec:riesz_proof}.

The only difference from the error-in-variable regression is that we replace the sufficient statistic $[\{b(D^{\tr},\bhX)\}^T Y^{\tr}] \in \R^{p'}$ with another sufficient statistic that we call the counterfactual moment $\hat{M}\in \R^{p'}$. The counterfactual moment resembles the expression $\theta_i=\E[\gamma_0(1,X_{i,\cdot})-\gamma_0(0,X_{i,\cdot})]$, and it is the \textit{only} aspect of the algorithm that changes for different causal parameters. Formally,
$
\heta=[\{b(D^{\tr},\bhX)\}^Tb(D^{\tr},\bhX)]^{\dagger}\hat{M}$ and $\hat{M}= [\{b(1,\bhX)\}^T-\{b(0,\bhX)\}^T] \1_m
$
where $\1_m \in\R^m$ is a vector of ones. As before, we do not clean the test set: for $i\in \te$,
$
\halpha(D_i,Z_{i,\cdot})=b\{D_{i},\fil(Z_{i,\cdot})\}\heta.
$

\begin{proposition}[The balancing weight exactly balances covariates]\label{prop:balance_ATE}
For any finite sample,
$$
\frac{1}{m}\sum_{i\in\tr}\hat{X}_{i,\cdot}=\frac{1}{m}\sum_{i\in\tr}D_i\hat{X}_{i,\cdot}\homega_i^{\tr}=\frac{1}{m}\sum_{i\in\tr}(1-D_i)\hat{X}_{i,\cdot} \homega_i^{\untreat},
$$
where $(\homega_i^{\treat},\homega_i^{\untreat})\in\R$ are balancing weights computed from $\heta=\{(\heta^{\treat})^T, (\heta^{\untreat})^T\}^T$ as
$
\homega_i^{\treat}=\hat{X}_{i,\cdot}\heta^{\treat}$ and $\homega_i^{\untreat}=-\hat{X}_{i,\cdot}\heta^{\untreat}.
$
\end{proposition}
Deterministically, the error-in-variable balancing weight exactly balances the full population, the treated subpopulation, and the untreated subpopulation with respect to their cleaned covariates. It is precisely the reweighting that would ensure comparability of treated and untreated groups in a stratified sampling design. We articulate a more general balancing property for generic causal parameters in Appendix~\ref{sec:riesz_proof}. We also clarify the sense in which the error-in-variable regression and balancing weight coincide on $\tr$ but not $\te$.

\textit{Step 4: Causal estimation and inference}. The final step uses the error-in-variable regression $\hgamma$ and error-in-variable balancing weight $\halpha$ learned from $\tr$, and evaluates them on $\te$\, according to the doubly robust estimating equation: for $i\in \te$,
$
\hpsi_i
=\hgamma(1,Z_{i,\cdot})-\hgamma(1,Z_{i,\cdot})+\halpha(D_i,Z_{i,\cdot})\{Y_i-\hgamma(D_i,Z_{i,\cdot})\}
$ is the empirical influence of that observation.
This process generates a vector $\hpsi\in\R^m$. Reversing the roles of $\tr$ and $\te$, we generate another such vector. Slightly abusing notation, we concatenate the two to obtain a vector $\hpsi \in \R^n$. We estimate the causal parameter as $\htheta=\textsc{mean}(\hpsi)$, its variance as $\hsigma^2=\textsc{var}(\hpsi)$, and its data cleaning-adjusted confidence interval as
$\textsc{CI}=\htheta \pm 1.96 \hsigma n^{-1/2}.
$

Our procedure deals with measurement error bias by cleaning the data. For the special case of ATE, the measurement error bias has a closed form solution in terms of the regression, propensity score, covariate density, and derivatives thereof \cite{battistin2014treatment}. We avoid estimation of the propensity scores, covariate density, and derivatives, which would be challenging in high dimensions. Instead, we simply combine PCA and OLS.

The way we impute missing values modifies multiple imputation \cite{rubin1976inference}. 
In multiple imputation, the analyst makes, say, two copies of the original data set, then imputes missing values (with some randomness so each imputation may be different). Estimates and standard errors from each copy are then averaged. Our procedure splits the sample into two folds: \tr\, and \te. We clean \tr\, and compute estimates and standard errors with \te, then reverse the roles and take the average. We opt for sample splitting, rather than copying, and we additionally consider measurement error.

\textbf{Adapting to the type and level of corruption}.
Next, we demonstrate that our four step procedure performs well in simulations with a broad variety of data corruptions. We run the same code in every setting; the procedure adapts to the \textit{type} and \textit{level} of data corruption, without prior knowledge of the corruption covariance structure.

To begin, we consider measurement error $Z_{i,\cdot}=X_{i,\cdot}+H_{i,\cdot}$, where $H_{i,\cdot}$ is Gaussian noise, in the average treatment effect simulation described above. Recall that $\theta_0=2.2$, $\bX\in\R^{100\times 100}$, and $r=5$. We implement our procedure on corrupted data 1000 times, collecting 1000 point estimates $\htheta$ and 1000 standard errors $\hsigma$. For a 20\% noise-to-signal ratio, we visualize the studentized point estimates $(\htheta-\theta_0)/\hsigma$ in Figure~\ref{fig:PCR_noise}. Qualitatively, the histogram closely resembles the standard normal density.

We quantify performance in coverage tables. In Table~\ref{fig:PCR_noise_table}, different rows correspond to different noise-to-signal ratios. Initially, we consider the oracle tuning of the PCA hyperparameter $k=r$. For each noise-to-signal ratio, we record the average point estimates, which are close to $\theta_0=2.2$. Next, we record the average standard errors, which adaptively increase in length to higher noise levels. Impressively, a 100\% noise-to-signal ratio setting corresponds to a confidence interval that is only about 10\% longer. These confidence intervals are the correct length, since about 950 of them include the true value $\theta_0=2.2$.

Table~\ref{fig:PCR_noise_table} revisits the issue of tuning the hyperparameter $k$. This time, we fix the noise-to-signal ratio to 20\%. Different rows correspond to different tunings: $k=r$, $k=r+2$, and $k=r+5$. Point estimates remain close to the true value $\theta_0=2.2$. The standard errors adaptively increase in length when $k$ deviates from $r$, though the length only increases about 10\%. The confidence intervals are again the correct length, attaining nominal coverage.

We repeat this exercise with other types of data corruption: missing values (Figure~\ref{fig:PCR_missing}), discretization (Figure~\ref{fig:PCR_discrete}), and differential privacy (Figure~\ref{fig:PCR_private}). For missing values, $Z_{i,\cdot}=X_{i,\cdot}\odot \pi_{i,\cdot}$ and we consider non-response of 10\%, 30\%, and 50\% of \textit{all covariate entries}. In Census Bureau surveys, key variables such as income are missing 40\% of the time. Fortunately, our procedure performs well even with this high level of missingness. For discretization, we consider randomized rounding $Z_{i,\cdot}=sign(X_{i,\cdot})\text{Poisson}(|X_{i,\cdot}|)$, which corresponds to a 33\% noise-to-signal ratio. Finally, for differential privacy, $Z_{i,\cdot}=X_{i,\cdot}+H_{i,\cdot}$ where $H_{i,\cdot}$ is Laplacian noise, and we obtain results that are nearly identical to measurement error. Across settings, our results are robust to hyperparameter tuning. 

\begin{figure}[H]
\vspace{-5pt}
\begin{centering}
\begin{subfigure}[c]{\textwidth}
    \setcounter{subfigure}{0}
    \setcounter{subsubfigure}{0}
    \begin{subsubfigure}[b]{0.48\textwidth}
        \centering
        \includegraphics[width=0.8\textwidth]{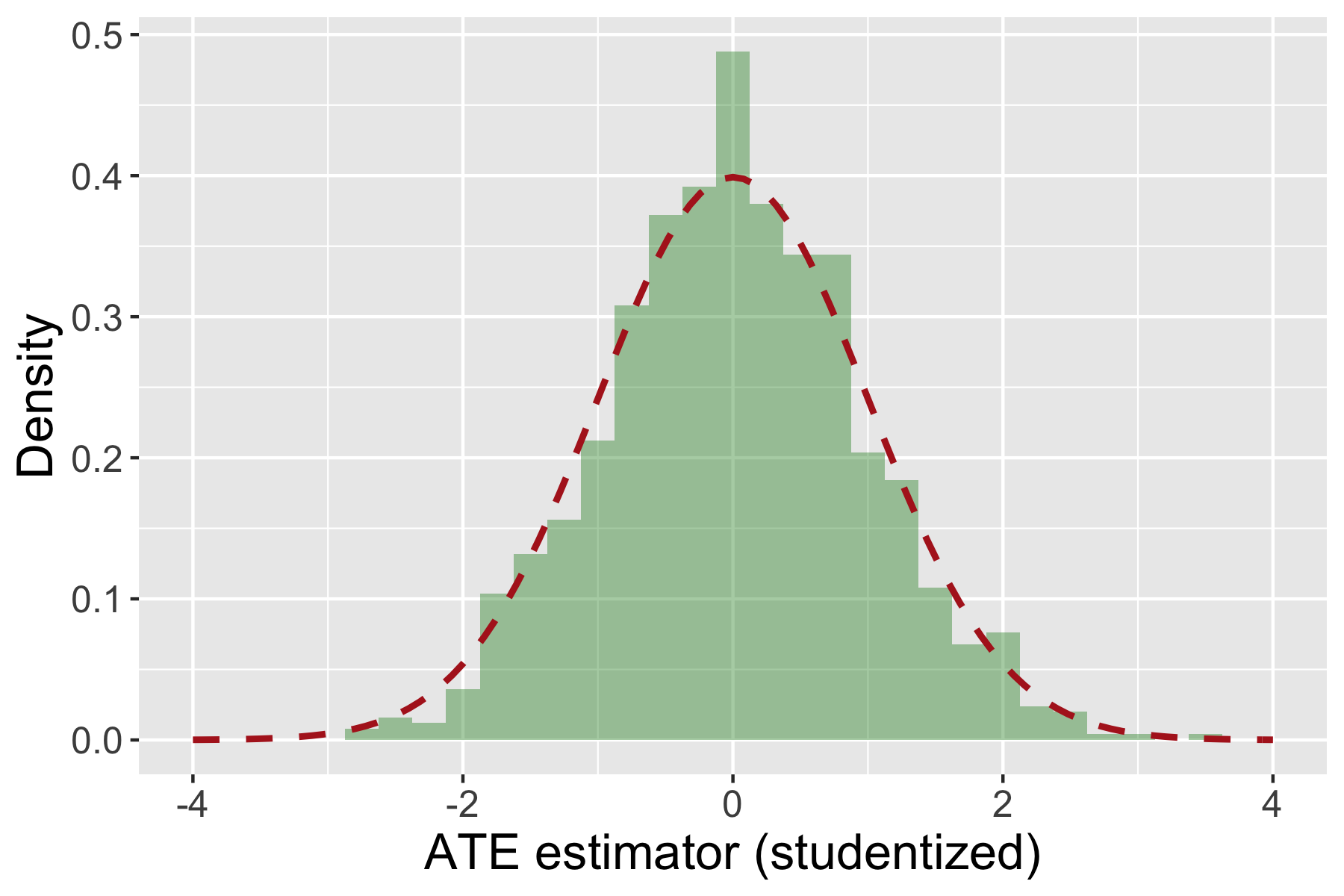}
        \vspace{-3pt}
        \caption{Measurement error inference}\label{fig:PCR_noise}
    \end{subsubfigure}
    \begin{subsubfigure}[b]{0.44\textwidth}
        \centering
        \resizebox{\textwidth}{!}{%
        \begin{tabular}{C{11ex} C{6ex} C{5ex} C{5ex} C{8ex} C{8ex}}
            \hline
            Meas. Err. & PCA & ATE & SE & 80\% CI & 95\% CI\tabularnewline
            \hline
            20\% & k=5 & 2.22 & 0.35 & 0.81 & 0.96\tabularnewline
            60\% & k=5 & 2.23 & 0.37 & 0.81 & 0.96\tabularnewline
            100\% & k=5 & 2.28 & 0.39 & 0.82 & 0.95\tabularnewline
            \hline
            \hline
            20\% & k=5 & 2.22 & 0.35 & 0.81 & 0.96\tabularnewline
            20\% & k=7 & 2.21 & 0.36 & 0.84 & 0.96\tabularnewline
            20\% & k=10 & 2.22 & 0.39 & 0.83 & 0.97\tabularnewline
            \hline
        \end{tabular}
        }
        \vspace{-3pt}
        \caption{Measurement error coverage}\label{fig:PCR_noise_table}
    \end{subsubfigure}
\end{subfigure}

\begin{subfigure}[c]{\textwidth}
    \setcounter{subfigure}{1}
    \setcounter{subsubfigure}{0}
    \begin{subsubfigure}[b]{0.48\textwidth}
        \centering
        \includegraphics[width=0.8\textwidth]{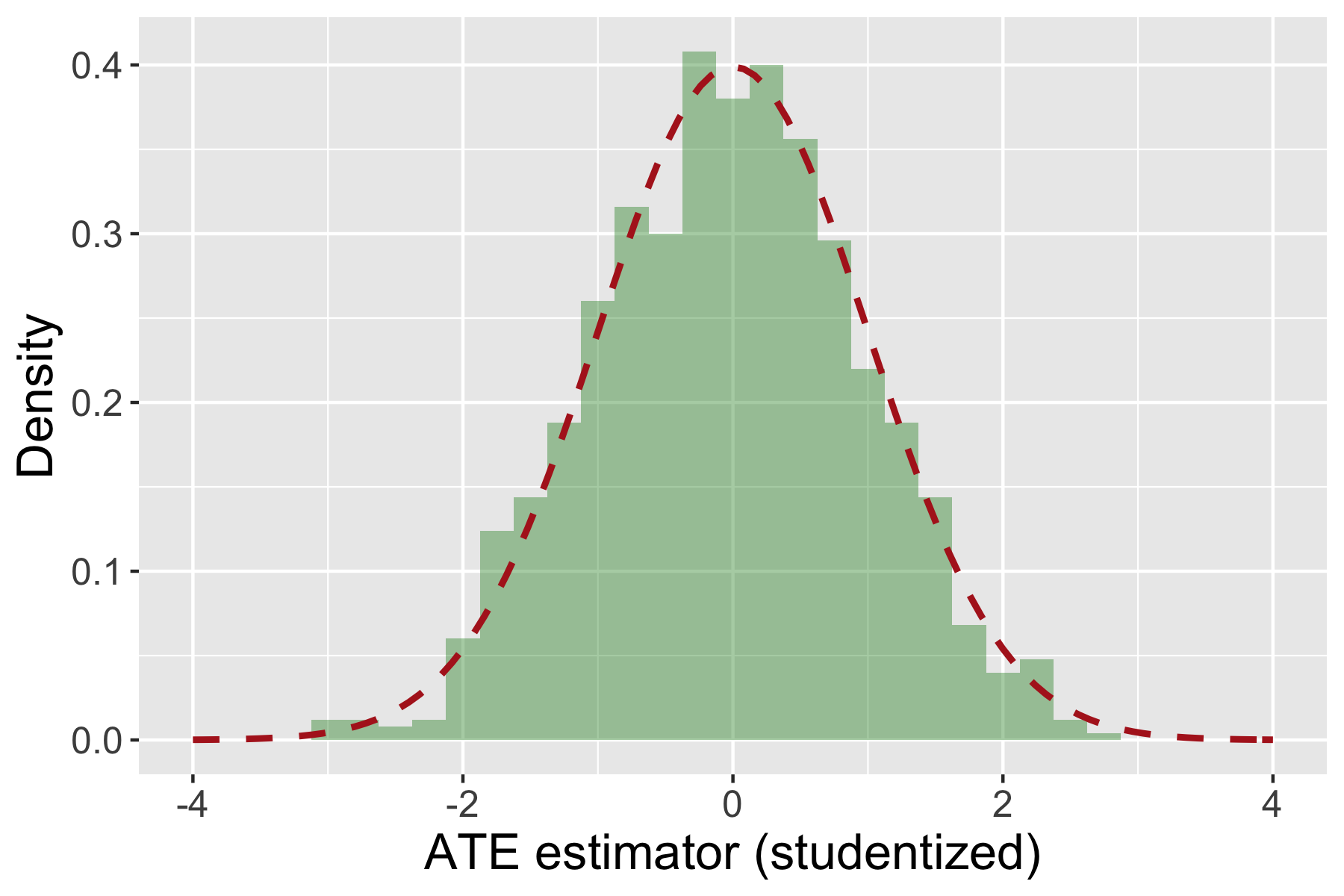}
        \vspace{-3pt}
        \caption{Missing values inference}\label{fig:PCR_missing}
    \end{subsubfigure}
    \begin{subsubfigure}[b]{0.44\textwidth}
        \centering
        \resizebox{\textwidth}{!}{%
        \begin{tabular}{C{11ex} C{6ex} C{5ex} C{5ex} C{8ex} C{8ex}}
            \hline
            Miss. Val. & PCA & ATE & SE & 80\% CI & 95\% CI\tabularnewline
            \hline
            10\% & k=5 & 2.20 & 0.35 & 0.81 & 0.96\tabularnewline
            30\% & k=5 & 2.24 & 0.37 & 0.81 & 0.94\tabularnewline
            50\% & k=5 & 2.35 & 0.41 & 0.79 & 0.94\tabularnewline
            \hline
            \hline
            10\% & k=5 & 2.20 & 0.35 & 0.81 & 0.96\tabularnewline
            10\% & k=7 & 2.19 & 0.37 & 0.81 & 0.95\tabularnewline
            10\% & k=10 & 2.19 & 0.42 & 0.82 & 0.96\tabularnewline
            \hline
        \end{tabular}
        }
        \vspace{-3pt}
        \caption{Missing values coverage}\label{fig:PCR_missing_table}
    \end{subsubfigure}
\end{subfigure}

\begin{subfigure}[c]{\textwidth}
    \setcounter{subfigure}{2}
    \setcounter{subsubfigure}{0}
    \begin{subsubfigure}[b]{0.48\textwidth}
        \centering
        \includegraphics[width=0.8\textwidth]{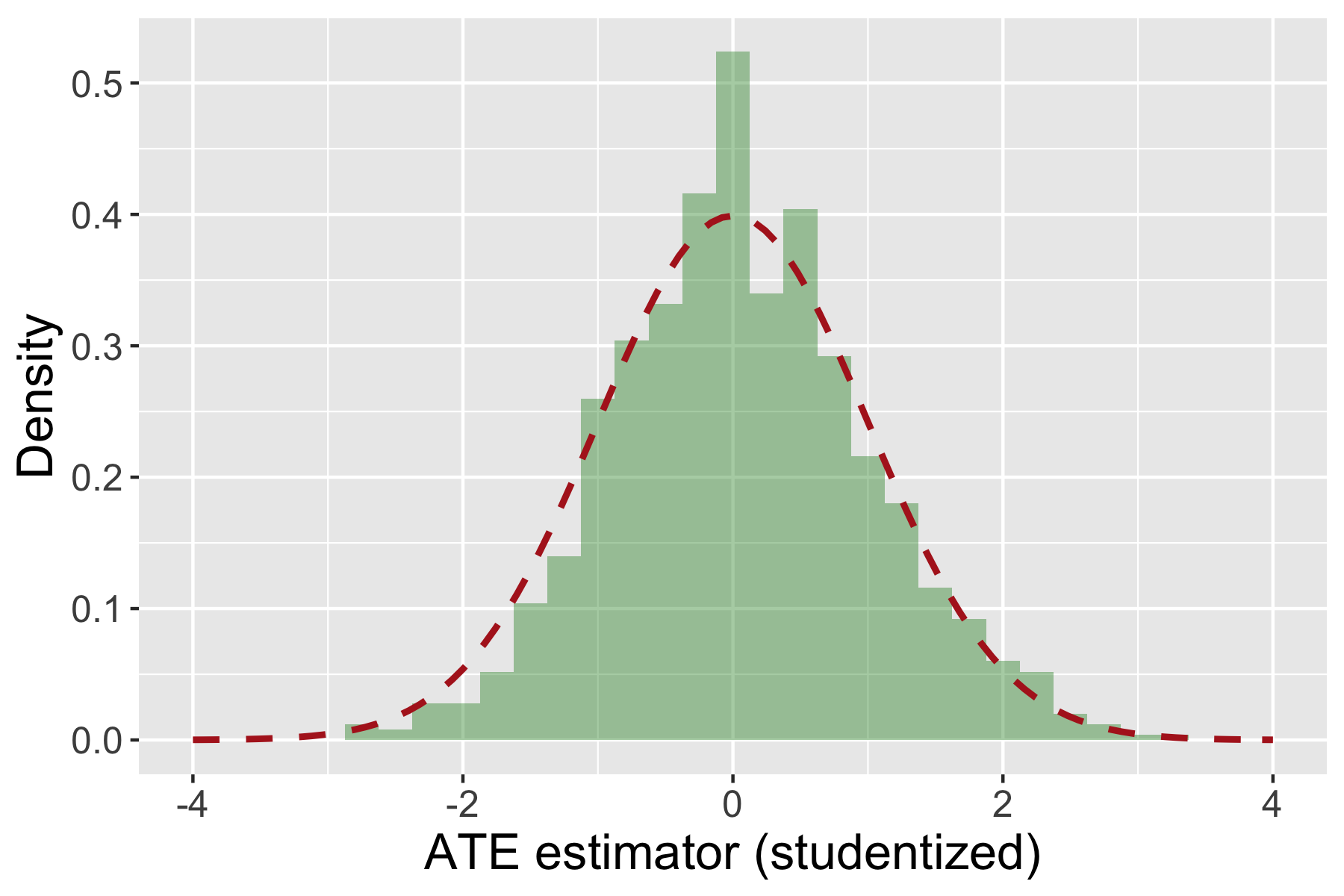}
        \vspace{-3pt}
        \caption{Discretization inference}\label{fig:PCR_discrete}
    \end{subsubfigure}
    \begin{subsubfigure}[b]{0.44\textwidth}
        \centering
        \resizebox{\textwidth}{!}{%
        \begin{tabular}{C{11ex} C{6ex} C{5ex} C{5ex} C{8ex} C{8ex}}
            \hline
            Discret. & PCA & ATE & SE & 80\% CI & 95\% CI\tabularnewline
            \hline
            33\% & k=5 & 2.23 & 0.36 & 0.81 & 0.96\tabularnewline
            33\% & k=7 & 2.23 & 0.37 & 0.80 & 0.95\tabularnewline
            33\% & k=10 & 2.23 & 0.41 & 0.81 & 0.95\tabularnewline
            \hline
        \end{tabular}
        }
        \vspace{12pt}
        \caption{Discretization coverage}\label{fig:PCR_discrete_table}
    \end{subsubfigure}
\end{subfigure}

\begin{subfigure}[c]{\textwidth}
    \setcounter{subfigure}{3}
    \setcounter{subsubfigure}{0}
    \begin{subsubfigure}[b]{0.48\textwidth}
        \centering
        \includegraphics[width=0.8\textwidth]{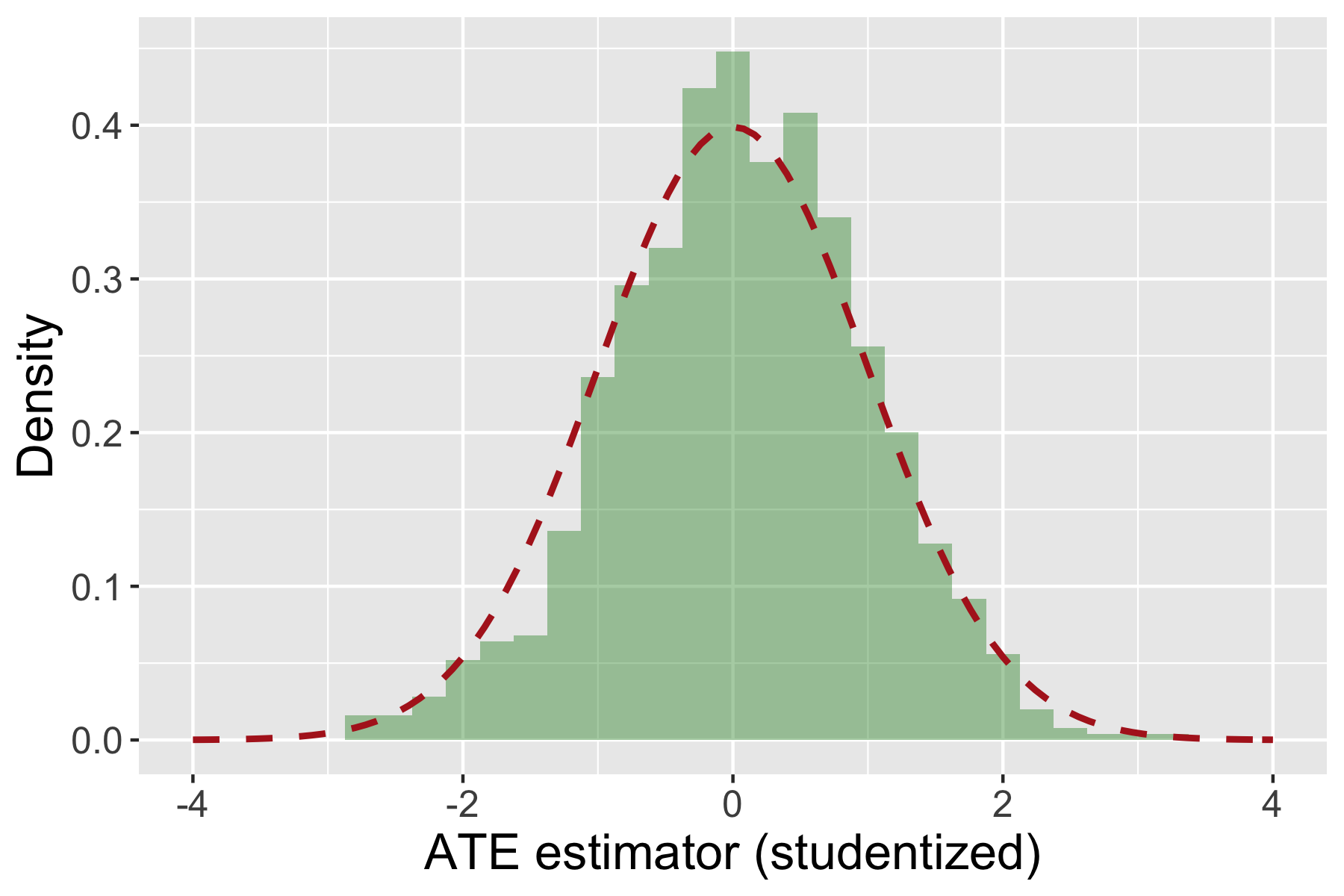}
        \vspace{-3pt}
        \caption{Differential privacy inference}\label{fig:PCR_private}
    \end{subsubfigure}
    \begin{subsubfigure}[b]{0.44\textwidth}
        \centering
        \resizebox{\textwidth}{!}{%
        \begin{tabular}{C{11ex} C{6ex} C{5ex} C{5ex} C{8ex} C{8ex}}
            \hline
            Diff. Priv. & PCA & ATE & SE & 80\% CI & 95\% CI\tabularnewline
            \hline
            20\% & k=5 & 2.19 & 0.35 & 0.84 & 0.97\tabularnewline
            60\% & k=5 & 2.23 & 0.37 & 0.81 & 0.96\tabularnewline
            100\% & k=5 & 2.29 & 0.39 & 0.81 & 0.95\tabularnewline
            \hline
            \hline
            20\% & k=5 & 2.19 & 0.35 & 0.84 & 0.97\tabularnewline
            20\% & k=7 & 2.20 & 0.36 & 0.84 & 0.97\tabularnewline
            20\% & k=10 & 2.19 & 0.39 & 0.86 & 0.97\tabularnewline
            \hline
        \end{tabular}
        }
        \vspace{-3pt}
        \caption{Differential privacy coverage}\label{fig:PCR_private_table}
    \end{subsubfigure}
\end{subfigure}
\vspace{-10pt}
\caption{Our approach adapts to the type and level of corruption.}\label{fig:PCR_all}
\vspace{-15pt}
\end{centering}
\end{figure}

%% file: 5_theory.tex
\section{Finite sample analysis}\label{sec:theory}

In the previous section, we articulate three reasons why inference after data cleaning is hard. First, data cleaning mixes signal and noise across observations. We introduce implicit data cleaning as an algorithmic solution, yet we still need to provide a theory of implicit data cleaning: why is it okay to not clean the test covariates? Second, the best rates of data cleaning are slower than $n^{-1/2}$. We incorporate the doubly robust estimating equation in the hope of achieving double rate robustness, yet we still need to prove that it works: how is causal inference still possible with standard errors of order $\hsigma n^{-1/2}$? Third, data cleaning recovers averages across matrix entries. How can we translate guarantees about recovering averages into guarantees about the coverage of data cleaning-adjusted confidence intervals? In this section, we answer these three theoretical questions with finite sample analysis.

We prove four theorems, each corresponding to a step in the procedure:
(i) data cleaning: $\bhX$ converges to $\bX^{\tr}$;
(ii) error-in-variable regression: $\hgamma$ converges to $\gamma_0$;
(iii) error-in-variable balancing weight: $\halpha$ converges to $\alpha_0$;
(iv) causal parameter: $\p\{\theta_0\in (\htheta \pm 1.96 \hsigma n^{-1/2})\}$ converges to $0.95$.
We have already verified that our key assumption is reasonable in practice for US Census-derived data. In a corollary, we verify that it is reasonable in theory: it holds for a broad class of linear and nonlinear factor models.

\textbf{Step 1: Data cleaning}. For the data cleaning guarantee, we place four assumptions on the corrupted data. To lighten notation, we suppress indexing by \tr.

\begin{assumption}[Bounded signal]\label{assumption:bounded}
There exists an absolute constant $\bar{A}<\infty$ such that for all $i \in [m]$ and $j \in [p]$, $|X_{ij}|\leq \bar{A}$.
\end{assumption}
Bounded true values are standard in the matrix completion literature.

\begin{assumption}[Measurement error]\label{assumption:measurement}
Each row of measurement error $H_{i,\cdot}$ is conditionally mean zero and subexponential, i.e. $\E[H_{i,\cdot}|X_{i,\cdot}]=0$ and there exists $a\geq 1$ and $K_{a}<\infty$ such that $\|H_{i,\cdot}|X_{i,\cdot}\|_{\psi_{a}}\leq K_{a}$. 
Hence there exists $\kappa^2>0$ such that $\|\E[H_{i,\cdot}^TH_{i,\cdot}|X_{i,\cdot}]\|_{op}\leq \kappa^2$. 
We assume measurement error is independent across rows.
\end{assumption}
Measurement error may be \textit{dependent} within a given row. 
If each coordinate of $H_{i,\cdot} \in \R^p$ is independent, then $K_{a}$ and $\kappa^2$ are constants (i.e. they do not scale with $p$) \cite[Lemma 3.4.2]{vershynin2018high}. More generally, $(K_a, \kappa)$ quantify the level of dependence among the entries of $H_{i, \cdot}$ within a row. Our model allows for a great deal of heteroscedasticity. In particular, the results to follow are conditional on $\bX$, so the distribution of $H_{ij}$ may depend on $X_{ij}$ as long as it is conditionally mean zero and has tails no wider than those of an exponential distribution. Assumption~\ref{assumption:measurement} encompasses discretization and differential privacy. 

\begin{assumption}[Missing values]\label{assumption:missing}
Each $\pi_{ij}$ is 1 with probability $\rho_j$ and $\NA$ otherwise. 
Identifying $\NA$ with $0$, we assume there exists $\bar{K}<\infty$ such that $\|\pi_{i,\cdot}-(\rho_1,...,\rho_p)|X_{i,\cdot}\|_{\psi_2}\leq \bar{K}$. 
Missingness $\pi_{i,\cdot}$ is independent of $H_{i,\cdot}$ given $X_{i,\cdot}$, and independent across rows.
%
%
\end{assumption}

Our missingness model generalizes the standard missingness model in the PCR error-in-variable literature in two ways: (i) the missingness of one variable may depend on the missingness of another, and (ii) different variables may be missing with different probabilities. 
These additional degrees of flexibility are crucial for Census data, where non-responses for different variables are often correlated and where non-response rates of different variables can be vastly different. As with measurement error, missingness is independent across rows, but it may be \textit{dependent} within a given row. If each coordinate of $\pi_{i,\cdot} \in \R^p$ is independent, then $\bar{K}$ is constant. More generally, $\bar{K}$ quantifies the level of dependence among the entries of $\pi_{i, \cdot}$ within a row. Our model allows for different probabilities of missingness for different variables, which may depend on the signal in a weak sense: the proof is conditional on $\bX$, so the probability $\rho_j$ may depend on $X_{\cdot,j}$. We define the additional notation
$
 \rho_{\min} \coloneqq \min_{j\in [p]} \rho_j$ and $\boldsymbol{\rho}=diag(\rho_1,...,\rho_p)\in\R^{p\times p}.
$
\begin{assumption}[Concentrated signal]\label{assumption:well_balanced_spectrum}
Consider the approximation $\bX^{\lr}$ to $\bX$, with singular values $s_1,...,s_r$. Assume that $s_1,...,s_r \ge C\sqrt{\frac{mp}{r}}$, where $C$ is an absolute constant.
\end{assumption}

Assumption~\ref{assumption:well_balanced_spectrum} is analogous to incoherence-style conditions in econometrics 
and the notion of pervasiveness in matrix completion. 
Similar to a strong factor assumption,
it ensures that the explanatory power of $\bX^{\lr}$ dominates the explanatory power of various error terms. It requires signal to be spectrally concentrated. A natural setting in which Assumption \ref{assumption:well_balanced_spectrum} holds is when  $X^{\lr}_{ij}=\Theta(1)$ and $s_1,...s_r = \Theta(\tau)$.
Then, for absolute constants $C, C', C'' > 0$,
$
C\cdot r\cdot \tau^2 = \sum_{k}s_k^2=\|\bX^{\lr}\|^2_{Fr}= C' \cdot mp  
$
which implies $\tau = C'' \sqrt{\frac{mp}{r}}$. Future work may extend our results to different spectral assumptions on $\bX^{\lr}$.

\begin{remark}
We parametrize our rates by the quality of low rank approximation.
\end{remark}
Without loss of generality, $\bX=\bX^{\lr}+\bE^{\lr}$, where $\bX^{\lr}$ is a low rank approximation to $\bX$, and $\bE^{\lr}$ is the approximation residual. The two key quantities are $r=rank\{\bX^{\lr}\}$ and $\Delta_E=\|\bE^{\lr}\|_{\max}$. It is \textit{with} loss of generality that $r$ and $\Delta_E$ are simultaneously well behaved. Intuitively, as $r$ decreases, $\Delta_E$ increases (and vice-versa). Indeed, if $\bX^{\lr}=\bX$ then trivially $r \le (m, p)$ and $\Delta_E = 0$; if $\bX^{\lr} = 0$, then $r = 0$ but $\Delta_E = \bar{A}$. Our corollary shows that, under a nonlinear factor model, both $r$ and $\Delta_E$ behave well: $r\ll (m, p)$ and $\Delta_E\rightarrow0$. Until that corollary, we parameterize rates by $(r,\Delta_E)$, which may be non-unique.\footnote{Since $r$ may be non-unique, there may be multiple valid choices of the hyperparameter $k$.} 

\begin{theorem}[Finite sample data cleaning rate]\label{theorem:cov}
Suppose Assumptions~\ref{assumption:bounded},~\ref{assumption:measurement},~\ref{assumption:missing}, and~\ref{assumption:well_balanced_spectrum} hold, $k=r$, and $\rho_{\min} > \frac{23\log (mp)}{m}$. Then for an absolute constant $C > 0$,
$$
\frac{1}{m}\E\|\bhX-\bX\|^2_{2,\infty}\leq C_1
    \cdot \frac{r\ln^5(mp)}{\rho_{\min}^4} 
    \left(
    \frac{1}{m}
    +\frac{1}{p}
    +\Delta_E^2
    \right),
$$
where $C_1=C\cdot \bar{A}^4 (K_a + \bar{K})^2 (\kappa +K_a+ \bar{K})^2$.
\end{theorem}

The norm of convergence is called the $(2,\infty)$ norm:
$
\frac{1}{m}\|\bhX-\bX\|^2_{2,\infty}=\max_{j\in [p]}\frac{1}{m}\|\hat{X}_{i,\cdot}-X_{\cdot,j}\|_2^2=\max_{j\in [p]} \frac{1}{m}\sum_{i=1}^m (\hat{X}_{ij}-X_{ij})^2
$
i.e. a maximum across columns and an average across rows. For any given variable $j\in[p]$, Theorem~\ref{theorem:cov} guarantees that data cleaning performs well on average across observations $i\in[m]$. Our rate requires both $m$ and $p$ to increase: more repeated measurements improve the quality of data cleaning. For the bound to be meaningful, $(r,\Delta_E)$ must be simultaneously well behaved, which is our key assumption. Recall that $(K_a,\kappa,\bar{K})$ quantify the level of corruption dependence within a row. As long as the dependence is weak, e.g. $(K_a,\kappa,\bar{K})$ scale as some power of $\ln(mp)$, this dependence in negligible. Our downstream results for the error-in-variable regression and balancing weight build on this data cleaning guarantee. Signal is spectrally concentrated, while noise is spectrally diffuse, so we can concentrate out the noise.

\textbf{Step 2: Error-in-variable regression}. 
We place three additional assumptions.
\begin{assumption}[Response noise]\label{assumption:noise}
We have $\E[\varepsilon_i | X_{i,\cdot}]=0$ and $\V[\varepsilon_i | X_{i,\cdot}]\leq \bar{\sigma}^2$. 
Response noise $\varepsilon_i$ is independent of $H_{i,\cdot}$ and $\pi_{i,\cdot}$ given $X_{i,\cdot}$, and independent across rows.
\end{assumption}
This condition permits measurement error and differential privacy of the outcome $Y_i$. 
Next we assume \tr\, and \te\, each contains a sufficient variety of observations. For a matrix $\bM \in \R^{m \times p}$, we define its row space as $\row(\bM)=span\{M_{i,\cdot}\}$. 

\begin{assumption}[Row space inclusion]\label{assumption:inclusion_row} $\row[b\{\bX^{\lr,\tr}\}] =\row[b\{\bX^{\lr,\te}\}]$.
\end{assumption}
This property permits $\bX^{\lr, \tr}\neq\bX^{\lr, \te}$, and also permits the two matrices to have different SVDs. In Appendix~\ref{sec:regression_proof}, we verify that Assumption~\ref{assumption:inclusion_row} holds with high probability under weak auxiliary conditions. Finally, we place a weak technical condition.

\begin{assumption}[Well conditioned estimators]\label{assumption:bounded_response_mod}
Let 
$\hs'_{k'}$ be the smallest non-zero singular value of $b(D^{\tr},\bhX)$. Assume that $\hs'_{k'}\gtrsim \frac{\bar{\varepsilon}}{\text{\normalfont polynomial}(m,p)}$ where $\E[\varepsilon_i^8]\leq \bar{\varepsilon}^8$.
\end{assumption}
For $(\hbeta,\heta)$ to be well conditioned, the singular value $\hat{s}'_{k'}$ should not be too small. In particular, it must be bounded below by an arbitrary negative power of $m$ and $p$. 

Before stating the result, we introduce a theoretical device $\beta^*$ as the coefficient of the best low rank nonlinear approximation to $\gamma_0$. In particular, we write
$
\gamma_0(D_i,X_{i,\cdot})=b(D_i,X^{\lr}_{i,\cdot})\beta^*+\phi_i^{\lr}
$
where $\phi_i^{\lr}$ is the approximation error.  It will be convenient to keep track of this approximation error by defining $\phi_i:=\gamma_0(D_i,X_{i,\cdot})-b(D_i,X_{i,\cdot})\beta^*$. There will be a trade-off: a richer dictionary $b$ leads to a smaller approximation error in terms of $\|\phi\|^2_2$, but more compounding of measurement error and missingness. 
The following result helps to characterize how the compounded data corruption magnifies $(\rho_{\min}^{-1},r,\Delta_E)$ but nothing else.

\begin{remark}
Our results hold for a broad class of dictionaries, with the dictionary-specific constant $C_b'$ and the concise notation $(\rho'_{\min},r',\Delta_E')$ in Theorems~\ref{theorem:fast_rate} and~\ref{theorem:fast_rate_RR}. Appendix~\ref{sec:dictionary} proves that
$$
C_b'\leq 2^{d_{\max}} \cdot \bar{A}^{2d_{\max}}_{\max}\cdot  \|\bhX\|_{\max}^{2d_{\max}},\quad
\frac{1}{\rho'_{\min}} \leq \frac{d_{\max}\bar{A}^{d_{\max}}}{\rho_{\min}},
$$
$$
r'\leq r^{d_{\max}},\quad\textrm{and}\quad
\Delta_E'\leq C \bar{A}^{d_{\max}} \cdot d_{\max} \Delta_E,
$$
where $d_{\max}$ is the degree of the polynomial dictionary. Appendix~\ref{sec:dictionary} articulates restrictions on the class of dictionaries. For the interacted dictionary, $d_{\max}=2$.
\end{remark}

\begin{remark}
Under further incoherence-style assumptions, we bound $\|\bhX\|_{\max}\leq C\sqrt{r}$ in Appendix~\ref{sec:cleaning_proof}. Alternatively, one can bound
$$
\|\bhX\|_{\max} \leq \|\bhX-\bX\|_{\max}+\|\bX\|_{\max}\leq \|\bhX-\bX\|^2_{2,\infty}+\bar{A}
$$
then appeal to Theorem~\ref{theorem:cov} with high probability. Doing so for $C_b'$ does not affect the powers of $(m,p)$ but does increase the complexity of the pre-factors.
\end{remark}

\begin{theorem}[Finite sample error-in-variable regression rate]\label{theorem:fast_rate}
Suppose that the conditions of Theorem~\ref{theorem:cov} hold, as well as Assumptions~\ref{assumption:noise},~\ref{assumption:inclusion_row}, and~\ref{assumption:bounded_response_mod}. 
If we have that \linebreak
$\rho'_{\min} \gg \tilde{C}\sqrt{r'}\ln^{\frac{3}{2}}(mp)\left\{\frac{1}{\sqrt{p}}\vee  \frac{1}{\sqrt{m}}\vee \Delta_E \right\}$,
where $\tilde{C}:= C  \bar{A} \Big(\kappa + \bar{K} + K_a \Big)$, then
\begin{multline*}
\Rc(\hgamma) 
\le 
 C_b'C_1C_2
    \cdot \bar{\sigma}^2 \cdot \frac{(r')^3\ln^{8}(mp)}{(\rho'_{\min})^6}\| \beta^* \|^2_1
    \left(
    \frac{1}{m}+\frac{p}{m^2}+\frac{1}{p}+\left(1+\frac{p}{m}\right)(\Delta'_E)^2+p(\Delta'_E)^4
    \right) \\
    + C_2
    \cdot
    \frac{(r')^2\ln^3(mp)}{(\rho'_{\min})^2}
    \Delta_{\phi}
    \left(1 + (\Delta'_E)^2\right),
\end{multline*}
where
$
\Delta_{\phi}=\frac{1}{m}\|\phi^{\tr} \|_2^2 \vee \frac{1}{m}\|\phi^{\te} \|_2^2,
$
$$
C_1=C\bar{A}^4(K_a + \bar{K})^2 (\kappa + \bar{K} + K_a)^2,
\quad\textrm{and}\quad
 C_2=C\cdot \bar{A}^4 (\kappa + \bar{K} + K_a)^2.
$$
\end{theorem}

\begin{corollary}[Simplified regression rate]\label{corollary:fast_rate}
Suppose the conditions of Theorem~\ref{theorem:fast_rate} hold. Further suppose $\gamma_0$ is exactly linear in signal, which is exactly low rank. Then
$$
\Rc(\hgamma) 
\le 
 C_1C_2
    \cdot \bar{\sigma}^2 \cdot \frac{r^3\ln^{8}(mp)}{\rho_{\min}^6}\| \beta^* \|^2_1
    \left(
    \frac{1}{m}+\frac{p}{m^2}+\frac{1}{p}\right).
$$
\end{corollary}

The norm of convergence is
$
\Rc(\hgamma)=\E\left[\frac{1}{m}\sum_{i\in \te}\{\hgamma(D_i,Z_{i,\cdot})-\gamma_0(D_i,X_{i,\cdot})\}^2\right],
$
a relaxation of mean square error, where the expectation is over randomness in \tr\, and \te. Two aspects of our problem necessitate this norm: (i) given the on-average data cleaning guarantee in Theorem~\ref{theorem:cov}, this is the best we can do; and (ii) for i.n.i.d. data, a population risk is otherwise not well defined.\footnote{Interestingly, even with i.i.d. data, (i) necessitates this norm.} Since the estimator $\hgamma$ does not involve cleaning \te, Theorem~\ref{theorem:fast_rate} provides the theory of implicit data cleaning. The bound requires both $m$ and $p$ to increase, $p \ll m^2$, and $\rho_{\min}\gg p^{-1/2}\vee m^{-1/2} \vee \Delta_E$. For the bound to be meaningful, $(r,\Delta_E)$ must be simultaneously well behaved and the corruption dependence must be weak. Finally, the bound includes the nonlinear approximation error $\Delta_{\phi}$ and the size of the theoretical device $\|\beta^*\|_1$, which is well behaved if $\beta^*$ is approximately sparse. In summary, we keep track of the low rank approximation error $\Delta_E$ and the nonlinear sparse approximation error $\Delta_\phi$. To deal with $\Delta_E$, we demonstrate that nonlinear factor models admit low rank approximation below. Due to our doubly robust approach, estimation of the causal parameter $\theta_0$ is robust to non-vanishing $\Delta_{\phi}$---a discussion we revisit later.

We make several innovations relative to previous work on PCR. First, we propose an error-in-variable regression estimator that does not clean the test covariates, and we develop a new theory of implicit data cleaning. Second, we define a new norm of convergence which we subsequently use in causal inference. Appendix~\ref{sec:regression_proof} compares our norm with those in previous work. Third, we allow for dependence of missingness across variables and for different probabilities of missingness across variables. This flexibility is realistic for Census data. Fourth, we consider a nonlinear regression function $\gamma_0$ that is approximated by a nonlinear dictionary of basis functions $b$. The dictionary of basis functions is important for causal inference because it allows for treatment effect heterogeneity, and it requires a novel characterization of which nonlinearities do not compound data corruption too much.

\textbf{Step 3: Error-in-variable balancing}. We place one additional assumption.
\begin{assumption}[Row space inclusion]\label{assumption:inclusion_M}
$
\hat{M} \in \row\{b(D^{\tr},\bhX)\}.
$
\end{assumption}

Whereas Assumption~\ref{assumption:inclusion_row} is about the low rank approximation of the signal across \tr\, and \te, Assumption~\ref{assumption:inclusion_M} is about the counterfactual moment in relation to \tr\, after cleaning. With $\hat{M}=[\{b(D^{\tr},\bhX)\}^T Y^{\tr}]$, which reverts to error-in-variable regression, Assumption~\ref{assumption:inclusion_M} immediately holds. In other cases, it limits the counterfactual queries that an analyst may ask. Because it concerns empirical quantities, it may be viewed as a diagnostic tool to determine whether the counterfactual can be extrapolated.

As before, we introduce a theoretical device $\eta^*$ as the coefficient of the best low rank nonlinear approximation to $\alpha_0$. In particular, we write
$
\alpha_0(D_i,Z_{i,\cdot})=b(D_i,X^{\lr}_{i,\cdot})\eta^*+\zeta_i^{\lr}
$
where $\zeta_i^{\lr}$ is the approximation error, and we study this approximation error by defining $\zeta_i:=\alpha_0(D_i,Z_{i,\cdot})-b(D_i,X_{i,\cdot})\eta^*$.\footnote{A further assumption that the treatment mechanism only depends on signal, i.e. $\E[D_i|X_{i,\cdot},H_{i,\cdot},\pi_{i,\cdot}]=\E[D_i|X_{i,\cdot}]$, implies $\alpha_0(D_i,Z_{i,\cdot})=\alpha_0(D_i,X_{i,\cdot})=b(D_i,X^{\lr}_{i,\cdot})\eta^*+\zeta_i^{\lr}$.} Again, there will be a trade-off: a richer dictionary $b$ leads to a smaller approximation error in terms of $\|\zeta\|^2_2$, but amplification of $(\rho^{-1}_{\min},r,\Delta_E)$. 

\begin{remark}
Our results hold for a broad class of causal parameters, with parameter-specific constants $(C_m',C_m'')$ in Theorem~\ref{theorem:fast_rate_RR}. Appendix~\ref{sec:riesz_proof} characterizes $(C_m',C_m'')$ for several examples. For ATE with the interacted dictionary,
$
C_m'=1$ and $C_m''=\bar{A}.
$
\end{remark}

\begin{theorem}[Finite sample error-in-variable balancing weight rate]\label{theorem:fast_rate_RR}
Suppose the conditions of Theorem~\ref{theorem:cov} hold, as well as Assumptions~\ref{assumption:inclusion_row},~\ref{assumption:bounded_response_mod}, and~\ref{assumption:inclusion_M}. If $\rho'_{\min} \gg \tilde{C}\sqrt{r'}\ln^{\frac{3}{2}}(mp)\left\{\frac{1}{\sqrt{p}}\vee  \frac{1}{\sqrt{m}}\vee \Delta_E \right\}$ and $\|\alpha_0\|_{\infty}\leq \bar{\alpha}$, then
\begin{multline*}
    \Rc(\halpha)
    \leq  C_3 \cdot \frac{(r')^5 \ln^{13}(mp)}{(\rho'_{\min})^{10}} \|\eta^*\|^2_1
    \cdot 
    \bigg\{
    \frac{1}{m}+\frac{1}{p}+\frac{p}{m^2}+\frac{m}{p^2}+\left(1+\frac{p}{m}+\frac{m}{p}\right)(\Delta'_E)^2 \\
    + (m+p)(\Delta'_E)^4+mp (\Delta'_E)^6
    \bigg\}
    + 2\Delta_{\zeta},
\end{multline*}
where $\Delta_{\zeta}=\frac{1}{m}\|\zeta^{\tr} \|_2^2 \vee \frac{1}{m}\|\zeta^{\te} \|_2^2$ and 
$$
    C_3=C \bar{A}^{14}(C_b'+\sqrt{C_m'}+C_m''+\bar{\alpha}+\bar{A})^2 (K_a + \bar{K})^4(\kappa + \bar{K} + K_a)^6.
$$
\end{theorem}

\begin{corollary}[Simplified balancing weight rate]\label{corollary:fast_rate_RR}
Suppose the conditions of Theorem~\ref{theorem:fast_rate_RR} hold. Further suppose $\alpha_0$ is exactly linear in signal, which is exactly low rank. Then
$$
\Rc(\halpha)
 \leq  \tilde{C}_3 \cdot \frac{r^5 \ln^{13}(mp)}{\rho^{10}_{\min}} \|\eta^*\|^2_1 \cdot 
    \left\{  \frac{1}{m}+\frac{1}{p}+\frac{p}{m^2}+\frac{m}{p^2}
\right\},
$$
where $\tilde{C}_3$ is $C_3$ but with $C_b'$ replaced by $1$.
\end{corollary}

The norm of convergence
$
\Rc(\halpha)=\E\left[\frac{1}{m}\sum_{i\in \te}\{\halpha(D_i,Z_{i,\cdot})-\alpha_0(D_i,Z_{i,\cdot})\}^2\right]
$
relaxes mean square error as before. Theorem~\ref{theorem:fast_rate_RR} imposes a stronger condition on $p$ than Theorem~\ref{theorem:fast_rate}: now, we need $m^{1/2} \ll p \ll m^2$. Once again, our bound keeps track of the 
low rank approximation error $\Delta_E$ and the nonlinear sparse approximation error $\Delta_{\zeta}$. Nonlinear factor models imply that the former vanishes, and our doubly robust approach allows the latter not to vanish, as we make precise below.

Theorem~\ref{theorem:fast_rate_RR} innovates in several ways. Most importantly, it analyzes a new estimator for a new estimand: the error-in-variable balancing weight in cross sectional data. A rich literature proposes balancing weight estimators for causal inference with clean data, but to our knowledge, ours is the first error-in-variable balancing weight estimator for causal inference with corrupted cross sectional data. Appendix~\ref{sec:riesz_proof} shows that Theorem~\ref{theorem:fast_rate_RR} holds for a broad class of counterfactual moments and therefore a broad class of causal parameters. 

The counterfactual moment $\hat{M}=[\{b(D^{\tr},\bhX)\}^T Y^{\tr}]$ recovers error-in-variable regression. We choose not to simply subsume Theorem~\ref{theorem:fast_rate} by Theorem~\ref{theorem:fast_rate_RR} for two reasons. First, doing so would require that $Y_i$ and $\varepsilon_i$ are bounded, which rules out differential privacy for the outcome. Second, Theorem~\ref{theorem:fast_rate} has lower powers of $(r,\rho_{\min}^{-1})$ and avoids the term $\frac{m}{p^2}$ so it is typically a tighter bound. If we are willing to accept these costs, then the application of Theorem~\ref{theorem:fast_rate_RR} to error-in-variable regression relaxes $\varepsilon_i \indep H_{i,\cdot},\pi_{i,\cdot} |  X_{i,\cdot}$ in Assumption~\ref{assumption:noise}. Both approaches allow for heteroscedasticity of $\V[\varepsilon_i|X_{i,\cdot}]$ in the traditional sense.

\textbf{Step 4: Causal estimation and inference}.
The corrupted data problem is an extended semiparametric problem. Let $W_{i,\cdot}=(D_i,X_{i,\cdot},H_{i,\cdot},\pi_{i,\cdot})$ concatenate the signal and the noise, so that $\mathbb{L}_2(\Wc)$ consists of square integrable functions of the form $\gamma:(D_i,X_{i,\cdot},H_{i,\cdot},\pi_{i,\cdot})\rightarrow \R$. Both the true regression $\gamma_0(D_i,X_{i,\cdot})$ and our error-in-variable estimator $\hgamma(D_i,Z_{i,\cdot})$ belong to this space, which serves as our hypothesis space for semiparametric analysis. 

\begin{assumption}[Distribution shift]\label{assumption:invariance_mod}
The extended outcome and treatment mechanisms, $\E[Y_i|D_i,X_{i,\cdot},H_{i,\cdot},\pi_{i,\cdot}]$ and $\E[D_i|X_{i,\cdot},H_{i,\cdot},\pi_{i,\cdot}]$, do not vary across observations.
\end{assumption}
Assumption~\ref{assumption:invariance_mod} implies that $\gamma_0(W_{i,\cdot})$ and $\alpha_0(W_{i,\cdot})$ do not vary across observations, though the marginal distributions $\p_i(W_i)$ may vary. Our corruption model implies $\gamma_0(W_{i,\cdot})=\gamma_0(D_i,X_{i,\cdot})$, and we are agnostic about whether $\alpha_0(W_{i,\cdot})=\alpha_0(D_i,X_{i,\cdot})$ for the extended hypothesis space.\footnote{If $\E[D_i|X_{i,\cdot},H_{i,\cdot},\pi_{i,\cdot}]=\E[D_i|X_{i,\cdot}]$, then $\alpha_0(W_{i,\cdot})=\alpha_0(D_i,X_{i,\cdot})$.} Our final assumption mildly strengthens common support.

\begin{assumption}[Bounded propensity]\label{assumption:mean_square_cont_mod}
The extended propensity score is bounded below and above, i.e. $1-\bar{\phi}\leq \E[D_i|X_{i,\cdot},H_{i,\cdot},\pi_{i,\cdot}]\leq\bar{\phi}$.\footnote{Our finite sample analysis allows $\bar{\phi}\uparrow 1$, and more generally $\bar{Q}\uparrow \infty$ in Remark~\ref{remark:inference_general}, as the sample size $n\uparrow \infty$.}
\end{assumption}

We introduce some additional notation to state the finite sample Gaussian approximation. Define the oracle influences $\psi_i=\psi(W_{i,\cdot},\theta_i,\gamma_0,\alpha_0)$, where the influence function is
$$
\psi(W_{i,\cdot},\theta,\gamma,\alpha)=\gamma(1,X_{i,\cdot},H_{i,\cdot},\pi_{i,\cdot})-\gamma(0,X_{i,\cdot},H_{i,\cdot},\pi_{i,\cdot})+\alpha(W_{i,\cdot})\{Y_i-\gamma(W_{i,\cdot})\}-\theta.
$$
$\E[\psi_i]=0$ since $\E[\gamma_0(1,X_{i,\cdot})-\gamma_0(0,X_{i,\cdot})]=\theta_i$ and $\E[\alpha_0(W_{i,\cdot})\{Y_i-\gamma_0(W_{i,\cdot})\}]=0$ by law of iterated expectations. We define the higher moments and average higher moments by
\begin{align*}
\sigma_i^2 & =\E[\psi_i^2], & \xi_i^3 & =\E[|\psi_i|^3], & \chi_i^4 & =\E[\psi_i^4]; \\
\sigma^2 & =\frac{1}{n}\sum_{i=1}^n \sigma_i^2, & \xi^3 & =\frac{1}{n}\sum_{i=1}^n \xi_i^3, & \chi^4 & =\frac{1}{n}\sum_{i=1}^n \chi_i^4.
\end{align*}

\begin{remark}\label{remark:inference_general}
Our results hold for a broad class of causal parameters, with parameter-specific constants $(\bar{Q},\bar{q})$ in Theorems~\ref{thm:dml_inid} and~\ref{thm:var_inid}. For ATE,
$\bar{Q}=2\left(\frac{1}{\bar{\phi}}+\frac{1}{1-\bar{\phi}}\right)$ and $\bar{q}=1$ under Assumptions~\ref{assumption:invariance_mod} and~\ref{assumption:mean_square_cont_mod}. Appendix~\ref{sec:target_proof} characterizes $(\bar{Q},\bar{q})$ for several other examples under generalizations of Assumptions~\ref{assumption:invariance_mod} and~\ref{assumption:mean_square_cont_mod}. $\bar{Q}$ may be a diverging sequence. 
\end{remark}

\begin{theorem}[Finite sample Gaussian approximation]\label{thm:dml_inid}
Suppose Assumptions~\ref{assumption:invariance_mod} and~\ref{assumption:mean_square_cont_mod} hold,
$
\V[\varepsilon_i \mid W_{i,\cdot}]\leq \bar{\sigma}^2$, $\|\alpha_0\|_{\infty}\leq\bar{\alpha},
$
and for $(i,j)\in \te$,
$
    \hgamma(W_{i,\cdot})\indep \hgamma(W_{j,\cdot})|\tr$ and $\halpha(W_{i,\cdot})\indep \halpha(W_{j,\cdot})|\tr.
$
Then with probability $1-\epsilon$,
$$
\sup_{z\in\mathbb{R}} \left| \p \left\{\frac{n^{1/2}}{\sigma}(\htheta-\theta_0)\leq z\right\}-\Phi(z)\right|\leq 0.56\left(\frac{\xi}{\sigma}\right)^3 n^{-\frac{1}{2}}+\frac{\Delta}{(2\pi)^{1/2}}+\epsilon,
$$
where $\Phi(z)$ is the standard Gaussian distribution function and
$$
\Delta=\frac{3 L}{\epsilon   \sigma}\left[(\bar{Q}^{1/2}+\bar{\alpha})\{\Rc(\hgamma)\}^{\bar{q}/2}+\bar{\sigma}\{\Rc(\halpha)\}^{1/2}+\{n \Rc(\hgamma) \Rc(\halpha) \}^{1/2}\right].
$$
\end{theorem}

\begin{theorem}[Finite sample variance estimation]\label{thm:var_inid}
Suppose Assumptions~\ref{assumption:invariance_mod} and~\ref{assumption:mean_square_cont_mod} hold, 
$\V[\varepsilon_i \mid W_{i,\cdot }]\leq \bar{\sigma}^2$, and 
$\|\halpha\|_{\infty}\leq\bar{\alpha}'$. Then with probability $1-\epsilon'$,
\begin{multline*}
    \left|\hat{\sigma}^2-(\sigma^2+\bias)\right| \leq \Delta' + \Delta''
    + 3\big[(\Delta')^{1/2}\left\{(\Delta'')^{1/2}+\sigma+\Delta_{\out}^{1/2}\right\} \\
    + (\Delta'')^{1/2}\left\{\Delta_{\out}^{1/2}+(\Delta')^{1/4}\Delta_{\out}^{1/4}\right\}
    + (\Delta')^{1/4}\Delta_{\out}^{1/4}\sigma\big],
\end{multline*}
where
$$
\bias=\Delta_{\out}+2\Delta_{\out}^{1/2}\sigma,\quad
\Delta_{\out}=\frac{1}{n}\sum_{i=1}^n[(\theta_i-\theta_0)^2],
$$
$$
 \Delta'=4(\htheta-\theta_0)^2+\frac{24 L}{\epsilon'}\left[\left\{\bar{Q}+(\bar{\alpha}')^2\right\}\Rc(\hgamma)^{\bar{q}}+\bar{\sigma}^2\Rc(\halpha)\right],
 \quad\textrm{and}\quad
 \Delta''=\left(\frac{2}{\epsilon'}\right)^{1/2}\chi^2 n^{-\frac{1}{2}}.
$$
\end{theorem}

\begin{corollary}[Confidence interval coverage]\label{cor:CI_inid}
Suppose the conditions of Theorems~\ref{thm:dml_inid} and~\ref{thm:var_inid} hold. Further assume
(i) moment regularity: $\{(\xi/\sigma)^3+\chi^2\}n^{-\frac{1}{2}}\rightarrow0$;
(ii) error-in-variable regression rate: $\left(\bar{Q}^{1/2}+\bar{\alpha}/\sigma+\bar{\alpha}'\right)\{\Rc(\hgamma)\}^{\bar{q}/2}\rightarrow 0$;
(iii) error-in-variable balancing weight rate: $\bar{\sigma}\{\Rc(\halpha)\}^{1/2}\rightarrow 0$;
(iv) product of rates is fast: $\{n \Rc(\hgamma) \Rc(\halpha)\}^{1/2} /\sigma \rightarrow0$. Then
$
\hat{\theta}\overset{p}{\rightarrow}\theta_0$, $\hsigma^2\overset{p}{\rightarrow} \sigma^2+\bias$, and $\p\{\theta_0\in (\htheta \pm 1.96 \hsigma n^{-1/2})\} \rightarrow 0.95+c $ where $\bias,c\geq 0.
$
If in addition $\Delta_{\out}\rightarrow 0$, i.e. there are not too many outliers, then
$\hat{\theta}\overset{p}{\rightarrow}\theta_0$, $\hsigma^2\overset{p}{\rightarrow} \sigma^2$, and $\p\{\theta_0\in (\htheta \pm 1.96 \hsigma n^{-1/2})\} \rightarrow 0.95.
$
\end{corollary}

\begin{remark}
Corollary~\ref{cor:CI_inid} holds for a broad class of semiparametric estimands such as the average elasticity and nonparametric estimands such as heterogeneous treatment effects. Moreover, it holds for not only the data cleaning and estimation procedure that we propose, but for any data cleaning and estimation procedure satisfying its weak conditions. 
\end{remark}

The rate conditions $\Rc(\hgamma)\rightarrow 0$, $\Rc(\halpha)\rightarrow 0$ , and $\{n \Rc(\hgamma) \Rc(\halpha) \}^{1/2}\rightarrow 0$ suffice for Gaussian approximation with standard deviation $\sigma n^{-1/2}$, generalizing the main result in \cite{chernozhukov2021simple} to the harder setting with corrupted and i.n.i.d. data. These rate conditions are in terms of a more general norm than previous work because of matrix completion in the data cleaning step. Nonetheless, we recover a familiar product rate condition from semiparametric theory. The conditions solve the two remaining theoretical challenges. First, they provide a framework to translate an on-average data cleaning guarantee into a data cleaning-adjusted confidence interval for the causal parameter, by using generalized norms. Second, they ensure that the standard deviation is $\sigma n^{-1/2}$ as long as the \textit{product} of error-in-variable rates (and hence the product of data cleaning rates) is of order $n^{-1/2}$. In summary, they allow for causal inference at rates faster than matrix completion, which is essential to achieving precision for the population while maintaining privacy for individuals.

A technical innovation is semiparametric variance estimation in the i.n.i.d. setting, which is essential to the validity of confidence intervals. We define $\Delta_{\out}$ to quantify the frequency of outliers. Since $\theta_i=\E[\gamma_0(1,X_{i,\cdot})-\gamma_0(0,X_{i,\cdot})]$, $\Delta_{\out}$ quantifies the shift in the marginal distributions of true covariates $\p_i(X_{i,\cdot})$. At best, $\Delta_{\out}=0$ in the i.i.d. case. At worst, $\Delta_{\out}$ is a constant (when individual treatment effects are bounded). The condition $\Delta_{\out}\rightarrow 0$, i.e. relatively few outliers, suffices for consistent variance estimation and nominal confidence intervals. When $\Delta_{\out} \not\rightarrow 0$, our variance estimator is asymptotically biased upwards by $\bias=\Delta_{\out}+2\Delta_{\out}^{1/2}\sigma$, implying conservative confidence intervals. At worst, our confidence intervals are valid but conservative by a theoretically quantifiable amount. 

Our exact characterization of $\bias$ may have broader consequences for design-based inference. Future work may study properties of our procedure in randomized experiments.

Data corruption only appears in the asymptotic variance $\sigma^2$ via the error-in-variable balancing weight $\alpha_0$. In the ATE example, noise appears in the asymptotic variance when the treatment mechanism depends on both signal and noise. If the treatment mechanism depends on signal alone, then our causal estimator implemented on corrupted data is asymptotically as efficient as our causal estimator implemented on clean data.

\textbf{Key assumption holds for nonlinear factor models}.
Finally, we tie together our various results and revisit our key assumption that covariates are approximately low rank. We show that nonlinear factor models (i) encode the intuition of approximate repeated measurements; (ii) imply that covariates are approximately low rank; and (iii) satisfy the rate conditions for causal inference. In a nonlinear factor model, $X_{ij}=g(\lambda_i,\mu_j)$ where $(\lambda_i,\mu_j)$ are latent factors corresponding to units and covariates, respectively. We assume that the function $g$ is smooth in its second argument, formalizing the repeated measurement intuition.

\begin{assumption}[Generalized factor model]\label{assumption:factor_model}
Assume $\bX$ is generated as
$
X_{ij}=g(\lambda_i,\mu_j)
$,
where $\lambda_i,\mu_j\in [0,1)^{q}$ and $g(\lambda_i,\cdot)\in\Hc(q,S,C_H)$. Here, $\Hc(q,S,C_H)$ is the H\"older class of functions $g: [0, 1)^q \to \mathbb{R}$ whose partial derivatives satisfy
$$
\sum_{s: |s| = \lfloor S \rfloor} \frac{1}{s!} |\nabla_s g(\mu) - \nabla_s g(\mu') | \le C_H \norm{\mu - \mu'}_{\max}^{S - \lfloor S \rfloor},~\forall \mu, \mu' \in [0, 1)^q,
$$
where $\lfloor S \rfloor$ is the largest integer below $S$.
%
\end{assumption}

A linear factor model is a special case where $ g(\lambda_i,\mu_j) = \lambda_i ^T \mu_j$, satisfying Assumption~\ref{assumption:factor_model} for all $S \in \mathbb{N}$ and some $C_H = C<\infty$. 
Assumption~\ref{assumption:factor_model} also allows for smooth nonlinear factor models, and it implies joint control over $(r,\Delta_E)$ as desired. Intuitively, as latent dimension $q$ increases, the rank $r$ increases. As smoothness $S$ increases, the approximation error $\Delta_E$ decreases. Our final result demonstrates that, as long as the ratio $q/S$ is small enough, the data cleaning adjusted confidence intervals are valid.

\begin{remark}
Our results hold for a broad class of dictionaries, with the concise notation $q'$ in Corollary~\ref{cor:GFM}. Appendix~\ref{sec:factor} proves that
$q'\leq d_{\max} q$, where $d_{\max}$ is the degree of the polynomial dictionary. For the interacted dictionary, $d_{\max}=2$.
\end{remark}

\begin{corollary}\label{cor:GFM}
Suppose the conditions of Theorems~\ref{theorem:fast_rate},~\ref{theorem:fast_rate_RR},~\ref{thm:dml_inid} and~\ref{thm:var_inid} hold, as well as Assumption~\ref{assumption:factor_model}. For simplicity, consider the semiparametric case where $\sigma,\bar{\sigma},\bar{\alpha},\bar{\alpha}',\bar{Q}$ are bounded above and $\bar{q}=1$. Suppose in addition
(i) moment regularity: $\{(\xi/\sigma)^3+\chi^2\}n^{-\frac{1}{2}}\rightarrow0$;
(ii) weak dependence: $(K_a,\kappa,\bar{K},\rho_{\min}^{-1})$ scale polynomially in $\ln(np)$;
(iii) nonlinear sparse approximation: $m\Delta_{\phi}\leq \|\beta^*\|_1^2<\infty$ and $m\Delta_{\zeta}\leq \|\eta^*\|_1^2<\infty$;
(iv) enough repeated measurements: $n^{\frac{2}{3}}\lesssim p \lesssim n^{\frac{3}{2}}$, i.e. $n=p^{\upsilon}$ or $p=n^{\upsilon}$ for $\upsilon \in [1,\frac{3}{2}]$;
(v) small latent dimension to smoothness ratio: $\frac{q'}{S}<\frac{3}{4}-\frac{\upsilon}{2}$. Then the conclusions of Corollary~\ref{cor:CI_inid} hold.
\end{corollary}
In summary, we allow either $n>p$ or $p>n$ as long as $(n,p)$ increase at similar rates. Given $(n,p)$, the ratio of the latent dimension $q$ over smoothness $S$ in the generalized factor model must be sufficiently low. For example, if $n=p$ and $q'=q$ then we require $q<\frac{S}{4}$: the latent dimension must be less than a quarter of the smoothness. A sufficiently low $\frac{q}{S}$ ratio ensures sufficiently fast learning rates $\Rc(\hgamma)$ and $\Rc(\halpha)$ for causal inference with standard error $\hsigma n^{-1/2}$. For the special case of a linear factor model, the $\frac{q}{S}$ ratio constraint becomes vacuous, and there is no restriction on the latent dimension $q$. The same is true for a polynomial factor model where $g(\lambda_i,\mu_j)=\text{polynomial}(\lambda_i, \mu_j)$.  The doubly robust framework allows us to slightly relax the conditions stated above and still obtain consistent estimation for $\theta_0$: either $\Delta_{\phi}\not\rightarrow 0$ or $\Delta_{\zeta}\not\rightarrow 0$, i.e. $\gamma_0$ or $\alpha_0$ may be incorrectly specified.

%% file: 6_application.tex
\section{Case study: Effect of import competition using Census data}\label{sec:application}

\textbf{Can we recover the same effects with data corruption}? Equipped with theoretical guarantees, we return to the motivating real world issue: measurement error, missing values, discretization, and differential privacy in US Census data. We replicate a seminal paper in labor economics \cite{autor2013china} that uses Census-derived data 
to ask: what is the effect of import competition on local labor markets in the US? We ask an additional question: can we recover the same effects after introducing various types and levels of synthetic corruption? In particular, we implement differential privacy at a level calibrated to the 2020 Census. Our empirical results represent a realistic use case of Census-derived data, yet an idealized data setting where the corruptions belong to our class. The causal parameter is the partially linear instrumental variable regression parameter described in Appendix~\ref{sec:examples}.

\cite{autor2013china} use Census data at the commuting zone (CZ) level. A CZ is an aggregate unit interpretable as a local economy, and 722 CZs make up the mainland US. CZ data are constructed from individual microdata published by the US government. The outcome $Y_i$ is percent change in US manufacturing employment; the treatment $D_i$ is percent change in imports from China; the instrument $U_i$ is percent change in imports from China to other countries; and the covariates $X_{i,\cdot}$ are CZ characteristics. In the augmented specification, the covariates $X_{i,\cdot}\in \R^{30}$ include approximate repeated measurements such as average disability, medical, and unemployment benefits, and appear to be approximately low rank in Figure~\ref{fig:scree2}.

Figures~\ref{fig:semi_noise},~\ref{fig:semi_missing}, and~\ref{fig:semi_discrete} present our initial semi-synthetic exercises. For reference, we visualize in red the 2SLS point estimate and confidence interval of \cite{autor2013china}, using clean data. Immediately next to \cite{autor2013china}'s results, we visualize our own point estimate and confidence interval with clean data. We recover essentially the same point estimate and a somewhat smaller confidence interval. The true covariates are approximately low rank, our procedure exploits this fact, and therefore it has an advantage. Subsequently, we implement our procedure with increasing levels of measurement error: 20\%, 40\%, 60\%, 80\%, and 100\% noise-to-signal ratio. Our point estimates remain stable, and the confidence intervals subtly increase in length. We obtain similar results with missing values and discretization: point estimates remain stable and the confidence intervals adaptively increase in length for higher noise-to-signal ratios, similar to Figure~\ref{fig:PCR_all}. Appendix~\ref{sec:sim} confirms similar results when standardizing the true covariates before the semi-synthetic exercises.

\begin{figure}[H]
\begin{centering}
    \begin{subfigure}[b]{0.48\textwidth}
        \centering
        \includegraphics[width=0.8\textwidth]{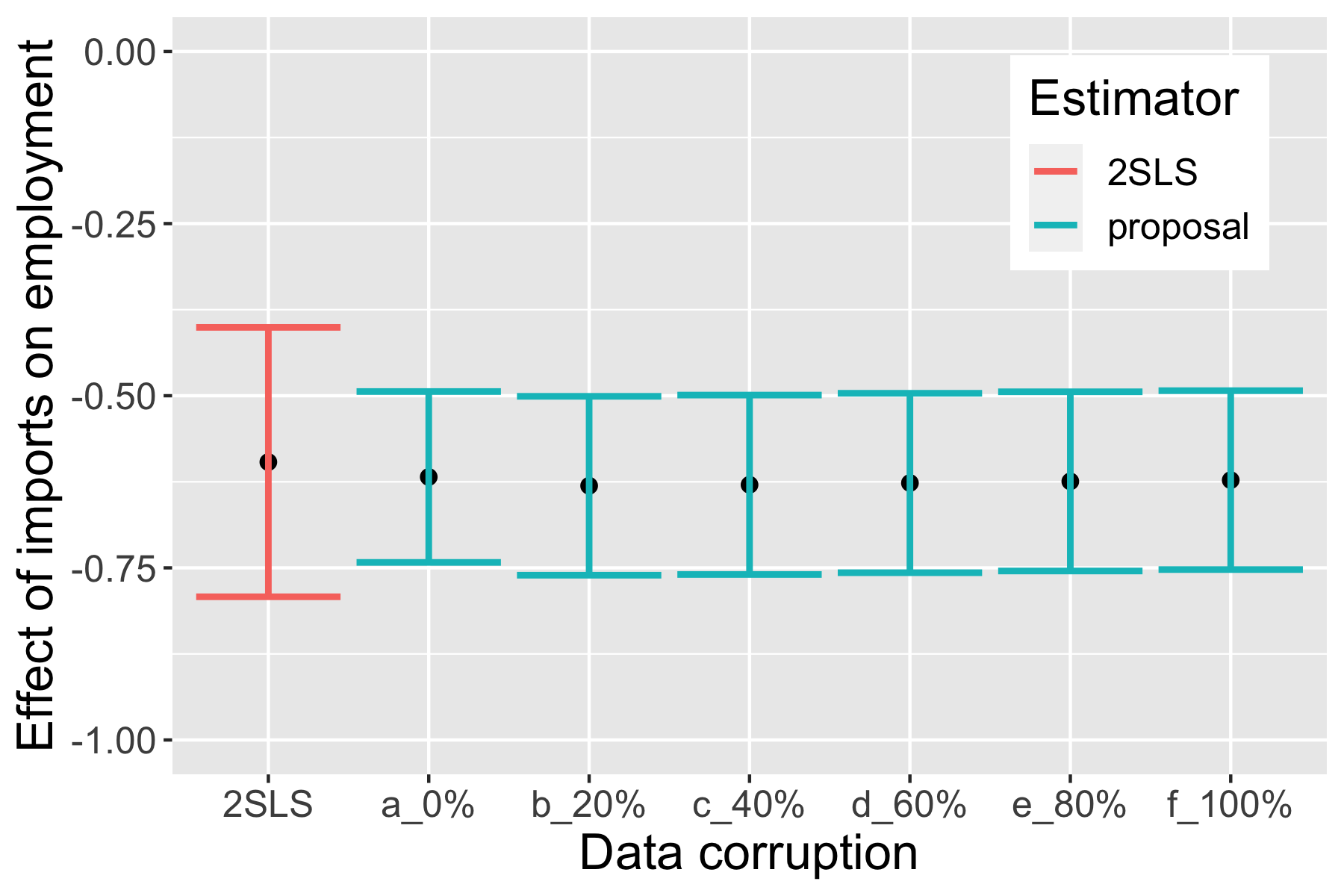}
        \vspace{-3pt}
        \caption{\label{fig:semi_noise} Measurement error}
    \end{subfigure}
    \begin{subfigure}[b]{0.48\textwidth}
        \centering
        \includegraphics[width=0.8\textwidth]{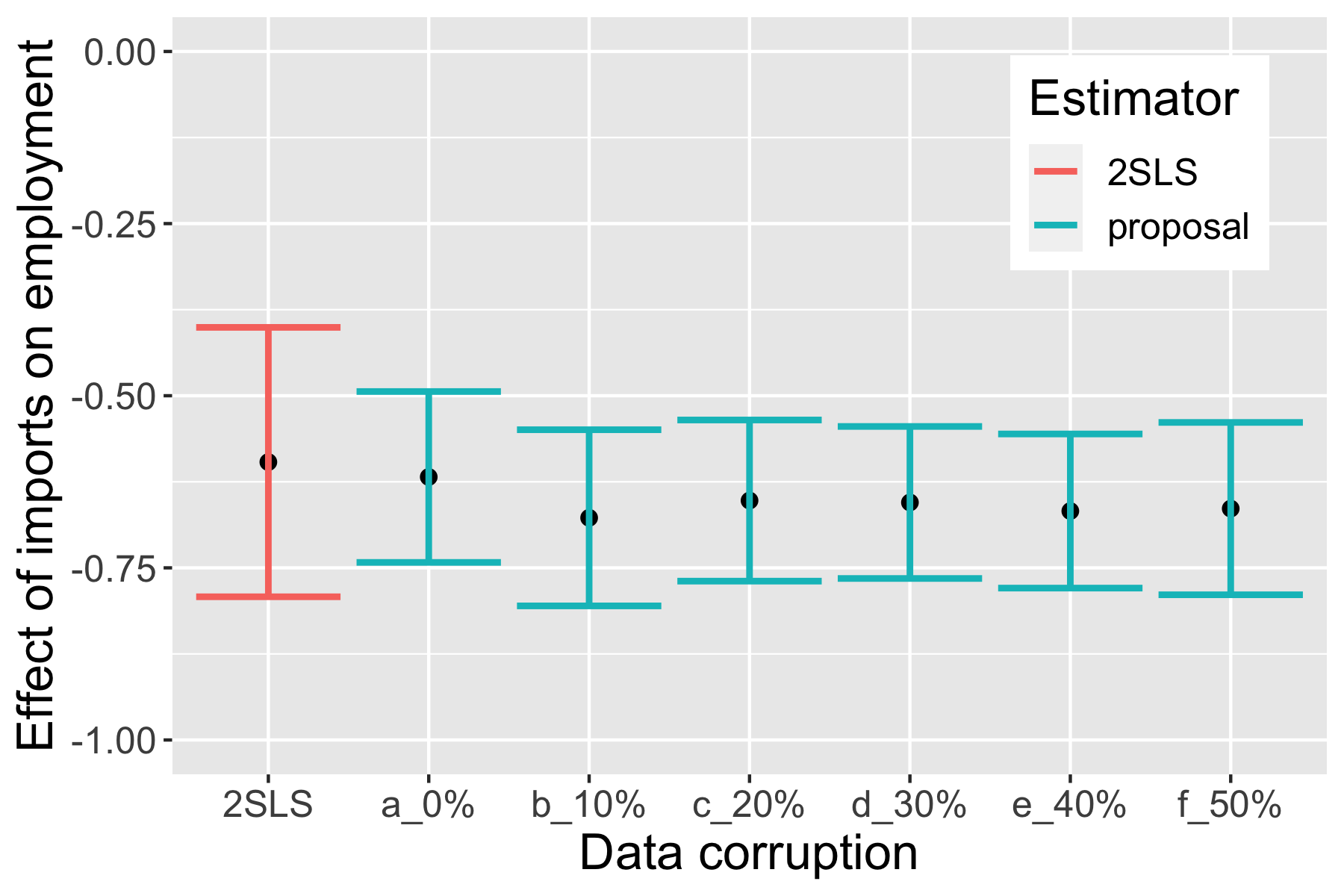}
        \vspace{-3pt}
        \caption{\label{fig:semi_missing} Missing values}
    \end{subfigure}
    
    \begin{subfigure}[b]{0.48\textwidth}
        \centering
        \includegraphics[width=0.8\textwidth]{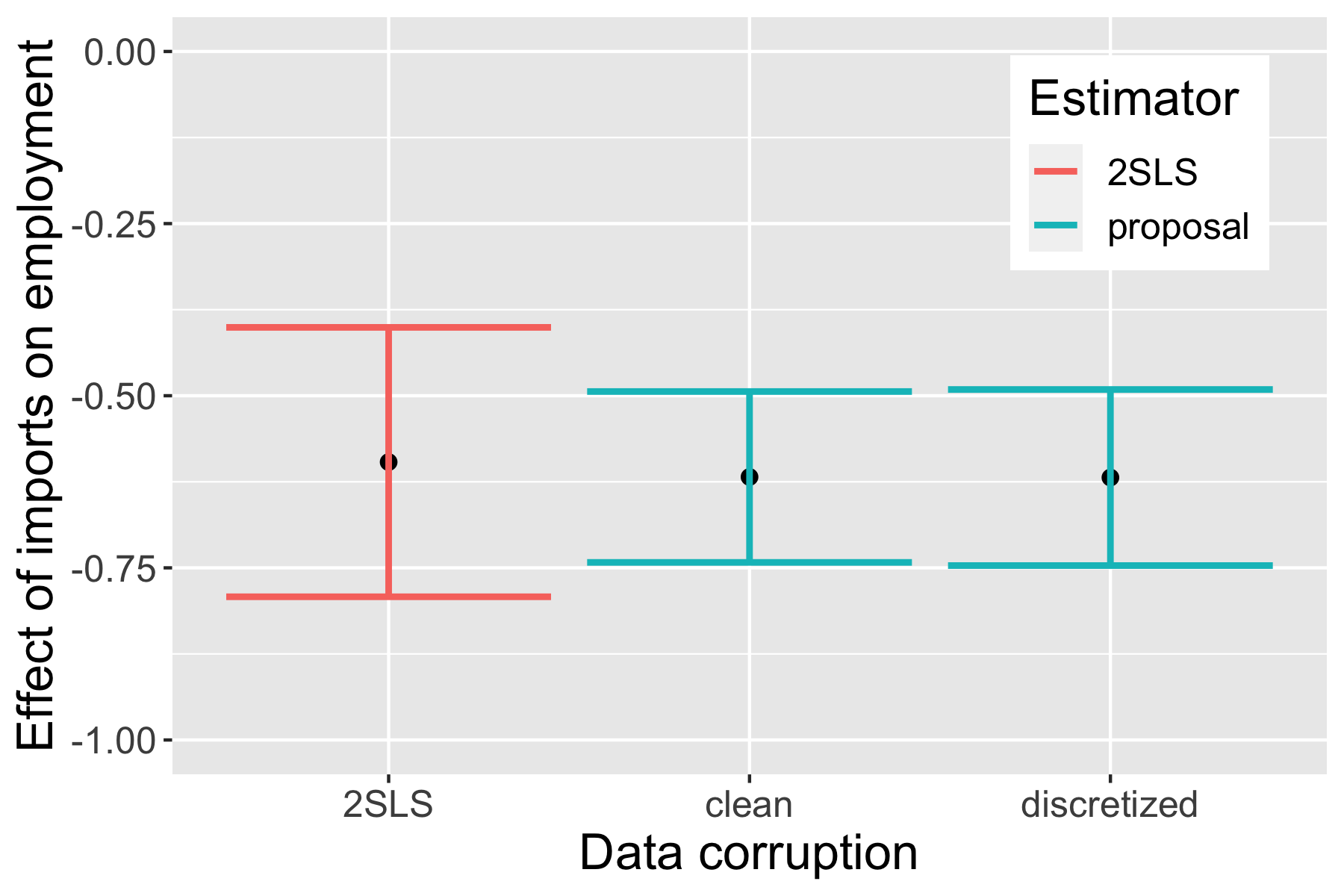}
        \vspace{-3pt}
        \caption{\label{fig:semi_discrete} Discretization}
    \end{subfigure}
    \begin{subfigure}[b]{0.48\textwidth}
        \centering
        \includegraphics[width=0.8\textwidth]{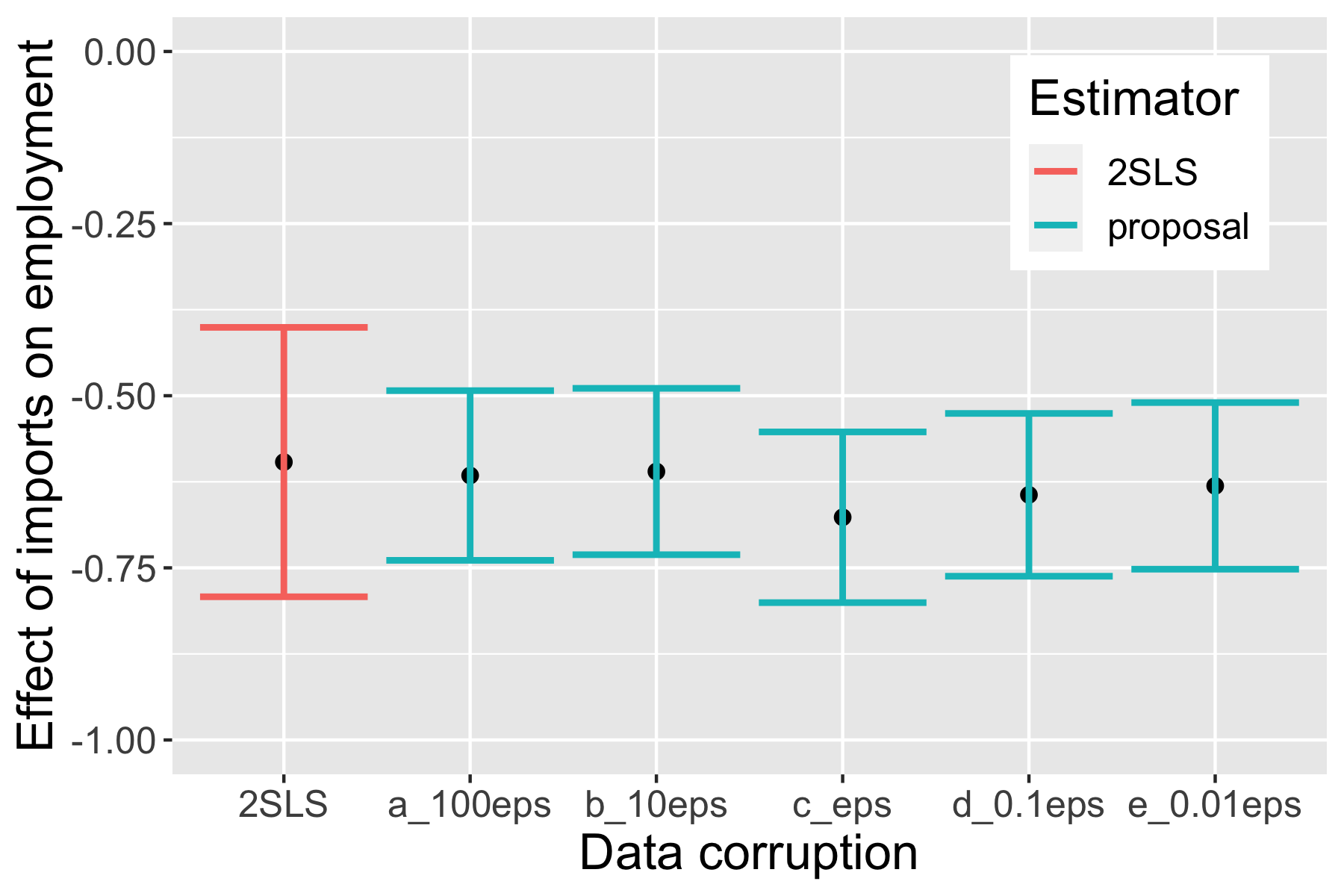}
        \vspace{-3pt}
        \caption{\label{fig:semi_calibrate} Differential privacy (calibrated)}
    \end{subfigure}
\vspace{-10pt}
\caption{Synthetic corruption}\label{fig:semi_all}
\vspace{-15pt}
\end{centering}
\end{figure}

\textbf{Formalizing privacy}. Next, we calibrate our semi-synthetic exercise to privacy levels mandated by the US Census Bureau. To do so, we clarify how our model of causal inference with corrupted data accommodates differential privacy mechanisms. With these formal results, we calibrate the variance of the Laplacian noise appropriately. In what follows, we focus on a one-off data release (formally called the non-interactive setting). 

We maintain the following thought experiment: we are the Census Bureau, and our goal is to publish \cite{autor2013china}'s CZ-level aggregated data set while protecting the privacy of individuals within CZs. In particular, we have access to the individual-level microdata, which we will \textit{not} publicly share; we will only publish the CZ-level summaries for aggregate units. Consider a particular commuting zone $i\in [n]$ with $L_i$ individuals, and denote its individual-level microdata by $\bM^{(i)} \in \R^{L_i\times p}$. We wish to publish $p$ summary statistics $X_{i,\cdot}$ for this CZ, where $X_{ij}=\frac{1}{L_i}\sum_{\ell=1}^{L_i} M^{(i)}_{\ell j}$, however we wish to maintain plausible deniability that each individual $\ell\in [L_i]$ contributed their data. The simulated attack on the 2010 Census found that Census block summary tables did not maintain this plausible deniability. 

\begin{definition}[Differential privacy for summary tables]\label{def:privacy}
A randomized mechanism $\mathcal{M}$ confers differential privacy with privacy loss $\epsilon$ if and only if for any two possible individual-level data sets $\bM$ and $\bM'$ differing in a single row, and for all events $E$ in the range of $\mathcal{M}$,
$
\frac{\p(\mathcal{M}(\bM)\in E )}{\p(\mathcal{M}(\bM')\in E)}\leq e^{\epsilon}
$
where the randomness is with respect to $\mathcal{M}$.
\end{definition}

The canonical mechanism that achieves differential privacy is to publish $\mathcal{M}(\bM^{(i)})=X_{i,\cdot}+H_{i,\cdot}$ instead of $X_{i,\cdot}$, where $H_{i,\cdot}$ is Laplacian noise with an appropriately calibrated variance.\footnote{More precisely, $\mathcal{M}_i:\R^{L_i\times p}\rightarrow \R^p $.}
In addition to the Laplace mechanism, the
discrete Gaussian, piece wise uniform, and bounded mechanisms 
induce measurement error that is subexponential and conditionally mean zero, which fits within our framework. 
For simplicity, we focus on the Laplace mechanism when relating privacy to our theoretical results.\footnote{Bureau policies mandate ``zero concentrated'' differential privacy, which is a closely related privacy concept for which our implementation suffices. See Appendix~\ref{sec:privacy} for details.}

\begin{proposition}[Strong protections for aggregate data]\label{prop:macro}
Suppose (i) each entry of microdata is bounded, i.e. $|M^{(i)}_{\ell j}|\leq\bar{A}_i$;
 (ii) no individual appears in two commuting zones.
Then the mechanism $Z_{ij}=X_{ij}+H_{ij}$ where $H_{ij}\overset{i.i.d.}{\sim}\text{Laplace}\left(\frac{2\bar{A}_i}{\epsilon} \frac{p}{L_i}\right)$ confers $\epsilon$  differential privacy and the measurement error parameters satisfy
$
K_a,\kappa \leq \max_{i\in [n]}\frac{2^{3/2} \bar{A}_i}{\epsilon} \frac{p}{L_i}.
$ This privacy guarantee is immune to data cleaning.
\end{proposition}

\begin{corollary}[Safety in numbers]\label{cor:macro}
Suppose the conditions of Proposition~\ref{prop:macro} hold and $\max_{i\in [n]}\frac{p}{L_i}  \lesssim \ln(np)$. Then the measurement error parameters satisfy $K_a,\kappa\lesssim \ln(np)$, and therefore our rates of data cleaning and error-in-variable estimation translate into data cleaning adjusted confidence intervals. 
\end{corollary}

In summary, the calibrated Laplacian variance depends on the privacy loss $\epsilon$, the most extreme true value $\bar{A}_i$, the number of covariates $p$, and number of individuals $L_i$ per aggregate unit. 
The auxiliary condition $\frac{p}{L_i}  \lesssim \ln(np)$ is a practical diagnostic: roughly speaking, the number of published covariates should not greatly exceed the number of individuals per aggregate unit. It sheds light on limitations because it is plausible for CZs, but implausible for Census blocks. Much empirical economic research studies CZs, which we study in our semi-synthetic exercise. Future research may empirically investigate, through simulated attacks, how vulnerable various data releases may be for different $\frac{p}{L_i}$ regimes.

Figure~\ref{fig:semi_calibrate} implements  differential privacy for \cite{autor2013china}'s CZ-level aggregated data set while protecting the privacy of individuals within CZs. We calibrate the Laplacian variance according to Proposition~\ref{prop:macro}, where $\epsilon$ is based on Bureau memos, $p=30$ in the augmented specification, and $(\bar{A}_i,L_i)$ are calculated from the microdata for each CZ; see Appendix~\ref{sec:privacy} for details.  
To study the robustness of our results to the privacy loss parameter, we consider $(100\epsilon,10\epsilon, \epsilon, 0.1\epsilon, 0.01\epsilon)$, which corresponds to privacy below and above the mandated level. Across levels, our point estimates and confidence intervals remain stable.

%% file: 7_conclusion.tex
\section{Conclusion}\label{sec:conclusion}

Recent developments in how the US Census Bureau publishes economic data motivate us to study a class of corruptions that is rich enough to encompass classical types of corruption, such as measurement error and missingness, as well as modern types, such as discretization and differential privacy mechanisms. Abstractly, our goal is to learn parameters in nonlinear, heterogeneous causal models from corrupted data; concretely, our goal is to characterize some scenarios in which it is possible to achieve both privacy and precision. To do so, we propose new data cleaning-adjusted confidence intervals that are computationally simple, statistically rigorous, and empirically robust in settings calibrated to empirical economic research. We build a framework to use matrix completion as data cleaning for downstream causal inference, bridging two rich literatures. Future work may extend our results to confounded noise and sample selection bias.

%% file: A_matrix.tex
\section{Data cleaning}\label{sec:cleaning}

In this appendix, we replace the symbol $X_{i,\cdot}$ with the symbol $A_{i,\cdot}$.
We suppress indexing by the folds $(\tr,\te)$ to lighten notation. As in Assumption \ref{assumption:missing}, we identify $\NA$ with $0$ in $\bZ$ hereafter. We slightly abuse notation by letting $n$ be the number of observations in $\tr$, departing from the notation of the main text. The entire section is conditional on $\bA$, which we omit. We write $\|\cdot\|=\|\cdot\|_{op}$ and let $C$ be an absolute constant.

Recall $\bA = \bA^{\lr} + \bE^{\lr}$ and $r = rank\{\bA^{\lr}\}$. We denote the SVDs
$
\bA^{\lr}=\bU\bS\bV^T$, $\bhA=\bhU_k\bhS_k\bhV_k^T$, and $\bZ\bhrho^{-1}=\bhU\bhS\bhV^T.
$
The first $k$ left singular vectors of $\bA^{\lr}$ are $\bU_k$. We denote $s_k=\Sigma_{kk}$ and $\hat{s}_k=\hat{\Sigma}_{kk}$.

Recall that $\bhA^{\tr}$ is constructed by taking $\tr$ covariates then filling and cleaning them using $\tr$ alone. As a theoretical device we also study $\bhA^{\te}$, obtained by taking $\te$ covariates, filling them using $\tr$, and cleaning them using $\te$. The analysis does not depend on whether $\bhrho^{\tr}$ or $\bhrho^{\te}$ is used when filling in missing values. 

Consider a matrix $\bB \in \mathbb{R}^{n \times p}$ with SVD $\bB = \sum_{i=1}^{n \wedge p} \sigma_i u_i v_i^T$. 
We define the linear function $\varphi^{\bB}_{\lambda}: \mathbb{R}^{n} \to \mathbb{R}^{n}$ as
$
	\varphi^{\bB}_{\lambda}(w) = \sum_{i = 1}^{n \wedge p} 1(\sigma_i \ge \lambda) u_i u_i^T w.
$
We use the shorthand $\varphi^{\bB} = \varphi_{0}^{\bB}$.


Define the events: 
$$
\Ec_1=\left\{ \|\bZ -  \bA\brho\| \leq (\sqrt{n} + \sqrt{p})\Delta_{H,op}   \right\},\quad
\Ec_2=\left\{ \max_{j \in [p]} \| Z_{\cdot, j} - \rho_j A_{\cdot, j}  \|_2^2 \leq  n \Delta_H  \right\},
$$
$$
\Ec_3=\left\{ \max_{j \in [p]}  \| \bU_k\bU^T_{k}(Z_{\cdot,j}-\rho_jA_{\cdot,j})  \|_2^2 \leq k\Delta_H  \right\},\quad
\Ec_4=\left\{\forall j\in[p],\; \frac{1}{\delta}\rho_j \le \hat{\rho}_j \le \delta \rho_j \right\},
$$
$$
\Ec_5=\left\{ \max_{j \in [p]}  |\hrho_j-\rho_j| \leq C \sqrt{\frac{\ln(np)}{n}} \right\},
$$
where
$$
\Delta_{H,op} = C \bar{A} (\kappa +  K_a+\bar{K} )\ln^{\frac{3}{2}}(np),\quad
\Delta_{H} =C (K_a + \bar{A} \bar{K})^2 \ln^{2}(np),\quad
\delta = \frac{1}{1 - \sqrt{\frac{22 \ln (np)}{ n \rho_{\min}}}}.
$$
The Online Appendix shows that under Assumptions~\ref{assumption:bounded},~\ref{assumption:measurement}, and~\ref{assumption:missing},
$
\p(\Ec^c)\leq  \frac{10}{n^{10}p^{10}},
$
where $\Ec := \cap_{k=1}^5 \Ec_{k}$.

\begin{lemma}\label{lem:operator_norm_noise_new_bound}
Set $k=r$. Then,
{ \small
\begin{align*}
    \|\bZ\hat{\pmb{\rho}}^{-1}-\bA^{\lr}\|  ~\Big|~ \{\Ec_1,\Ec_4, \Ec_5\}  &\le C \frac{\delta}{\rho_{\min}} \left(  (\sqrt{n}+\sqrt{p})\Delta_{H,op} + \|\bE^{\lr}\| + \sqrt{\frac{\ln(np)}{n}} \| \bA^{\lr} \|\right);\\
    \| \bU\bU^T - \bhU_r\bhU_r^T \| ~\Big|~ \{\Ec_1, \Ec_4, \Ec_5\}
 &\le C \frac{\delta}{\rho_{\min} s_r} \left( (\sqrt{n}+\sqrt{p}) \Delta_{H,op} + \|\bE^{\lr}\| + \sqrt{\frac{\ln(np)}{n}}  \| \bA^{\lr} \|\right); \\
\| \bV\bV^T - \bhV_{r} \bhV_{r}^T \| ~\Big|~ \{\Ec_1, \Ec_4, \Ec_5\} 
&\le C \frac{\delta}{\rho_{\min} s_r} \left\{  (\sqrt{n}+\sqrt{p})\Delta_{H,op} + \|\bE^{\lr}\| + \sqrt{\frac{\ln(np)}{n}}  \| \bA^{\lr} \|\right\}; \\
| s_r - \hat{s}_r | ~\Big|~ \{\Ec_1, \Ec_4, \Ec_5\}
&\le  C \frac{\delta}{\rho_{\min}} \left\{  (\sqrt{n}+\sqrt{p})\Delta_{H,op} + \|\bE^{\lr}\| + \sqrt{\frac{\ln(np)}{n}}  \| \bA^{\lr} \|\right\}.
\end{align*}
}
\end{lemma}

\begin{proof}
To begin, write
$
\|\bZ\hat{\pmb{\rho}}^{-1}-\bA^{\lr}\| 
\le {\| \hat{\pmb{\rho}}^{-1} \| \|\bZ-\bA^{\lr} \hat{\pmb{\rho}} \|}
= \frac{\|\bZ-\bA^{\lr}\hat{\pmb{\rho}}\|}{\min_j\hat{\rho}_j}.
$
By triangle inequality
$
\|\bZ-\bA^{\lr}\hat{\pmb{\rho}}\| 
\le \|\bZ-\bA^{\lr}\pmb{\rho}\| + \|\bA^{\lr}\|\|\pmb{\rho}-\hat{\pmb{\rho}}\|. 
$
Applying $\Ec_1$ to the first term, $\Ec_5$ to the second term, and $\Ec_4$ to the denominator yields the first result. From the first result, Wedin's $\sin \Theta$ Theorem yields the second and third results while Weyl's inequality yields the fourth result.  
\end{proof}

\begin{lemma}[Eq. 43 of \cite{agarwal2021robustness}]\label{lemma:column_representation}
Take $\lambda^*=\hat{s}_k$. Then
$
\hat{A}_{\cdot, j} = \frac{1}{\hrho_j} \varphi^{\bZ \hat{\pmb{\rho}}^{-1}}_{\lambda^*}(Z_{\cdot, j}). 
$
\end{lemma}

\begin{lemma}\label{lem:combined_JASA48}
   Suppose $k=r$, Assumptions~\ref{assumption:bounded} and~\ref{assumption:well_balanced_spectrum} hold, and $\rho_{\min}>\frac{23\ln (np)}{n}$. Then
   \begin{align*}
    \norm{\bhA - \bA}^2_{2, \infty}  ~\Big|~ \Ec
     &\leq 
    C\bar{A}^4(K_a + \bar{K})^2 (\kappa + \bar{K} + K_a)^2
    \cdot \frac{r\ln^5(np)}{\rho_{\min}^4}
    \left(
    1
    +\frac{n}{p}
    +n\Delta_E^2
    \right). 
\end{align*}
\end{lemma}

\begin{proof}
Fix a column index $j \in [p]$. Observe that
$
	\hat{A}_{\cdot, j} - A_{\cdot, j}
		= \Big\{ \hat{A}_{\cdot, j} - \varphi^{\bZ \hat{\pmb{\rho}}^{-1}}_{\lambda^*} \big( A_{\cdot, j} \big) \Big\} + \Big\{ \varphi^{\bZ \hat{\pmb{\rho}}^{-1}}_{\lambda^*} \big( A_{\cdot, j} \big) - A_{\cdot, j} \Big\}.
$
Recall that $\varphi^{\bZ \hat{\pmb{\rho}}^{-1}}_{\lambda^*}$ projects onto the span of the top $r$ left singular vectors $\{\hat{u}_1,...,\hat{u}_r\}$ of $\bZ \bhrho^{-1}$, which are are also the top $r$ left singular vectors of $\bZ$ since $\bhrho^{-1}$ is diagonal.
Hence $
\varphi^{\bZ \hat{\pmb{\rho}}^{-1}}_{\lambda^*} (A_{\cdot, j}) - A_{\cdot, j} \in span\{\hat{u}_1, \ldots, \hat{u}_r \}^{\perp}.
$
By Lemma \ref{lemma:column_representation},
$
\hat{A}_{\cdot, j} -  \varphi^{\bZ \hat{\pmb{\rho}}^{-1}}_{\lambda^*} (A_{\cdot, j}) = \frac{1}{\hrho_j}\varphi^{\bZ \hat{\pmb{\rho}}^{-1}}_{\lambda^*} (Z_{\cdot, j}) - \varphi^{\bZ \hat{\pmb{\rho}}^{-1}}_{\lambda^*} (A_{\cdot, j}) \in span\{\hat{u}_1, \ldots, \hat{u}_r \}.
$
Therefore
$
	\Big\| \hat{A}_{\cdot, j} - A_{\cdot, j} \Big\|_2^2
		= \Big\| \hat{A}_{\cdot, j} - \varphi^{\bZ \hat{\pmb{\rho}}^{-1}}_{\lambda^*} \big( A_{\cdot, j} \big) \Big\|_2^2		
			+ \Big\| \varphi^{\bZ \hat{\pmb{\rho}}^{-1}}_{\lambda^*} \big( A_{\cdot, j} \big) - A_{\cdot, j} \Big\|_2^2.
$
    Again applying Lemma \ref{lemma:column_representation}, 
\begin{align*}
&\norm{\hat{A}_{\cdot, j} - \varphi^{\bZ \hat{\pmb{\rho}}^{-1}}_{\lambda^*}(A_{\cdot, j})}_2^2 
	= \norm{\frac{1}{\hrho_j}  \varphi^{\bZ \hat{\pmb{\rho}}^{-1}}_{\lambda^*} (Z_{\cdot, j} - \rho_j A_{\cdot, j} ) 
		+ \frac{\rho_j - \hrho_j}{\hrho_j} \varphi^{\bZ \hat{\pmb{\rho}}^{-1}}_{\lambda^*}( A_{\cdot, j}) }_2^2		\\
	&\leq 2 \, \norm{\frac{1}{\hrho_j}  \varphi^{\bZ \hat{\pmb{\rho}}^{-1}}_{\lambda^*} (Z_{\cdot, j} - \rho_j A_{\cdot, j} ) }_2^2 
		+ 2 \, \norm{ \frac{\rho_j - \hrho_j}{\hrho_j} \varphi^{\bZ \hat{\pmb{\rho}}^{-1}}_{\lambda^*}( A_{\cdot, j} )}_2^2		\\
	&\leq \frac{2\delta^2}{\rho_j^2} \norm{\varphi^{\bZ \hat{\pmb{\rho}}^{-1}}_{\lambda^*}(Z_{\cdot, j} - \rho_j A_{\cdot, j})}_2^2
		+ C\frac{\delta^2}{\rho_j^2} \frac{\ln(np)}{n} \| A_{\cdot, j} \|_2^2
\end{align*}
where the last line uses $\Ec_4$ and $\Ec_5$. Note that
$
\norm{\varphi^{\bZ \hat{\pmb{\rho}}^{-1}}_{\lambda^*}(Z_{\cdot, j} - \rho_j A_{\cdot, j})}_2^2  \le 2  \,  \Big \| \varphi^{\bZ \hat{\pmb{\rho}}^{-1}}_{\lambda^*}(Z_{\cdot, j} - \rho_j A_{\cdot, j}) 
	 	- \varphi^{\bA^{\lr}}(Z_{\cdot, j} - \rho_j A_{\cdot, j})  \Big \|_2^2 +2 \,  \Big \| \varphi^{\bA^{\lr}}(Z_{\cdot, j} - \rho_j A_{\cdot, j})  \Big \|_2^2.
$
Since $k  = r$, $\varphi^{\bZ \hat{\pmb{\rho}}^{-1}}_{\lambda^*}(w) = \bhU_r\bhU_r^T w$ and $\varphi^{\bA^{\lr}}(w) = \bU\bU^T w$ for $w \in \mathbb{R}^n$. By Lemma~\ref{lem:operator_norm_noise_new_bound},
\begin{align*} 
	 &\Big \| \varphi^{\bZ \hat{\pmb{\rho}}^{-1}}_{\lambda^*}(Z_{\cdot, j} - \rho_j A_{\cdot, j}) 
	 	- \varphi^{\bA^{\lr}}(Z_{\cdot, j} - \rho_j A_{\cdot, j})  \Big \|_2
	 \le \| \bU\bU^T - \bhU_r\bhU_r^T \|
		\big\| Z_{\cdot, j} - \rho_j A_{\cdot, j} \big\|_2	\nonumber
	\\&\le C \frac{\delta}{\rho_{\min} s_r} \left(  (\sqrt{n}+\sqrt{p})\Delta_{H,op} + \|\bE^{\lr}\| + \sqrt{\frac{\ln(np)}{n}}  \| \bA^{\lr} \|\right)	\big\| Z_{\cdot, j} - \rho_j A_{\cdot, j} \big\|_2.	
\end{align*}
Combining the inequalities together, we have $\Big\| \hat{A}_{\cdot, j} - \varphi^{\bZ \hat{\pmb{\rho}}^{-1}}_{\lambda^*} \big( A_{\cdot, j} \big) \Big\|_2^2$ is bounded by
\begin{align*}
	&\frac{C\delta^4}{\rho_{\min}^2} 
		\left(  \frac{(\sqrt{n}+\sqrt{p})\Delta_{H,op}}{ \rho_{\min} s_r} +  \frac{\|  \bE^{\lr} \|}{\rho_{\min} s_r} +  \frac{  \sqrt{\ln(np) / n}\|\bA^{\lr} \|}{\rho_{\min} s_r} \right)^2
				\big\| Z_{\cdot, j} - \rho_j A_{\cdot, j} \big\|_2^2	\nonumber\\
		&\quad	+ \frac{4\delta^2}{\rho_{\min}^2}  \Big \| \varphi^{\bA^{\lr}}(Z_{\cdot, j} - \rho_j A_{\cdot, j})  \Big \|_2^2
				+ C\frac{\delta^2}{\rho_{\min}^2} \frac{\ln(np)}{n} \| A_{\cdot, j} \|_2^2.		\label{eq:term_step2}
\end{align*}
Since
$\bA = \bA^{\lr} + \bE^{\lr}$,
\begin{align*}
	&\norm{\varphi^{\bZ \hat{\pmb{\rho}}^{-1}}_{\lambda^*} \big( A_{\cdot, j} \big) - A_{\cdot, j} }_2^2 \leq 2 \, \norm{\varphi^{\bZ \hat{\pmb{\rho}}^{-1}}_{\lambda^*} \big( A^{\lr}_{\cdot, j} \big) - A^{\lr}_{\cdot, j}  }_2^2 +  2 \,\norm{ \varphi^{\bZ \hat{\pmb{\rho}}^{-1}}_{\lambda^*} \big( E^{\lr}_{\cdot, j} \big) - E^{\lr}_{\cdot, j}}_2^2		\nonumber\\
		&\leq 2 \, \| \bU\bU^T - \bhU_r\bhU_r^T \|^2 \norm{ A^{\lr}_{\cdot, j} }_2^2 + 2 \,\norm{ E^{\lr}_{\cdot, j} }_2^2 \nonumber
		\\ &\le   	C \delta^2 \left(  \frac{(\sqrt{n}+\sqrt{p}) \Delta_{H,op}}{ \rho_{\min} s_r} +  \frac{\|  \bE^{\lr} \|}{\rho_{\min} s_r} +  \frac{  \sqrt{\ln(np) / n}\|\bA^{\lr} \|}{\rho_{\min} s_r} \right)^2
		 \norm{ A^{\lr}_{\cdot, j} }_2^2 + 2\, \norm{ E^{\lr}_{\cdot, j} }_2^2,
\end{align*}
where the final inequality appeals to Lemma \ref{lem:operator_norm_noise_new_bound}.
To conclude, substitute the bounds on each term, appeal to $\Ec_2$ and $\Ec_3$, then simplify using Assumptions~\ref{assumption:bounded} and~\ref{assumption:well_balanced_spectrum}.
\end{proof}

\begin{lemma}\label{lem:adverse_2_infty}
Suppose Assumptions~\ref{assumption:bounded},~\ref{assumption:measurement}, and~\ref{assumption:missing} hold. Then
$
\E\left[\|\bhA-\bA\|^2_{2,\infty} \1\{\Ec^c\}\right]\leq \Delta_{adv}\frac{1}{n^{2}p^{5}}$, where $\Delta_{adv}:=C\Big\{\bar{A}^2+ K_a^2 \ln^{2}(np)\Big\}.
$
\end{lemma}

\begin{proof}

By Cauchy-Schwarz, $\E\left[\|\bhA-\bA\|^2_{2,\infty} \1\{\Ec^c\} \right]
    \leq \sqrt{\E\left[\|\bhA-\bA\|^4_{2,\infty}\right]} \sqrt{\E\left[\1^2\{\Ec^c\}\right]}.$ \\
Within the first factor,
$
    \max_{j \in [p]} \| \hat{A}_{\cdot, j} \|_2
    \leq \max_{j \in [p]} \frac{1}{\hrho_j} \| Z_{\cdot, j} \|_2  
    \leq n \cdot \sqrt{n}(\bar{A}+\max_{i,j}|H_{ij}|),
$
so
$$
    \E\left[\|\bhA-\bA\|^4_{2,\infty}\right]
    \leq \E[\{n^{\frac{3}{2}}(\bar{A}+\max_{i,j}|H_{ij}|)+\sqrt{n}\bar{A} \}^4 ] \leq C n^6 \{\bar{A}^4+ K_a^4 \ln^{4}(np)\}.
$$
The final inequality holds since
$
\E[\max_{i,j}|H_{ij}|^4]\leq C K_a^4 \ln^{\frac{4}{a}}(np).
$ Since $\p(\Ec^c)\leq \frac{C}{n^{10}p^{10}}$, we conclude that
$$
    \E\left[\|\bhA-\bA\|^2_{2,\infty} \1\{\Ec^c\} \right]
    \leq C\sqrt{n^6 (\bar{A}^4+ K_a^4 \ln^{4}(np))} \sqrt{\frac{1}{n^{10}p^{10}}}.
$$
\end{proof}

\begin{proof}[Proof of Theorem~\ref{theorem:cov}]
The result follows from Lemmas \ref{lem:combined_JASA48} and \ref{lem:adverse_2_infty}. By the law of iterated expectations, results conditional on $\bA$ imply the same unconditional on $\bA$.
\end{proof}

%% file: B_regression.tex
\section{Error-in-variable regression}\label{sec:regression}

We write the proofs without nonlinear dictionaries for clarity. 

Recall that the $\fil$\, operator rescales using $\bhrho$ calculated from \tr. Denote the SVDs
$
\bA^{\lr,\tr}=\bU\bS\bV^T$, $\fil(\bZ^{\tr})=\bZ^{\tr} \bhrho^{-1}=\bhU\bhS\bhV^T$, $\bhA^{\tr}=\bhU_k\bhS_k\bhV_k^T.
$
In this notation, $\bV$ is an orthonormal basis for $\row\{\bA^{\lr,\tr}\}$. 
Let $\bV_{\perp}$ be an orthonormal basis for its orthogonal complement. Likewise we define $\bhV_{k,\perp}$.
Define $s_k$ and $\hat{s}_k$ as the $k$-th singular values of $\bA^{\lr,\tr}$ and $\bhA^{\tr}$, respectively.
Next, denote the SVDs
$
\bA^{\lr,\te}=\bU'\bS'(\bV')^T$, $\fil(\bZ^{\te})=\bZ^{\te} \bhrho^{-1} =\bhU'\bhS'(\bhV')^T
$, $\bhA^{\te}=\bhU'_k\bhS'_k(\bhV'_k)^T.
$
We define $\bV'_{\perp}$ and $\bhV'_{\perp}$ analogously to $\bV_{\perp}$. Define $s'_k$ and $\hat{s}'_k$ as the $k$-th singular values of $\bA^{\lr,\te}$ and $\bhA^{\te}$, respectively.
Finally, denote the SVD of the row-wise concatenation of $\bA^{\lr,\tr}$ and $\bA^{\lr,\te}$ as
$
\tbU\tbS\tbV^T.
$
We define $\tbV_{\perp}$ analogously to $\bV_{\perp}$ but with respect to the row-wise concatenation of $\bA^{\lr,\tr}$ and $\bA^{\lr,\te}$.

We define $\beta^*\in\R^p$ as the unique solution to the following optimization problem across $\tr$ and $\te$:
$
\min_{\beta\in \R^p} \|\beta\|_2 $ such that $\beta\in \argmin \left\|\begin{bmatrix}\gamma_0(\bA^{\tr}) \\ \gamma_0(\bA^{\te})\end{bmatrix}-\begin{bmatrix}\bA^{\lr,\tr} \\ \bA^{\lr,\te}\end{bmatrix}
\beta\right\|_2^2.
$
$\beta^*$ is not the quantity of interest, but rather a theoretical device. It defines the unique, minimal-norm, low-rank, linear approximation to the regression $\gamma_0$.

Recall $Y_i=A^{\lr}_{i,\cdot}\beta^*+\phi_i^{\lr}+\varepsilon_i$.
Denote by $Y^{\tr}\in \R^n$ the concatenation of $(Y_i)_{i\in\tr}$. Likewise for $\varepsilon^{\tr}$ and $\phi^{\lr,\tr}$. In the argument for \tr\err, all objects correspond to \tr. For this reason, we suppress superscipt \tr\, when possible. For $i\in \te$, let
    $\hgamma_i=Z_{i,\cdot}\bhrho^{-1} \hbeta $ and $
    \gamma_{i}= \gamma_0(A_{i,\cdot})$
which form the vectors $\hat{\bgamma},\bgamma_0\in\R^n$.

Define the events:
$$
\tilde{\Ec}_1\coloneqq \left\{\|\bhA^{\te} - \bA^{\te}\|^2_{2, \infty}, \|\bhA^{\tr} - \bA^{\lr,\tr}\|^2_{2, \infty},
\le \tilde{\Delta}_1\right\},
$$
$$
\tilde{\Delta}_1 \coloneqq  
C_1
    \cdot \frac{r\ln^5(np)}{\rho_{\min}^4}
    \left(
    1
    +\frac{n}{p}
    +n\Delta^{2}_E
    \right)
    \quad\textrm{and}\quad
    C_1=C\bar{A}^4(K_a + \bar{K})^2 (\kappa + \bar{K} + K_a)^2;
$$
$$
\tilde{\Ec}_2 \coloneqq \left\{\|\bZ^{\te} \bhrho^{-1} -  \bA^{\lr,\te}\|^2
\le \tilde{\Delta}_2\right\},
\quad
\tilde{\Delta}_2 \coloneqq  
C\bar{A}^2 (\kappa + \bar{K} + K_a)^2 \frac{\ln^3(np)}{\rho^2_{\min}} \left(n+p+np\Delta_{E}^2\right);
$$
$$
\tilde{\Ec}_3 \coloneqq \left\{\| \bV \bV^T -\bhV_r \bhV_r^T\|^2, \| \bV' (\bV')^T -\bhV'_r (\bhV'_r)^T\|^2
\le \tilde{\Delta}_3\right\},
\quad
\tilde{\Delta}_3 \coloneqq \frac{r}{np} \tilde{\Delta}_2;
$$
$$
\tilde{\Ec}_4 \coloneqq \left\{\hat{s}_r \gtrsim  s_r\right\};
$$
$$
\tilde{\Ec}_5 \coloneqq \left\{\langle \bhA (\hbeta - \beta^*), \varepsilon\rangle \leq  \tilde{\Delta}_5
\right\},
\quad
\tilde{\Delta}_5 \coloneqq C \bar{\sigma}^2 \ln(np) \left\{r+\|\phi^{\lr}\|_2+\|\beta^*\|_1(\sqrt{n}\bar{A}   + \tilde{\Delta}_1^{1/2}) \right\}.
$$
The Online Appendix uses the results in Appendix~\ref{sec:cleaning} to show that if the conditions of Theorem~\ref{theorem:cov} hold and 
$
    \rho_{\min} \gg \tilde{C}\sqrt{r}\ln^{\frac{3}{2}}(np)\left(\frac{1}{\sqrt{p}}\vee\frac{1}{\sqrt{n}} \vee \Delta_E \right)$, where $\tilde{C}:= C  \bar{A} \Big(\kappa + \bar{K} + K_a \Big),
$ 
then $\p(\tilde{\mathcal{E}}^c)\leq \frac{C}{n^{10}p^{10}}$ where $\tilde{\Ec} \coloneqq \cap^5_{k=1} \tilde{\Ec}_k$.

\begin{lemma}\label{lem:perp}
If Assumption~\ref{assumption:inclusion_row} holds then
$
 \bhV_{k,\perp}^T \hbeta = 0$ and $
\bV_{\perp}^T \beta^* = (\bV_{\perp}')^T \beta^* = 0.
$
\end{lemma}

\begin{proof}
    We generalize \cite[Property 4.1]{agarwal2020principal}, using our new definition of $\beta^*$ and noting that $
\row(\bV^T)=\row\{(\bV')^T\}=\row(\tbV^T)
$.
\end{proof}

\begin{lemma}\label{lem:training_error_intermediate_bound}
Deterministically,
$
\| \bhA \hbeta -\bA^{\lr} \beta^* \|^2_2 
\le C \big\{\|\bhA - \bA^{\lr}\|^2_{2, \infty} \|\beta^*\|^2_1 \ \vee \ \|\phi^{\lr}\|^2_2 \ \vee \ \langle \bhA (\hbeta - \beta^*), \varepsilon\rangle \big\}.
$
\end{lemma}

\begin{proof}
    To begin, write
$
\| \bhA \hbeta - Y \|_2^2 
 = \| \bhA \hbeta - \bA^{\lr} \beta^*-\phi^{\lr}\|_2^2 + \|\varepsilon\|_2^2 - 2 \langle \bhA \hbeta - \bA^{\lr} \beta^*, \varepsilon \rangle+2\langle \phi^{\lr},\varepsilon \rangle.
$
By optimality of $\hbeta$, we have
$
\| \bhA \hbeta - Y \|_2^2 \leq \| \bhA \beta^* - Y \|_2^2 = \| (\bhA - \bA^{\lr}) \beta^*-\phi^{\lr}\|_2^2 + \| \varepsilon\|_2^2 - 2 \langle (\bhA - \bA^{\lr}) \beta^*, \varepsilon\rangle+2\langle \phi^{\lr},\varepsilon \rangle.
$
Combining these results, 
$
\| \bhA \hbeta - \bA^{\lr} \beta^*-\phi^{\lr}\|_2^2 \leq \| (\bhA - \bA^{\lr}) \beta^*-\phi^{\lr}\|_2^2 + 2 \langle \bhA (\hbeta - \beta^*), \varepsilon\rangle$. Moreover, since 
$
\| \bhA \hbeta - \bA^{\lr} \beta^* -\phi^{\lr}\|_2^2 = \| \bhA \hbeta - \bA^{\lr} \beta^*\|_2^2 + \| \phi^{\lr}\|_2^2 - 2\langle \bhA \hbeta - \bA^{\lr} \beta^*, \phi^{\lr} \rangle$ and $
\| (\bhA - \bA^{\lr}) \beta^*-\phi^{\lr}\|_2^2 = \| (\bhA - \bA^{\lr}) \beta^*\|_2^2 + \| \phi^{\lr}\|_2^2 - 2\langle (\bhA - \bA^{\lr}) \beta^*, \phi^{\lr} \rangle$
we conclude that
$
\| \bhA \hbeta - \bA^{\lr} \beta^*\|_2^2 
\leq \| (\bhA - \bA^{\lr}) \beta^*\|_2^2+2\langle\bhA(\hbeta-\beta^*),\phi^{\lr} \rangle+ 2 \langle \bhA (\hbeta - \beta^*), \varepsilon\rangle.$ %
By Cauchy-Schwarz and triangle inequalities, 
$
\langle\bhA(\hbeta-\beta^*),\phi^{\lr} \rangle
\leq (\|\bhA\hbeta-\bA^{\lr}\beta^*\|_2+\|\bhA\beta^*-\bA^{\lr}\beta^*\|_2)\cdot\|\phi^{\lr}\|_2.
$

In summary, for $a=\| \bhA \hbeta -\bA^{\lr} \beta^* \|^2_2$, $b=\|\bhA\beta^*-\bA^{\lr}\beta^*\|^2_2$, and $c=b+2\sqrt{b}\|\phi^{\lr}\|_2+2\langle \bhA (\hbeta - \beta^*), \varepsilon\rangle$, we have shown
$
a
\leq 2\sqrt{a}\|\phi^{\lr}\|_2+c$, which we now analyze. Since $a\geq 0$, there are three possible cases: (i) $c\geq 0$, $2\sqrt{a} \|\phi^{\lr}\|_2\geq c$, so $
    a \leq 4\sqrt{a} \|\phi^{\lr}\|_2$ implies $a\leq 16 \|\phi^{\lr}\|_2^2
    $; (ii) $c\geq 0$, $2\sqrt{a} \|\phi^{\lr}\|_2<c$, so $
    a\leq 2c
    $; (iii) $c < 0$, so $
    a < 2\sqrt{a} \|\phi^{\lr}\|_2$ implies $a<  4\|\phi^{\lr}\|_2^2.
    $
The three cases imply 
$
a \leq 2c \vee 16 \|\phi^{\lr}\|_2^2.
$

Finally, let $d \coloneqq 2b + 4\sqrt{b}\|\phi^{\lr} \|_2 + 2 \|\phi^{\lr}\|_2^2$.
Then
$
d
= 2\{\sqrt{b}+\|\phi^{\lr}\|_2\}^2 \le 4\{b+\|\phi^{\lr}\|^2_2\}.
$
Note $2c\leq  d + 4\langle \bhA (\hbeta - \beta^*), \varepsilon\rangle$.
Together with the earlier results, this implies
$
a \le C \left\{b \ \vee \ \|\phi^{\lr}\|^2_2 \ \vee \ \langle \bhA (\hbeta - \beta^*), \varepsilon\rangle \right\}. 
$
Finally note $b \le \|\bhA- \bA^{\lr}\|^2_{2, \infty} \|\beta^*\|^2_1$.
\end{proof}

\begin{proposition}[Projected \tr\err]\label{prop:param_est_mod}
Suppose conditions of Theorem \ref{theorem:cov} hold.
Further suppose Assumptions~\ref{assumption:noise} and~\ref{assumption:inclusion_row} hold. Let $k=r$ and
$
\rho_{\min} \gg \tilde{C}\sqrt{r}\ln^{\frac{3}{2}}(np)\left(\frac{1}{\sqrt{p}}\vee\frac{1}{\sqrt{n}} \vee  \Delta_E \right)$.
Then with probability at least $1-O\{(np)^{-10}\}$, $\| \bhV_r\bhV_r^T(\hbeta - \beta^*) \|^2_2$ is bounded by
{ \small
\begin{align*}
&
C \bar{A}^4 (K_a+\bar{K})^2(\kappa + \bar{K} + K_a)^2 \frac{\bar{\sigma}^2}{\rho^4_{\min}}\cdot 
r\ln^6(np) \cdot 
\left\{ 
 \ \frac{1}{np}\|\phi^{\lr}\|^2_2 + \|\beta^*\|^2_1 \left( \frac{\sqrt{n}}{\|\beta^*\|_1 np} +
        \frac{r}{np} + 
        \frac{r}{p^2}+ \frac{r}{p} \Delta_E^2\right)
\right\}.
\end{align*}
}
\end{proposition}

\begin{proof}
    We show that for any $k$, $\|\bhV_{k} \bhV_{k}^T(\hbeta-\beta^*)\|_2^2 
   \leq 
    \frac{C}{\hat{s}_k^2} 
\Big\{
\|\bhA - \bA^{\lr}\|^2_{2, \infty} \|\beta^*\|^2_1 \ \vee \ \|\phi^{\lr}\|^2_2 \ \vee \ \langle \bhA (\hbeta - \beta^*), \varepsilon\rangle
\Big\}$. Appealing to $\p(\tilde{\mathcal{E}}^c)\leq \frac{C}{n^{10}p^{10}}$ yields the result. Since $\bhV_k$ is an isometry, $\| \bhV_k \bhV_k^T (\hbeta - \beta^*)  \|_2^2  
= \|\bhV_k^T (\hbeta - \beta^*)  \|_2^2$. Therefore $\| \bhA(\hbeta - \beta^*)\|_2^2 
=  (\hbeta - \beta^*)^T \bhV_k \bhS_k^2 \bhV_k^T (\hbeta - \beta^*)
 \geq \hat{s}_k^2 \| \bhV_k^T (\hbeta - \beta^*) \|_2^2.
$ Next, consider $\| \bhA(\hbeta - \beta^*)\|_2^2 
\leq 2\| \bhA\hbeta -\bA^{\lr} \beta^* \|_2^2 + 2 \| \bA^{\lr} - \bhA \|_{2,\infty}^2 \|\beta^*\|_1^2.$ Combining, 
$
\| \bhV_k \bhV_k^T (\hbeta - \beta^*)  \|_2^2  \leq \frac{2}{\hat{s}_k^2} \Big\{\| \bhA \hbeta -\bA^{\lr} \beta^* \|_2^2 +  \| \bA^{\lr} - \bhA \|_{2,\infty}^2 \|\beta^*\|_1^2\Big\}.
$ Bound $
\| \bhA \hbeta -\bA^{\lr} \beta^* \|^2_2 
$ by Lemma~\ref{lem:training_error_intermediate_bound}.
\end{proof}

\begin{proposition}[\tr\err]\label{prop:param_est}
Suppose conditions of Proposition~\ref{prop:param_est_mod} hold. Then with probability at least $1-O\{(np)^{-10}\}$, $\| \hbeta - \beta^* \|^2_2$ is bounded by
{ \small
\begin{align*}
C \bar{A}^4 (K_a+\bar{K})^2(\kappa + \bar{K} + K_a)^2 \frac{\bar{\sigma}^2}{\rho^4_{\min}}\cdot 
r\ln^6(np) \cdot 
\Big\{ 
 \ \frac{1}{np}\|\phi^{\lr}\|^2_2 + \|\beta^*\|^2_2 \left(
        \frac{r}{n} + 
        \frac{r}{p}+ r \Delta_E^2\right)
\Big\}.
\end{align*}
}
\end{proposition}

\begin{proof}
  We show
$\| \hbeta - \beta^* \|^2_2 
\le 
C\Big[
\| \bV\bV^T - \bhV_{k} \bhV_{k}^T \|^2 \| \beta^* \|^2_2 
+ \frac{1}{\hat{s}_k^2} 
\Big\{
\|\bhA - \bA^{\lr}\|^2_{2, \infty} \|\beta^*\|^2_1 \ \vee \ \|\phi^{\lr}\|^2_2 \ \vee \ \langle \bhA (\hbeta - \beta^*), \varepsilon\rangle
\Big\}
\Big].$ Appealing to $\p(\tilde{\mathcal{E}}^c)\leq \frac{C}{n^{10}p^{10}}$ yields the result. 
Write $\| \hbeta - \beta^*\|_2^2 =\| \bhV_k \bhV_k^T (\hbeta - \beta^*)  \|_2^2 + \| \bhV_{k, \perp}\bhV_{k, \perp}^T (\hbeta - \beta^*)\|_2^2$. Proposition~\ref{prop:param_est_mod} bounds the former. By Lemma~\ref{lem:perp}, the latter equals $ \| \bhV_{k, \perp}\bhV_{k, \perp}^T\beta^*\|_2^2=\| (\bhV_{k, \perp}\bhV_{k, \perp}^T \beta^* - \bV_{\perp}\bV_{\perp}^T) \beta^*\|_2^2$, which we bound by $\| \bhV_{k, \perp}\bhV_{k, \perp}^T  - \bV_{\perp}\bV_{\perp}^T\|^2 \|\beta^*\|_2^2=\| \bV\bV^T  - \bhV_{k}\bhV_{k}^T\|^2 \|\beta^*\|_2^2$.
\end{proof}

\begin{proposition}[\te\err]\label{prop:gen_test}
Let the conditions of Theorem \ref{theorem:fast_rate} hold. Then $\E[\| \bhA^{\te} \hbeta - \bA^{\te} \beta^*\|_2^2 \1\{\tilde{\mathcal{E}}\}]$ is bounded by
{ \small
\begin{align*}
 C_1C_2
    \cdot \bar{\sigma}^2 \cdot \frac{r^3\ln^{8}(np)}{\rho_{\min}^6} \| \beta^* \|^2_1
    \left\{
    1+\frac{p}{n}+\frac{n}{p}+(n+p)\Delta_E^2+np\Delta_E^4\right\}
   +C_2
    \cdot \frac{r^2\ln^3(np)}{\rho_{\min}^2}
    \left(
   1
    +\Delta^{2}_E
    \right) \|\phi^{\tr}\|_2^2.
\end{align*}
}
\end{proposition}

\begin{proof}
     We show $
\| \bhA^{\te} \hbeta - \bA^{\te} \beta^*\|_2^2 \leq C\sum_{m=1}^3 \Delta_m
$ where $\Delta_1:=\big\{\|\bZ^{\te} \bhrho^{-1} -  \bA^{\lr,\te}\|^2+\|\bA^{\lr,\te}\|^2 \| \bV \bV^T -\bhV_r \bhV_r^T\|^2\big\} \| \hbeta - \beta^* \|_2^2 $, $\Delta_2:=\frac{\|\bA^{\lr,\te}\|^2}{\hat{s}_r^2} 
\Big\{
\|\bhA^{\tr} - \bA^{\lr,\tr}\|^2_{2, \infty} \|\beta^*\|^2_1 \ \vee \ \|\phi^{\lr,\tr}\|^2_2 \ \vee \ \langle \bhA^{\tr} (\hbeta - \beta^*), \varepsilon\rangle
\Big\}$, and $\Delta_3:=\|\bhA^{\te} - \bA^{\te}\|_{2,\infty}^2 \|\beta^*\|_1^2$. Appealing to $\tilde{\mathcal{E}}$ yields the result. To begin, write $\| \bhA^{\te} \hbeta - \bA^{\te} \beta^*\|_2^2 \leq 2 \| \bhA^{\te} \big( \hbeta - \beta^* \big) \|_2^2 + 2 \| (\bhA^{\te} - \bA^{\te}) \beta^*\|_2^2.$ We bound the latter term by matrix H\"older:
$
\| (\bhA^{\te} - \bA^{\te}) \beta^*\|_2^2 \leq \|\bhA^{\te} - \bA^{\te}\|_{2,\infty}^2 \|\beta^*\|_1^2
$. In what remains, we analyze the former term using $\| \bhA^{\te} \big( \hbeta - \beta^* \big)  \|_2^2 
 \leq 2 \|\big\{\bhA^{\te} -  \bA^{\lr,\te}\} \big( \hbeta - \beta^* \big)  \|_2^2 + 2 \|\bA^{\lr,\te} \big( \hbeta - \beta^* \big) \|_2^2. $
 
 By Weyl's inequality, we have
$
\| \bhA^{\te} - \bZ^{\te} \bhrho^{-1} \|  =\hat{s}'_{r+1} =\hat{s}'_{r+1}-s'_{r+1} \leq \|\bZ^{\te} \bhrho^{-1} -  \bA^{\lr,\te}\|.
$
In turn, this gives $\| \bhA^{\te} -  \bA^{\lr,\te}\| \leq 2 \|\bZ^{\te} \bhrho^{-1} -  \bA^{\lr,\te}\|$ and hence 
$\|\big\{\bhA^{\te} -  \bA^{\lr,\te}\} \big( \hbeta - \beta^* \big)  \|_2^2  \le 4\|\bZ^{\te} \bhrho^{-1} -  \bA^{\lr,\te}\|^2\cdot \| \hbeta - \beta^* \|_2^2.
$

Assumption \ref{assumption:inclusion_row} implies
$(\bVp)^T \bV_{\perp} = 0$ and hence $\bA^{\lr,\te} \bV_{\perp}\bV_{\perp}^T = 0$. As a result, 
\begin{align*}
&\|\bA^{\lr,\te} \big( \hbeta - \beta^* \big) \|_2^2 = \| \bA^{\lr,\te} (\bV \bV^T + \bV_{\perp}\bV_{\perp}^T) \big( \hbeta - \beta^* \big) \|_2^2  \\
& = \| \bA^{\lr,\te} \bV \bV^T \big( \hbeta - \beta^* \big) \|_2^2 \leq \| \bA^{\lr,\te}\|^2 ~ \| \bV \bV^T \big( \hbeta - \beta^* \big) \|_2^2
\end{align*}
where $\| \bV \bV^T \big( \hbeta - \beta^* \big) \|_2^2 \leq 2 \| \bV \bV^T -\bhV_r \bhV_r^T\|^2 \|\hbeta - \beta^* \|_2^2 + 2 \| \bhV_r \bhV_r^T\big( \hbeta - \beta^* \big) \|_2^2.$ Finally appeal to the proof of Proposition~\ref{prop:param_est_mod}.
\end{proof}

\begin{proposition}[Implicit cleaning]\label{prop:gen_geom}
Let the conditions of Theorem \ref{theorem:fast_rate} hold. Then $\E[\| \bZ^{\te} \bhrho^{-1}\hbeta-\bhA^{\te} \hbeta\|_2^2  \1\{\tilde{\mathcal{E}}\}]$ has the same bound as Proposition~\ref{prop:gen_test}.
\end{proposition}

\begin{proof}
   We show $\| \bZ^{\te} \bhrho^{-1}\hbeta-\bhA^{\te} \hbeta\|_2^2 \leq C\|\bZ^{\te} \bhrho^{-1} -  \bA^{\lr,\te}\|^2\cdot \big\{\|\hbeta-\beta^*\|_2^2
    + 
    \|\bhV'_{r}(\bhV'_{r})^T -\bV'(\bV')^T\|^2 \|\beta^*\|^2_2
    \big\}.$ Appealing to $\tilde{\mathcal{E}}$ yields the result.  Using the definitions
$
\bZ^{\te} \bhrho^{-1} =\bhU'\bhS'(\bhV')^T
$, $\bhA^{\te}=\bhU'_r\bhS'_r(\bhV'_r)^T$, and $
\bhA^{\te}_{\perp}=\bhU'_{r,\perp}\bhS'_{r,\perp}(\bhV'_{r,\perp})^T 
$, write $
\bZ^{\te} \bhrho^{-1}=\quad \bhA^{\te}+\bhA^{\te}_{\perp}
$. Therefore $
 \| \bZ^{\te} \bhrho^{-1}\hbeta-\bhA^{\te} \hbeta\|_2 \leq \|\bhU'_{r,\perp}\|\cdot  \|\bhS'_{r,\perp}\|\cdot \|(\bhV'_{r,\perp})^T \hbeta\|_2=  \|\bhS'_{r,\perp}\|\cdot \|(\bhV'_{r,\perp})^T \hbeta\|_2.
$ By Weyl's inequality, $
\|\bhS'_{r,\perp}\| =\hat{s}'_{r+1} =\hat{s}'_{r+1}-s'_{r+1}\leq \|\bZ^{\te} \bhrho^{-1} -  \bA^{\lr,\te}\|.
$ Moreover, $ \|(\bhV'_{r,\perp})^T \hbeta\|_2
        =  \|\bhV'_{r,\perp}(\bhV'_{r,\perp})^T \hbeta\|_2 \leq \|\bhV'_{r,\perp}(\bhV'_{r,\perp})^T (\hbeta-\beta^*)\|_2+ \|\bhV'_{r,\perp}(\bhV'_{r,\perp})^T \beta^*\|_2.$  Focusing on the former term,
    $
    \|\bhV'_{r,\perp}(\bhV'_{r,\perp})^T (\hbeta-\beta^*)\|_2\leq \|\hbeta-\beta^*\|_2.
    $ By Lemma~\ref{lem:perp}, the latter term equals
    { \small 
    \begin{align*}
        &\|\bhV'_{r,\perp}(\bhV'_{r,\perp})^T \bV'(\bV')^T \beta^*\|_2 
        \leq \|\{\bhV'_{r,\perp}(\bhV'_{r,\perp})^T -\bV'_{\perp}(\bV'_{\perp})^T\}\bV'(\bV')^T \beta^*\|_2
        + \|\bV'_{\perp}(\bV'_{\perp})^T \bV'(\bV')^T \beta^*\|_2 \\
        &=\|\{\bhV'_{r,\perp}(\bhV'_{r,\perp})^T -\bV'_{\perp}(\bV'_{\perp})^T\}\bV'(\bV')^T \beta^*\|_2 =\|\{\bhV'_{r,\perp}(\bhV'_{r,\perp})^T -\bV'_{\perp}(\bV'_{\perp})^T\}\beta^*\|_2 \\
        &= \|\{\bhV'_{r}(\bhV'_{r})^T -\bV'(\bV')^T\}\beta^*\|_2 \leq \|\bhV'_{r}(\bhV'_{r})^T -\bV'(\bV')^T\| \|\beta^*\|_2,
    \end{align*}
    }
    which is dominated by the former term by the proof of Proposition~\ref{prop:param_est}.
\end{proof}

\begin{proof}[Proof of Theorem~\ref{theorem:fast_rate}]
To begin, write
$
    \E\| \hat{\bgamma} - \bgamma_0\|_2^2 
    \leq 2 \E\left[\| \bZ^{\te} \bhrho^{-1}\hbeta-\bA^{\te} \beta^*\|_2^2\1\{\tilde{\mathcal{E}}\}\right]
     +2 \E\left[\| \bZ^{\te} \bhrho^{-1}\hbeta-\bA^{\te} \beta^*\|_2^2\1\{\tilde{\mathcal{E}}^c\}\right]
   +2\|\phi^{\te}\|_2^2.
$ Consider the first term. Using the bound
\begin{multline*}
\E\left[\| \bZ^{\te} \bhrho^{-1}\hbeta-\bA^{\te} \beta^*\|_2^2\1\{\tilde{\mathcal{E}}\}\right]
        \leq 2\E\left[\| \bZ^{\te} \bhrho^{-1}\hbeta-\bhA^{\te} \hbeta\|_2^2\1\{\tilde{\mathcal{E}}\}\right] \\
+ 2 \E\left[\| \bhA^{\te} \hbeta-\bA^{\te} \beta^*\|_2^2\1\{\tilde{\mathcal{E}}\}\right],
\end{multline*}
we may appeal to Propositions~\ref{prop:gen_test} and~\ref{prop:gen_geom}. Consider the second term. Write \\
$
\E\left[\| \bZ^{\te} \bhrho^{-1}\hbeta-\bA^{\te} \beta^*\|_2^2\1\{\tilde{\mathcal{E}}^c\}\right]
       \leq 2 \E\left[\| \bZ^{\te} \bhrho^{-1}\hbeta\|_2^2\1\{\tilde{\mathcal{E}}^c\}\right]+2 \E\left[\| \bA^{\te} \beta^*\|_2^2\1\{\tilde{\mathcal{E}}^c\}\right].
$
Since $
\| \bA^{\te} \beta^*\|_2^2 \leq  \| \bA^{\te}\|^2_{2,\infty} \|\beta^*\|^2_1\leq n\bar{A}^2\|\beta^*\|^2_1
$, we have that $\E\left[\| \bA^{\te} \beta^*\|_2^2\1\{\tilde{\mathcal{E}}^c\}\right] \leq n\bar{A}^2\|\beta^*\|^2_1\p(\tilde{\mathcal{E}}^c)$. By Cauchy-Schwarz inequality,  
$$
\E\big[\| \bZ^{\te} \bhrho^{-1}\hbeta\|_2^2\1\{\tilde{\mathcal{E}}^c\}\big] \leq \sqrt{\E\big[\| \bZ^{\te} \bhrho^{-1}\hbeta\|_2^4\big]}\sqrt{\p(\tilde{\mathcal{E}}^c)}.
$$
These expressions are dominated by the bound from Proposition~\ref{prop:gen_test} under Assumption~\ref{assumption:bounded_response_mod}.
\end{proof}

%% file: C_riesz.tex
\section{Error-in-variable balancing weight}\label{sec:riesz}

As before, we write the proofs without nonlinear dictionaries for clarity.  

We define $\eta^*\in\R^p$ as the unique solution to the following optimization problem across $\tr$ and $\te$:
$
\min_{\eta\in \R^p} \|\eta\|_2$ such that $\eta\in \argmin \left\|\begin{bmatrix}\alpha_0(\bW^{\tr}) \\ \alpha_0(\bW^{\te})\end{bmatrix}-\begin{bmatrix}\bA^{\lr,\tr} \\ \bA^{\lr,\te}\end{bmatrix}
\eta\right\|_2^2.
$
$\eta^*$ is not the quantity of interest, but rather a theoretical device. It defines the unique, minimal-norm, low-rank, linear approximation to the balancing weight $\alpha_0$. 

$\hat{M}$ is the counterfactual moment. Write $\bhG=\frac{1}{n}(\bhA^{\tr})^T\bhA^{\tr}$ as the covariance matrix after data cleaning. In this notation, $\bhG \heta=\hat{M}^T$, and these feasible objects are computed from \tr. We analogously define $M^*$ and $\bG^*$ using the low rank approximation to the signal. In this notation, $\bG^*\eta^*=(M^*)^T$, and these infeasible objects are defined from \tr\, and \te. See the Online Appendix for a more formal statement. Finally, let $\Delta_{RR}:=n\cdot \left\{\|\hat{M}^T-(M^*)^T\|_{\max}+\|\bG^*-\bhG\|_{\max}\|\eta^*\|_1\right\}.$

Define the event: $
    \tilde{\Ec}_5 \coloneqq \left\{\Delta_{RR} \leq \tilde{\Delta}_5 \right\}$ where  $\tilde{\Delta}_5 \coloneqq  C\bar{A}^{5}(\sqrt{C_m'}+C_m''+\bar{\alpha}+\bar{A}) \frac{(K_a + \bar{K})^2(\kappa + \bar{K} + K_a)^2}{\rho^4_{\min}} r\cdot \ln^{5}(np)
    \cdot  n \|\eta^*\|_1\left(  \frac{1}{n}+\frac{1}{p}+\frac{n}{p^2}+\Delta^2_E+n\Delta_E^4
\right)^{\frac{1}{2}}.$ Set $\tilde{\Ec} \coloneqq \cap^5_{k=1} \tilde{\Ec}_k$ where the remaining events are defined in Appendix~\ref{sec:regression}. The Online Appendix shows that if the conditions of Theorem~\ref{theorem:fast_rate_RR} hold, then $
 \p(\tilde{\Ec}^c)\leq  \frac{C}{n^{10}p^{10}}
 $.

\begin{lemma}\label{lem:perp_RR}
If Assumptions~\ref{assumption:inclusion_row} and~\ref{assumption:inclusion_M} hold, 
$
 \bhV_{k,\perp}^T \heta = 0$ and $ \bV_{\perp}^T \eta^* = (\bV_{\perp}')^T \eta^* = 0
$.
\end{lemma}

\begin{proof}
    For the former result, $\heta$ is the unique solution to the program
    $
   \min_{\eta\in\R^p} \|\eta\|_2$ such that $\eta \in \argmin -2\hat{M} \eta+\eta^T \bhG \eta
    $
    where $\hat{M} \in \row(\bhA^{\tr})$ by Assumption~\ref{assumption:inclusion_M} and $\row(\bhG)=\row\{(\bhA^{\tr})^T\bhA^{\tr}\}=\row(\bhA^{\tr})$. Therefore $\heta\in \row(\bhA^{\tr})$, so we can appeal to the same reasoning as Lemma~\ref{lem:perp}. The latter result is similar.
\end{proof}

\begin{lemma}\label{lem:training_error_intermediate_bound_RR}
 $
\|\bhA \heta-\bA^{\lr}\eta^*\|_2^2\leq C\big\{\|\bhV_k\bhV_k^T (\heta-\eta^*)\|_1 \cdot \Delta_{RR}  \vee \|\bhA-\bA^{\lr}\|^2_{2,\infty} \|\eta^*\|_1^2 \big\}$.
\end{lemma}

\begin{proof}
To begin, write $
    \|\bhA\heta-\bA^{\lr}\eta^*\|_2^2 \leq 
    2\|\bhA(\heta-\eta^*)\|_2^2+ 2\|(\bhA-\bA^{\lr})\eta^*\|_2^2.
    $ We bound the latter term by $\|\bhA-\bA^{\lr}\|^2_{2,\infty} \|\eta^*\|_1^2.
    $ Hereafter, we focus on the former term:
    \begin{align*}
         &\frac{1}{n}\|\bhA(\heta-\eta^*)\|_2^2
         =\frac{1}{n}(\heta-\eta^*)^T\bhA^T \bhA (\heta-\eta^*)
         = (\heta-\eta^*)^T\bhV_k\bhV_k^T\bhG (\heta-\eta^*)
    \end{align*}
    which is bounded by $\|\bhV_k\bhV_k^T(\heta-\eta^*)\|_1\cdot \|\bhG (\heta-\eta^*)\|_{\max}.$
    Using the first order conditions, 
    \begin{align*}
        &\|\bhG (\heta-\eta^*)\|_{\max} 
        \leq \|\bhG \heta-\hat{M}^T\|_{\max}+\|\hat{M}^T-(M^*)^T\|_{\max} +\|(M^*)^T-\bG^*\eta^*\|_{\max}+\|\bG^*\eta^*-\bhG\eta^*\|_{\max} \\
        &= 0+\|\hat{M}^T-(M^*)^T\|_{\max}+0+ \|\bG^*-\bhG\|_{\max}\|\eta^*\|_1.\qedhere
    \end{align*}
\end{proof}

\begin{proposition}[Projected \tr\err]\label{prop:param_est_RR2_mod}
Suppose the conditions of Theorem \ref{theorem:cov} hold.
Further suppose Assumptions~\ref{assumption:inclusion_row},~\ref{assumption:inclusion_M},~\ref{assumption:invariance_mod}, and~\ref{assumption:mean_square_cont_mod} hold, and $\|\alpha_0\|_{\infty}\leq \bar{\alpha}$. Let $k=r$ and
$
\rho_{\min} \gg \tilde{C}\sqrt{r}\ln^{\frac{3}{2}}(np)\left(\frac{1}{\sqrt{p}}\vee\frac{1}{\sqrt{n}} \vee  \Delta_E \right)$, where $\tilde{C}:= C  \bar{A} \Big(\kappa + \bar{K} + K_a \Big).
$
Then with probability at least $1-O\{(np)^{-10}\}$, $\| \bhV_r\bhV_r^T(\heta - \eta^*) \|^2_2$ is bounded by
{ \small
\begin{align*}
&C \bar{A}^{10} (\sqrt{C_m'}+C_m''+\bar{\alpha}+\bar{A})^2(K_a + \bar{K})^4(\kappa + \bar{K} + K_a)^4 \cdot r^4\cdot \ln^{10}(np)\cdot  \frac{\|\eta^*\|^2_1}{\rho^8_{\min}} 
    \left( \frac{1}{np}+\frac{1}{p^2}+\frac{n}{p^3}+\frac{1}{p}\Delta^2_E+\frac{n}{p}\Delta_E^4
\right).
\end{align*}
}
\end{proposition}

\begin{proof}
  We show that for any $k$, $ \|\bhV_{k} \bhV_{k}^T(\heta-\eta^*)\|_2^2 
   \leq 
    C\left\{\frac{1}{\hat{s}_k^2} \|\bhA - \bA^{\lr}\|^2_{2, \infty} \|\eta^*\|^2_1+
\frac{1}{\hat{s}_k^4} p\cdot \Delta^2_{RR}
\right\}.$ Appealing to $\p(\tilde{\mathcal{E}}^c)\leq \frac{C}{n^{10}p^{10}}$ yields the result. As in the proof of Proposition~\ref{prop:param_est_mod}, $
   \| \bhV_k \bhV_k^T (\heta - \eta^*)  \|_2^2\leq  \frac{2}{\hat{s}_k^2} \Big\{\| \bhA \heta -\bA^{\lr} \eta^* \|_2^2 +  \| \bA^{\lr} - \bhA \|_{2,\infty}^2 \|\eta^*\|_1^2\Big\}.
   $
Using Lemma~\ref{lem:training_error_intermediate_bound_RR}, we conclude that
$
\| \bhV_k \bhV_k^T (\heta - \eta^*)  \|_2^2 
\leq 
\frac{C}{\hat{s}_k^2} 
\Big\{
\|\bhV_k\bhV_k^T (\heta-\eta^*)\|_1\cdot \Delta_{RR} \vee \| \bA^{\lr} - \bhA \|_{2,\infty}^2 \|\eta^*\|_1^2
\Big\}.
$
There are two cases. In the first case, $ \|\bhV_k \bhV_k^T(\heta-\eta^*)\|_2^2
   \leq C \frac{1}{\hat{s}_k^2} 
\|\bhA - \bA^{\lr}\|^2_{2, \infty} \|\eta^*\|^2_1$, giving the first term. In the second case, $ \|\bhV_k \bhV_k^T(\heta-\eta^*)\|_2^2
   \leq C \frac{1}{\hat{s}_k^2} \|\bhV_k\bhV_k^T (\heta-\eta^*)\|_1\cdot \Delta_{RR}$. Bounding $\|\bhV_k\bhV_k^T (\heta-\eta^*)\|_1 \leq \sqrt{p}\|\bhV_k\bhV_k^T (\heta-\eta^*)\|_2$, dividing both sides by $\|\bhV_k \bhV_k^T(\heta-\eta^*)\|_2$, and squaring gives the second term.
\end{proof}

\begin{proposition}[\tr\err]\label{prop:param_est_RR2}
Suppose the conditions of Proposition~\ref{prop:param_est_RR2_mod} hold.
Then with probability at least $1-O\{(np)^{-10}\}$, $\| \heta - \eta^* \|^2_2$ is bounded by
{ \small
\begin{align*}
    &
C \bar{A}^{10} (\sqrt{C_m'}+C_m''+\bar{\alpha}+\bar{A})^2(K_a + \bar{K})^4(\kappa + \bar{K} + K_a)^4 \cdot r^4\cdot \ln^{10}(np)\cdot  \frac{\|\eta^*\|^2_2}{\rho^8_{\min}} 
    \left( \frac{1}{n}+\frac{1}{p}+\frac{n}{p^2}+\Delta^2_E+n\Delta_E^4
\right).
\end{align*}
}
\end{proposition}

\begin{proof}
     We show $\|\bhV_{k} \bhV_{k}^T(\heta-\eta^*)\|_2^2 
   \leq 
    C\left\{\frac{1}{\hat{s}_k^2} \|\bhA - \bA^{\lr}\|^2_{2, \infty} \|\eta^*\|^2_1+
\frac{1}{\hat{s}_k^4} p\cdot \Delta^2_{RR}
\right\}.$ Appealing to $\p(\tilde{\mathcal{E}}^c)\leq \frac{C}{n^{10}p^{10}}$ yields the result. The argument is similar to Proposition~\ref{prop:param_est}, replacing Lemma~\ref{lem:perp} with Lemma~\ref{lem:perp_RR} and Proposition~\ref{prop:param_est_mod} with Proposition~\ref{prop:param_est_RR2_mod}.
\end{proof}

\begin{proposition}[\te\err]\label{prop:gen_test_RR}
Let the conditions of Theorem \ref{theorem:fast_rate_RR} hold. Then $\E[\| \bhA^{\te} \heta - \bA^{\te} \eta^*\|_2^2 \1\{\tilde{\Ec}\}]$ is bounded by
{ \small
\begin{align*}
C_3 \cdot \frac{r^5\ln^{13}(np)}{\rho^{10}_{\min}}  \cdot  \|\eta^*\|^2_1
    \left\{  1+\frac{p}{n}+\frac{n}{p}+\frac{n^2}{p^2}+\left(n+p+\frac{n^2}{p}\right)\Delta^{2}_E+(np+n^2)\Delta_E^4+n^2p \Delta^{6}_E
\right\}.
\end{align*}
}
\end{proposition}

\begin{proof}
    The proof is similar to Proposition~\ref{prop:gen_test}, updating $\Delta_2=\frac{\|\bA^{\lr,\te}\|^2}{\hat{s}_r^2} 
\Big\{ 
\|\bhV_k\bhV_k^T (\heta-\eta^*)\|_1\cdot \Delta_{RR} \vee \| \bA^{\lr,\tr} - \bhA^{\tr} \|_{2,\infty}^2 \|\eta^*\|_1^2
\Big\}$. In particular, we bound $\| \bhV_r \bhV_r^T\big( \heta - \eta^* \big) \|_2^2$ using Proposition~\ref{prop:param_est_RR2_mod} instead of Proposition~\ref{prop:param_est_mod}.
\end{proof}

\begin{proposition}[Implicit cleaning]\label{prop:gen_geom_RR}
Let the conditions of Theorem \ref{theorem:fast_rate_RR} hold. Then $\E[\| \bZ^{\te} \bhrho^{-1}\heta-\bhA^{\te} \heta \|_2^2 \1\{\tilde{\Ec}\}]$ has the same bound as Proposition~\ref{prop:gen_test_RR}.
\end{proposition}

\begin{proof}
    The proof is analogous to Proposition~\ref{prop:gen_geom}.
\end{proof}

\begin{proof}[Proof of Theorem~\ref{theorem:fast_rate_RR}]
The proof is analogous to Theorem~\ref{theorem:fast_rate}, instead appealing to Propositions~\ref{prop:gen_test_RR} and~\ref{prop:gen_geom_RR}.
\end{proof}

%% file: D_target.tex
\section{Data cleaning-adjusted confidence intervals}\label{sec:target}

Let
$
\psi_i=\psi(W_{i,\cdot},\theta_i,\gamma_0,\alpha_0) $ where $ \psi(w,\theta,\gamma,\alpha)=m(w,\gamma)+\alpha(w)\{y-\gamma(w)\}-\theta
$
and $\gamma\mapsto m(w,\gamma)$ is linear.
We take as given that $(\gamma_0,\alpha_0)$ exist, though the latter is implied by Assumption~\ref{assumption:mean_square_cont_mod}. The Gateaux derivative of $\psi(w,\theta,\gamma,\alpha)$ with respect to its argument $\gamma$ in the direction $u$ is
$
\{\partial_{\gamma} \psi(w,\theta,\gamma,\alpha)\}(u)=\frac{\partial}{\partial \tau} \psi(w,\theta,\gamma+\tau u,\alpha) \bigg|_{\tau=0}.
$

Let $L$ be the number of folds, with fold $\ell$ indexed by $I_{\ell}$. Train $(\hgamma_{\ell},\halpha_{\ell})$ on observations in $I_{\ell}^c$, which serves as \tr. Let $m=|I_{\ell}|=n/L$ be the number of observations in $I_{\ell}$, which serves as \te. Denote by $\E_{\ell}[\cdot]$ the average over observations in $I_{\ell}$. This generalized notation allows us to reverse the roles of $\tr$\, and $\te$, and to allow for more than two folds. The target and oracle are
$
 \hat{\theta}=L^{-1}\sum_{\ell=1}^L\E_{\ell} [m(W_{i,\cdot},\hgamma_{\ell})+\halpha_{\ell}(W_{i,\cdot})\{Y_i-\hgamma_{\ell}(W_{i,\cdot})\}]$ and $
    \bar{\theta}=L^{-1}\sum_{\ell=1}^L\E_{\ell} [m(W_{i,\cdot},\gamma_0)+\alpha_0(W_{i,\cdot})\{Y_i-\gamma_0(W_{i,\cdot})\}]
$. For $i \in I_{\ell}$, let $\hat{\psi}_i=\psi(W_{i,\cdot},\hat{\theta},\hgamma_{\ell},\halpha_{\ell})$.

\begin{lemma}\label{lem:Delta_inid}
Suppose Assumptions~\ref{assumption:invariance_mod} and~\ref{assumption:mean_square_cont_mod} hold,
$
\E[\varepsilon_i^2|W_{i,\cdot}]\leq \bar{\sigma}^2$, $\|\alpha_0\|_{\infty}\leq\bar{\alpha}
$, and that for $(i,j)\in I_{\ell}$,
$
\hgamma_{\ell}(W_{i,\cdot})\indep \hgamma_{\ell}(W_{j,\cdot})|I^c_{\ell}$ and $\halpha_{\ell}(W_{i,\cdot})\indep \halpha_{\ell}(W_{j,\cdot})|I^c_{\ell}.
$
Then with probability $1-\epsilon$, 
$
\frac{n^{1/2}}{\sigma}|\hat{\theta}-\bar{\theta}|\leq \Delta= \frac{3 L}{\epsilon  \sigma}\left[(\bar{Q}^{1/2}+\bar{\alpha})\{\mathcal{R}(\hgamma_{\ell})\}^{\bar{q}/2}+\bar{\sigma}\{\mathcal{R}(\halpha_{\ell})\}^{1/2}+ \{n\mathcal{R}(\hgamma_{\ell})\mathcal{R}(\halpha_{\ell})\}^{1/2}\right].
$
\end{lemma}

\begin{proof}
We generalize \cite[Proposition S6]{chernozhukov2021simple} to the new norm, noting that Jensen's inequality and $\bar{q}\in(0,1]$ imply that
$
        \E_{\ell}[\E[m(W_{i,\cdot},u)^2|I^c_{\ell}]]
        \leq \bar{Q} \E_{\ell}[\{\E[u(W_i)^2|I^c_{\ell}]\}^{\bar{q}}] 
        \leq \bar{Q} \{\E_{\ell}[\E[u(W_i)^2|I^c_{\ell}]]\}^{\bar{q}} 
   $. 
   Moreover, using the shorthand $u_i=u(W_{i,\cdot})$ and $v_i=v(W_{i,\cdot})$,
    $
        \E[\E_{\ell}\{|u(W_{i,\cdot})v(W_{i,\cdot})|\}]
    =\frac{1}{m}\E[u^Tv]
    \leq \frac{1}{m}(\E[u^Tu])^{1/2}(\E[v^Tv])^{1/2} 
    =\left(\frac{1}{m}\E[\E_{\ell}[u_i^2]]\right)^{1/2}\left(\frac{1}{m}\E[\E_{\ell}[v_i^2]]\right)^{1/2}=\sqrt{\mathcal{R}(\hgamma_{\ell})}\sqrt{\mathcal{R}(\halpha_{\ell})}.
   $
\end{proof}

\begin{proof}[Proof of Theorem~\ref{thm:dml_inid}]
We generalize \cite[Theorem 1]{chernozhukov2021simple} using Lemma~\ref{lem:Delta_inid} and an i.n.i.d. Berry Esseen lemma for $
   \bar{\theta}-\theta_0=
   \E_n\psi_i$ \cite{shevtsova2010improvement}.
\end{proof}

\begin{lemma}\label{lem:Delta2_inid}
Suppose Assumptions~\ref{assumption:invariance_mod} and~\ref{assumption:mean_square_cont_mod} hold, 
$
\E[\varepsilon_i^2|W_{i,\cdot}]\leq \bar{\sigma}^2$, and $\|\halpha_{\ell}\|_{\infty}\leq\bar{\alpha}'
$. Then with probability $1-\epsilon'/2$, 
$
\E_n\{(\hat{\psi}_i-\psi_i+\theta_0-\theta_i)^2\}\leq \Delta'=4(\hat{\theta}-\theta_0)^2+\frac{24 L}{\epsilon'}\left[\{\bar{Q}+(\bar{\alpha}')^2\}\mathcal{R}(\hgamma_{\ell})^{\bar{q}}+\bar{\sigma}^2\mathcal{R}(\halpha_{\ell})\right].
$
\end{lemma}

\begin{proof}
We generalize \cite[Proposition S10]{chernozhukov2021simple} to the new norm, appealing to Jensen's inequality and $\bar{q}\in(0,1]$ as in the proof of Lemma~\ref{lem:Delta_inid}. 
\end{proof}

\begin{lemma}\label{lem:other_half_inid}
With probability $1-\epsilon'/2$,
$
|\E_n(\psi_i^2)-\sigma^2|\leq \Delta''=\left(\frac{2}{\epsilon'}\right)^{1/2}\frac{\chi^2}{n^{1/2}}.
$
\end{lemma}
\begin{proof}
Let
$
B_i=\psi_i^2 $ and $\bar{B}=\E_n [B_i]
$
so that
$
\E[\bar{B}]=
\frac{1}{n}\sum_{i=1}^n \E[B_i] = \frac{1}{n}\sum_{i=1}^n \sigma_i^2=
\sigma^2$ and $
    \V[\bar{B}]=\frac{\sum_{i=1}^n \V(B_i)}{n^2}\leq
\frac{\sum_{i=1}^n \E[B_i^2]}{n^2}=\frac{\sum_{i=1}^n \chi^4_i}{n^2}= \frac{\chi^4}{n}.
$
By Markov inequality
$
\p(|\bar{B}-\E[\bar{B}]|>t)\leq \frac{\V[\bar{B}]}{t^2}=\frac{\epsilon'}{2}.
$
Solving gives
$
t=\Delta''.
$
\end{proof}

\begin{proof}[Proof of Theorem~\ref{thm:var_inid}]
    To begin, write
   $
        \hat{\sigma}^2-(\sigma^2+\bias)
        =\{\hat{\sigma}^2-\E_n(\psi_i^2)-\bias\}+\{\E_n(\psi_i^2)-\sigma^2\} 
        \leq \{\hat{\sigma}^2-\E_n(\psi_i^2)-\bias\}+\Delta''
  $
where the inequality holds with probability $1-\epsilon'/2$ by Lemma~\ref{lem:other_half_inid}. In what follows, we focus on the former term. In particular, we write
$
        \hat{\sigma}^2=\E_n\{(\hat{\psi}_i-\psi_i)^2\}+2\E_n\{(\hat{\psi}_i-\psi_i)\psi_i\}+\E_n(\psi_i^2),
  $ then solve for $\bias=\bias_1+\bias_2$ as a function of $\Delta_{\out}$ in the decomposition
$
\hat{\sigma}^2-\E_n(\psi_i^2)-\bias=\E_n\{(\hat{\psi}_i-\psi_i)^2\}-\bias_1+2\E_n\{(\hat{\psi}_i-\psi_i)\psi_i\}-\bias_2.
$
To derive $\bias_1$, open the square and write
{ \small
\begin{align*}
    &\E_n\{(\hat{\psi}_i-\psi_i)^2\} 
    =\E_n\{(\hat{\psi}_i-\psi_i+\theta_0-\theta_i)^2\}+\E_n\{(\theta_i-\theta_0)^2\}+2\E_n\{(\hat{\psi}_i-\psi_i+\theta_0-\theta_i)(\theta_i-\theta_0)\} \\
    &\leq \E_n\{(\hat{\psi}_i-\psi_i+\theta_0-\theta_i)^2\}+\E_n\{(\theta_i-\theta_0)^2\}+2[\E_n\{(\hat{\psi}_i-\psi_i+\theta_0-\theta_i)^2\}]^{1/2}
    [\E_n\{(\theta_i-\theta_0)^2\}]^{1/2}\\
    &\leq \Delta'+\Delta_{\out}+2(\Delta')^{1/2}\Delta_{\out}^{1/2}
\end{align*}
}
where the last line holds with probability $1-\epsilon'/2$ by Lemma~\ref{lem:Delta2_inid}. Taking $\bias_1=\Delta_{\out}$, we have shown
$
\E_n\{(\hat{\psi}_i-\psi_i)^2\}-\bias_1 \leq \Delta'+2(\Delta')^{1/2}\Delta_{\out}^{1/2}.
$
To derive $\bias_2$, write
\begin{align*}
    &\E_n\{(\hat{\psi}_i-\psi_i)\psi_i\} 
    \leq \left[\E_n\{(\hat{\psi}_i-\psi_i)^2\}\right]^{1/2}\left\{|\E_n(\psi_i^2)-\sigma^2|+\sigma^2\right\}^{1/2} \\
    &\leq \{\Delta'+\Delta_{\out}+2(\Delta')^{1/2}\Delta_{\out}^{1/2}\}^{1/2} \cdot \{\Delta''+\sigma^2\}^{1/2}
\end{align*}
where the last line holds with probability $1-\epsilon'$ appealing to Lemmas~\ref{lem:Delta2_inid} and~\ref{lem:other_half_inid} as well as the analysis for $\bias_1$. In summary,
\begin{align*}
    2\E_n\{(\hat{\psi}_i-\psi_i)\psi_i\} &\leq 2\{\Delta'+\Delta_{\out}+2(\Delta')^{1/2}\Delta_{\out}^{1/2}\}^{1/2} \cdot \{\Delta''+\sigma^2\}^{1/2} \\
    &\leq 2\{(\Delta')^{1/2}+\Delta_{\out}^{1/2}+2^{1/2}(\Delta')^{1/4}\Delta_{\out}^{1/4}\}\cdot \{(\Delta'')^{1/2}+\sigma\}.
\end{align*}
Taking $\bias_2=2\Delta_{\out}^{1/2}\sigma$, we have shown
$$
2\E_n\{(\hat{\psi}_i-\psi_i)\psi_i\}-\bias_2\leq 2(\Delta')^{1/2}\{(\Delta'')^{1/2}+\sigma\}+2\Delta_{\out}^{1/2}(\Delta'')^{1/2}+2^{3/2}(\Delta')^{1/4}\Delta_{\out}^{1/4}\{(\Delta'')^{1/2}+\sigma\}.
$$
Thus with probability $1-\epsilon'$, $\hat{\sigma}^2-(\sigma^2+\bias)
    \leq\{\hat{\sigma}^2-\E_n(\psi_i^2)-\bias\}+\Delta''$, equalling
    { \small
\begin{align*}
    &\E_n\{(\hat{\psi}_i-\psi_i)^2\}-\bias_1+2\E_n\{(\hat{\psi}_i-\psi_i)\psi_i\}-\bias_2+\Delta'' \\
    &\leq \Delta'+2(\Delta')^{1/2}\Delta_{\out}^{1/2} 
    +2(\Delta')^{1/2}\{(\Delta'')^{1/2}+\sigma\}+2\Delta_{\out}^{1/2}(\Delta'')^{1/2}+2^{3/2}(\Delta')^{1/4}\Delta_{\out}^{1/4}\{(\Delta'')^{1/2}+\sigma\}
    +\Delta''.
\end{align*}
} 
Combining terms yields the desired result.
\end{proof}

%% file: E_examples.tex
\section{Additional examples}\label{sec:examples}

\textbf{Semiparametric estimands}.
We consider causal parameters of the form
$
\theta_0=\frac{1}{n}\sum_{i=1}^n \theta_i$, where $\theta_i=\E[m(W_{i,\cdot},\gamma_0)]$,
in an i.n.i.d. data generating process where
$
     Y_i=\gamma_0(D_i,X_{i,\cdot})+\varepsilon_i$, $
      Z_{i,\cdot}=(X_{i,\cdot}+H_{i,\cdot})\odot \pi_{i,\cdot}$, and $W_{i,\cdot}=(D_i,X_{i,\cdot},H_{i,\cdot},\pi_{i,\cdot}).
$
$(D_i,X_{i,\cdot})$ concatenate the various arguments of $\gamma_0$, which we hereby call regressors. 
This model includes the scenario in which  some variables are corrupted and other are not. 
Which regressors are corrupted or uncorrupted constrains the construction of technical regressors; see Appendix~\ref{sec:dictionary}. We concatentate signal and noise as $W_{i,\cdot}$. Appendix~\ref{sec:target_proof} generalizes Assumption~\ref{assumption:invariance_mod} to impose invariance of the regression $\gamma_0$ and generalized balancing weight $\alpha_0$ across observations.

\begin{example}[Average treatment effect]\label{ex:ATE}
Let $(D_i,X_{i,\cdot})$ concatenate treatment $D_i\in\{0,1\}$ and covariates $X_{i,\cdot}\in\R^{p}$. Denote $\gamma_0(D_i,X_{i,\cdot}):=\E[Y_i|D_i,X_{i,\cdot}]$. Under the assumption of selection on $X_{i,\cdot}$, the average treatment effect is given by 
$
\theta_i=\E[\gamma_0(1,X_{i,\cdot})-\gamma_0(0,X_{i,\cdot})].
$
With uncorrupted treatment and corrupted covariates, $W_{i,\cdot}=(D_i,X_{i,\cdot},H_{i,\cdot},\pi_{i,\cdot})$ where $(H_{i,\cdot},\pi_{i,\cdot})$ are measurement error and missingness for the covariates.\footnote{More generally, treatment observations may be corrupted as well. For readability, we exposit the simpler and plausible case that treatment is uncorrupted.}
\end{example}
While the true regression $\gamma_0(D_i,X_{i,\cdot})$ is only a function of signal $(D_i,X_{i,\cdot})$, our regression estimator $\hgamma(D_i,Z_{i,\cdot})$ is a function of both signal and noise $W_{i,\cdot}$. In other words, the hypothesis space for estimation is the extended space of functions $\mathbb{L}_2(\mathcal{W})$, and we must define an extended functional over $\mathbb{L}_2(\mathcal{W})$. In Example~\ref{ex:ATE}, the extended functional is
$
\gamma\mapsto \E[\gamma(1,X_{i,\cdot},H_{i,\cdot},\pi_{i,\cdot})-\gamma(0,X_{i,\cdot},H_{i,\cdot},\pi_{i,\cdot})].
$

\begin{example}[Local average treatment effect]\label{ex:LATE}
Let $(U_i, X_{i,\cdot})$ concatenate instrument $U_i\in\{0,1\}$ and covariates $X_{i,\cdot}\in\R^{p}$. Denote $\gamma_0(U_i,X_{i,\cdot}):=\E[Y_i|U_i,X_{i,\cdot}]$ and $\delta_0(U_i,X_{i,\cdot}):=\E[D_i|U_i,X_{i,\cdot}]$. Under standard instrumental variable assumptions, the local average treatment effect for the subpopulation of compliers is given by
$
\beta_0=\frac{\theta_0}{\theta'_0} $ where $\theta_i=\E[\gamma_0(1,X_{i,\cdot})-\gamma_0(0,X_{i,\cdot})]$ and $\theta_i'=\E[\delta_0(1,X_{i,\cdot})-\delta_0(0,X_{i,\cdot})].
$
With uncorrupted instrument and corrupted covariates, $W_{i,\cdot}=(U_i,X_{i,\cdot},H_{i,\cdot},\pi_{i,\cdot})$ where $(H_{i,\cdot},\pi_{i,\cdot})$ are measurement error and missingness for the covariates. 
\end{example}

\begin{example}[Average policy effect]\label{ex:policy}
Let $X_{i,\cdot}\in\R^{p}$ be the covariates. Consider the counterfactual transportation of covariates $x_{i,\cdot}\mapsto t(x_{i,\cdot})$. Denote $\gamma_0(X_{i,\cdot}):=\E[Y_i|X_{i,\cdot}]$. The average policy effect is
$
\theta_i=\E[\gamma_0\{t(X_{i,\cdot})\}-\gamma_0(X_{i,\cdot})].
$
With corrupted covariates, $W_{i,\cdot}=(X_{i,\cdot},H_{i,\cdot},\pi_{i,\cdot})$ where $(H_{i,\cdot},\pi_{i,\cdot})$ are measurement error and missingness for covariates. 

\end{example}

\begin{example}[Price elasticity of demand]\label{ex:deriv}
Let $Y_i$ be price of a particular good. Let $(D_i,X_{i,\cdot})$ concatenate quantities sold of the particular good $D_i$ and other goods $X_{i,\cdot}\in\R^{p}$. Denote $\gamma_0(D_{i},X_{i,\cdot})=\E[Y_i|D_{i},X_{i,\cdot}]$. The average price elasticity of demand is
$
\theta_i=\E\left[\nabla_{d}\gamma_0(D_i,X_{i,\cdot})\right].
$
With uncorrupted quantity for the particular good and corrupted quantities for the other goods, $W_{i,\cdot}=(D_i,X_{i,\cdot},H_{i,\cdot},\pi_{i,\cdot})$ where $(H_{i,\cdot},\pi_{i,\cdot})$ are measurement error and missingness for the other goods. 

\end{example}

\textbf{Weighted estimands}. In empirical economic research with aggregate units, it is common to weight units by their size. It is also common to consider partially linear models. For example, the estimand of \cite{autor2013china} may be viewed as a weighted partially linear instrumental variable regression. To bridge theory with practice, we provide these examples next. A weighted functional $\theta_0\in\R$ is a scalar that takes the form
$
\theta_0=\frac{1}{n}\sum_{i=1}^n \theta_i$ where $\theta_i=\E[\ell_i m(W_{i,\cdot},\gamma_0)]
$
and $\ell_i$ is the weight for aggregate unit $i$. For simplicity, we take the weights $\ell_i$ to be known, but their uncertainty can be incorporated as well. 

\begin{example}[Weighted partially linear regression]\label{ex:p_ols}
Let $(D_i,X_i)$ concatenate treatment $D\in\R$ and covariates $X_{i,\cdot}\in\R^p$. Denote $\gamma_0(D_i,X_{i,\cdot})=\E[Y_i|D_i,X_{i,\cdot}]$. The weighted partially regression coefficient is given by
$
\theta_i=\E\left[\ell_i\{\gamma_0(d+1,X_{i,\cdot})-\gamma_0(d,X_{i,\cdot})\}\right].
$
With uncorrupted treatment and corrupted covariates, $W_{i,\cdot}=(D_i,X_{i,\cdot},H_{i,\cdot},\pi_{i,\cdot})$ where $(H_{i,\cdot},\pi_{i,\cdot})$ are measurement error and missingness for the covariates.
\end{example}

\begin{example}[Weighted partially linear instrumental variable regression]\label{ex:p_iv}
Let $(U_i, X_{i,\cdot})$ concatenate instrument $U_i\in\R$ and covariates $X_{i,\cdot}\in\R^{p}$. Denote $\gamma_0(U_i,X_{i,\cdot}):=\E[Y_i|U_i,X_{i,\cdot}]$ and $\delta_0(U_i,X_{i,\cdot}):=\E[D_i|U_i,X_{i,\cdot}]$. Under standard instrumental variable assumptions, the weighted partially linear instrumental variable regression coefficient is given by
$
\beta_0=\frac{\theta_0}{\theta'_0}$, where $\theta_i=\E[\ell_i\{\{\gamma_0(u+1,X_{i,\cdot})-\gamma_0(u,X_{i,\cdot})\}]$ and $ \theta_i'=\E[\ell_i\{\delta_0(u+1,X_{i,\cdot})-\delta_0(u,X_{i,\cdot})\}].
$
With uncorrupted instrument and corrupted covariates, $W_{i,\cdot}=(U_i,X_{i,\cdot},H_{i,\cdot},\pi_{i,\cdot})$ where $(H_{i,\cdot},\pi_{i,\cdot})$ are measurement error and missingness for the covariates. 
\end{example}

\textbf{Nonparametric estimands}. A local functional $\theta_0^{\lim}\in \R$ is a scalar that takes the form
$
\theta^{\lim}_{0}
=\lim_{h\rightarrow 0} \theta_0^h$, where $\theta_0^h=\frac{1}{n}\sum_{i=1}^n\theta^h_i$, $ \theta_i^h=\E[m_h(W_{i,\cdot},\gamma_0)]
=\E[\ell_h(W_{ij}) m(W_{i,\cdot},\gamma_0)]
$. Here, $\ell_h$ is a Nadaraya Watson weighting with bandwidth $h$ and $W_{ij}$ is a scalar component of $W_{i,\cdot}$. $\theta^{\lim}_{0}$ is a nonparametric quantity. However, it can be approximated by the sequence $\{\theta_0^h\}$. Each $\theta_0^h$ can be analyzed like a weighted functional as long as we keep track of how certain quantities depend on $h$. By this logic, finite sample semiparametric theory for $\theta^h_0$ translates to finite sample nonparametric theory for $\theta_0^{\lim}$ up to some approximation error. In this sense, our analysis encompasses both semiparametric and nonparametric estimands. 

\begin{example}[Heterogeneous treatment effect]\label{ex:CATE}
Let $(D_i,V_{i},X_{i,\cdot})$ concatenate treatment $D_i\in\{0,1\}$, covariate of interest $V_{i}\in \R$, and other covariates $X_{i,\cdot}\in\R^{p}$. Denote $\gamma_0(D_i,V_i,X_{i,\cdot}):=\E[Y_i|D_i,V_i,X_{i,\cdot}]$. Under the assumption of selection on $(V_i,X_{i,\cdot})$ and identicial distribution of $V_i$, the heterogeneous treatment effect for the subpopulation with subcovariate value $v$ is given by 
$
\theta_i=\E[\gamma_0(1,V_i,X_{i,\cdot})-\gamma_0(0,V_i,X_{i,\cdot})|V_i=v]=\lim_{h\rightarrow 0}\E[\ell_h(V_i)\{\gamma_0(1,V_i,X_{i,\cdot})-\gamma_0(0,V_i,X_{i,\cdot})\}]
$
where 
$
\ell_h(V_i)=\frac{K\left\{(V_i-v)/h\right\}}{ \omega}$, $\omega=\E [K\left\{(V_i-v)/h\right\}]
$,
and $K$ is a standard kernel function. With uncorrupted treatment, uncorrupted covariate of interest, and corrupted other covariates, $W_{i,\cdot}=(D_i,V_i,X_{i,\cdot},H_{i,\cdot},\pi_{i,\cdot})$ where $(H_{i,\cdot},\pi_{i,\cdot})$ are measurement error and missingness for the other covariates. 

\end{example}

Appendix~\ref{sec:target_proof} formally defines our general class of semiparametric and nonparametric estimands. Each example belongs to the class under generalizations of Assumption~\ref{assumption:mean_square_cont_mod}.

\textbf{Missing outcomes}. So far, $Y_i$ has been uncorrupted. Measurement error and differential privacy of $Y_i$ are allowed by response noise $\varepsilon_i$. An important additional issue is outcome attrition: for some observations, $Y_i$ is missing in a way that may depend on the true regressors. The enriched observation model is
$
    Y_i=\gamma_0(D_i,X_{i,\cdot},S_i)+\varepsilon_i$, 
    $Z_{i,\cdot}=[X_{i,\cdot}+H_{i,\cdot}] \odot \pi_{i,\cdot}$, and $
    \tilde{Y}_i=Y_i\cdot S_i
$ with $S_i\in \{1,\NA\}$. Instead of $(Y_i,D_i,X_{i,\cdot})$, the analyst observes $(\tilde{Y}_i,D_i,Z_{i,\cdot})$. Outcome $Y_i$ may be missing at random conditional on true regressors $(D_i,X_{i,\cdot})$, of which $X_{i,\cdot}$ may be corrupted. The extended semiparametric model is
$
    \E[\tilde{Y}_i|D_i,X_{i,\cdot},H_{i,\cdot},\pi_{i,\cdot},S_i=1]
    =\E[Y_i|D_i,X_{i,\cdot},H_{i,\cdot},\pi_{i,\cdot},S_i=1]
    =
\E[Y_i|D_i,X_{i,\cdot},S_i=1]=\gamma_0(D_i,X_{i,\cdot},S_i=1).
$
For this extension, replace $Y_i$ with $\tilde{Y}_i$ and replace $(D,X_{i,\cdot})$ with $(D_i,X_{i,\cdot},S_i)$.

%% file: F_dictionary.tex
\section{Nonlinearity}\label{sec:dictionary}

We characterize the class of nonlinear dictionaries $b:\R^p \rightarrow \R^{p'}$ for which our main results go through. We discuss two classes of dictionaries and delay proofs to the end.

\textbf{Polynomial dictionary}. We refer to the following three simple properties as dictionary continuity, since they imply that the data cleaning results for original regressors imply similar data cleaning results for technical regressors constructed from the dictionary. We state the properties then verify them for the polynomial dictionary of degree $d_{\max}$. 

\begin{assumption}[Dictionary continuity]\label{assumption:smooth_dictionary}
%
(i) For any two matrices $\bM^{(1)}, \bM^{(2)} \in \R^{n \times p}$, 
            $\|b(\bM^{(1)}) - b(\bM^{(2)}) \|^2_{2, \infty} \le C_b' \|\bM^{(1)} - \bM^{(2)} \|^2_{2, \infty}$;
 (ii) for any $\bM \in \R^{n \times p}$, $rank\{b(\bM)\} \le \left\{rank(\bM)\right\}^{C_b''}$;
 (iii) for any $v \in \R^p$, $\| b(v) \|_{\max} \le  \left(\| v \|_{\max}\right)^{C_b'''}$.
\end{assumption}
For much of our argument to go through, it suffices that the dictionary exhibits three simple properties: clean original regressors imply clean technical regressors; low rank original regressors imply low rank technical regressors; and bounded original regressors imply bounded technical regressors. Polynomial dictionaries have these properties.

\begin{definition}[Polynomial dictionary]
Let $v = (v_1, v_2, \dots, v_p) \in \R^p$.
Consider the dictionary $b^{\poly}$, where for $k \in [p']$, $b^{\poly}_k(v)=\prod_{\ell=1}^{d(k)} v_{\ell}$ with $v_{\ell} \in \{v_1, \dots, v_p\}$. 
\end{definition}
That is, each basis function $b^{\poly}_k(v)$ in the dictionary is a polynomial of degree $d(k)\leq d_{\max}$ constructed from coordinates of $v$, allowing for repeats. This class of dictionaries is commonly used in empirical economic research. It nests as a special case the interacted dictionary studied in the main text, which permits a rich model of heterogeneous treatment effects. Pleasingly, for this class, the dictionary constants $(C_b', C_b'', C_b''')$ do not depend on $p'$. Rather, $(C_b', C_b'', C_b''')$ depend on the maximum degree $d_{\max}$ of the polynomial dictionary. 
\begin{proposition}[Verifying dictionary continuity]\label{prop:poly_dict} 
$b^{\poly}$ of degree $d_{\max}$ satisfies Assumption~\ref{assumption:smooth_dictionary} with $C_b' \le 2^{d_{\max}} \cdot  \|\bM^{(1)}\|^{2d_{\max}}_{\max} \cdot \|\bM^{(2)}\|^{2d_{\max}}_{\max}$,
$C_b'' \le d_{\max}$, and $C_b''' \le d_{\max}$.
\end{proposition}
This class of dictionaries preserves the low rank approximation in the following sense.
\begin{proposition}[Low rank approximation is preserved]\label{prop:preserve_approx}
Suppose Assumption~\ref{assumption:bounded} holds and the true covariates satisfy $\bX=\bX^{\lr}+\bE^{\lr}$ where $r=rank\{\bX^{\lr}\}$ and $\Delta_E=\|\bE^{\lr}\|_{\max}$. Consider $b=b^{\poly}$ of degree $d_{\max}$. Then
$r':=rank\{b(\bX^{\lr})\}\leq r^{d_{\max}}$ and $\Delta'_E:=\|b(\bX)-b(\bX^{\lr})\|_{\max}\leq C \bar{A}^{d_{\max}} \cdot d_{\max} \Delta_E$.
\end{proposition}
The same logic applies for dictionaries applied to $(D_i,X_{i,\cdot})$ rather than $X_{i,\cdot}$. The generalization of Appendix~\ref{sec:cleaning} with nonlinear dictionaries is immediate from these results.

\textbf{Polynomial dictionary with uncorrupted nonlinearity}.
Assumption~\ref{assumption:smooth_dictionary} suffices to generalize our data cleaning results. For analysis of the error-in-variable estimators, we impose a further assumption, which constrains which kinds of terms can appear as technical regressors. Consider the polynomial dictionary of degree $d_{\max}$, where the only source of nonlinearity is powers and interactions with regressors known to be uncorrupted.

\begin{definition}[Polynomial dictionary with uncorrupted nonlinearity]\label{def:interacted_dictionary}
Suppose the observed regressors consist of one uncorrupted regressor $D_i$ and several corrupted regressors $X_{i,\cdot}$. Consider a polynomial dictionary $b^{\poly}$ of degree $d_{\max}$ such that each basis function $b^{\poly}_k$ is at most linear in the corrupted regressors. By definition, $p'\leq C\cdot d_{\max}p$.
\end{definition}
For example, in Example~\ref{ex:ATE} where $D_i$ is uncorrupted, the interacted dictionary
$
b:(D_i,X_{i,\cdot})\mapsto \{D_iX_{i,\cdot}, (1-D_i)X_{i,\cdot}\}
$
satisfies this property. 
In Example~\ref{ex:deriv} where $D_i$ is uncorrupted, the nonlinear dictionary
$
b:(D_i,X_{i,\cdot})\mapsto (1,D_i,X_{i,\cdot}, D_iX_{i,\cdot}, D_i^2)
$ satisfies this property as well since it contains $D_i^2$ but does not contain $X_{ij}^2$. 
Intuitively, this family of dictionaries avoids compounding measurement error because the corrupted regressors are not multiplied with each other. For readability, we focus on the case of one uncorrupted regressor, which can be conceptualized as
$
b:(D_i,X_{i,\cdot})\mapsto (1,D_i,...,D_i^{d_{\max}},X_{i,\cdot},D_iX_{i,\cdot},...,D_i^{d_{\max}-1}X_{i,\cdot})
$
where $D_i$ is uncorrupted and $X_{i,\cdot}$ are corrupted. Definition~\ref{def:interacted_dictionary} naturally generalizes to the case of multiple uncorrupted regressors. We require three properties to hold after the dictionary is applied to the data.

\begin{assumption}[Dictionary is non-collapsing]\label{assumption:collapse}
The dictionary does not collapse in the following sense.
(i) Recall that we set $k:=rank(\bhX)$ equal to $r:=rank\{\bX^{\lr}\}$. We further assume $k':=rank\{b(D,\bhX)\}$ is equal to $r':=rank[b\{D,\bX^{\lr}\}]$.
 (ii) Assumption~\ref{assumption:well_balanced_spectrum} posits that the smallest singular value of $\bX^{\lr}$ is $s_{r}\geq C \sqrt{\frac{np}{r}}$. We further posit that the smallest singular value of $b\{D,\bX^{\lr}\}$ is $s'_{r'}\geq C\sqrt{\frac{np}{r'}}$. 
(iii) Using the notation of one uncorrupted regressor, the technical regressors $(1,D_i,...D_i^{d_{\max}})$ are full rank.
\end{assumption}
The first property in Assumption~\ref{assumption:collapse} ensures two matrices of equal rank get mapped to two new matrices of equal rank. The second property imposes that singular values, after dictionary mapping, remain well balanced. In particular, we allow for a weaker signal to noise ratio for technical regressors since $r'\geq r$. We do \textit{not} impose $s'_{r'}\geq C\sqrt{\frac{np'}{r'}}$, which is a stronger and less plausible requirement since it implies that the signal to noise ratio increases with the dictionary dimension $p'$. The third property is a technical assumption which allows the theory of implicit data cleaning to generalize. 

Appendices~\ref{sec:regression} and~\ref{sec:riesz} generalize to accommodate nonlinear dictionaries under this additional assumption. See the previous draft for explicit algebra. We turn to proofs.


\begin{lemma}\label{lem:Cb1}
For $b^{\poly}$, $C_b' \le 2^{d_{\max}} \cdot  \|\bM^{(1)}\|^{2d_{\max}}_{\max} \cdot \|\bM^{(2)}\|^{2d_{\max}}_{\max}$.
\end{lemma}

\begin{proof}
We introduce the notation $[b^{\poly}(\bM)]_{ik}=\prod_{\{j(k)\}} M_{ij(k)}$, where $j(k)\in [p]$, $M_{ij(k)} \in \{M_{i1}, \dots, M_{ip}\}$, and $|\{j(k)\}|=d(k)$. We will simplify notation in the following way. Fix $k$. Let $M_{i\ell}$ refer to the $\ell$-th element of the product, where $\ell\in [d(k)]$. Therefore
$
[b^{\poly}(\bM)]_{ik}=\prod_{\{j(k)\}} M_{ij(k)}=\prod_{\ell=1}^{d(k)} M_{i\ell}.
$
Then for any column $k \in [p']$,
\begin{align*}
    &\|b(\bM^{(1)})_{\cdot, k} - b(\bM^{(2)})_{\cdot, k} \|^2_2 
    = \sum_{i=1}^n \left( \prod_{\ell=1}^{d(k)} M^{(1)}_{i\ell} - \prod_{\ell=1}^{d(k)} M^{(2)}_{i\ell} \right)^2
    \\ 
    &\le 2\sum_{i=1}^n \left( \prod_{\ell=1}^{d(k)} M^{(1)}_{i\ell} -  M^{(2)}_{i1} \prod_{\ell=2}^{d(k)} M^{(1)}_{i\ell} \right)^2+ 2\sum_{i=1}^n \left(  \prod_{\ell=1}^{d(k)} M^{(2)}_{i\ell} - M^{(2)}_{i1} \prod_{\ell=2}^{d(k)} M^{(1)}_{i\ell} \right)^2.
\end{align*}

The first term equals
$
   \sum_{i=1}^n \left(M^{(1)}_{i1} - M^{(2)}_{i1}\right)^2  \left(\prod_{\ell=2}^{d(k)} M^{(1)}_{i\ell} \right)^2 
 \le \|\bM^{(1)}\|_{\max}^{2 d_{\max}} \sum_{i=1}^n \left(M^{(1)}_{i1} - M^{(2)}_{i1}\right)^2   
   \le \|\bM^{(1)}\|_{\max}^{2 d_{\max}} \|\bM^{(1)} - \bM^{(2)}\|^2_{2, \infty}.
$
The second term equals
$
   \sum_{i=1}^n \left(  M^{(2)}_{i1} \left(\prod_{\ell=2}^{d(k)} M^{(2)}_{i\ell} - \prod_{\ell=2}^{d(k)} M^{(1)}_{i\ell}\right)\right)^2
    \le \| \bM^{(2)}\|_{\max}^2 \sum_{i=1}^n    \left(\prod_{\ell=2}^{d(k)} M^{(2)}_{i\ell} - \prod_{\ell=2}^{d(k)} M^{(1)}_{i\ell}\right)^2$. Recursing with $\sum_{i=1}^n \left(\prod_{\ell=2}^{d(k)} M^{(2)}_{i\ell} - \prod_{\ell=2}^{d(k)} M^{(1)}_{i\ell}\right)^2$ gives the desired result.
\end{proof}

\begin{lemma}\label{lem:Cb2}
For $b^{\poly}$, $C_b'' \le d_{\max}$.
\end{lemma}

\begin{proof}
Fix $\bM\in\R^{n\times p}$ with rank $r$. 
For notational simplicity, let $M_{i \ell}$ refer to the $\ell$-th element of the product in $[b^{\poly}(\bM)]_{ik}$.
%
%
Observe that $b^{\poly}(\bM)$ can be equivalently represented as
$
b^{\poly}(\bM)=\bB^{(1)}\odot,...,\odot \bB^{(d_{\max})},
$
where $\odot$ means Hadamard product, $\bB^{(\ell)}\in\R^{n\times p'}$, and for $\ell \in [d_{\max}], i \in [n], k \in [p']$,  $[\bB^{(\ell)}]_{ik} =M_{i \ell}$ if  $\ell \le d(k) $ and $[\bB^{(\ell)}]_{ik} =1$ if $ \ell > d(k)$.
Since each column of each $\bB^{(\ell)}$ is either a column of $\bM$ or a column of ones, it has rank at most $r$. 
The rank of a Hadamard product is bounded by the product of ranks and so
$
rank\{b^{\poly}(\bM)\}\leq \prod_{\ell=1}^{d_{\max}}r=r^{d_{\max}}.
$
\end{proof}

\begin{lemma}\label{lem:Cb3}
For $b^{\poly}$, $C_b''' \le d_{\max}$.
\end{lemma}

\begin{proof}
Denote $v\in \R^p$ with $\|v\|_{\infty}\leq \bar{A}$. Then
$
b^{\poly}_k(v)=\prod_{\ell=1}^{d(k)} v_{\ell}\leq\bar{A}^{d_{\max}}.
$
\end{proof}

\begin{proof}[Proof of Proposition~\ref{prop:poly_dict}]
Immediate from Lemmas~\ref{lem:Cb1},~\ref{lem:Cb2}, and~\ref{lem:Cb3}.
\end{proof}

%
\begin{lemma}\label{lemma:wlog}
If Assumption~\ref{assumption:bounded} holds, then $\| \bX^{\lr} \|_{\max} \leq 3\bar{A}$.
\end{lemma}

\begin{proof}
Suppose we have $\bX^{\lr}$ with rank $r$ such that $\|\bX^{\lr}\|_{\max}>3\bar{A}$. By reverse triangle inequality
$
\Delta_{E,\bX^{\lr}}=\|\bX^{\lr}-\bX\|_{\max}\geq \|\bX^{\lr}\|_{\max}-\|\bX\|_{\max}>2\bar{A}.
$
We construct $\bB^{\lr}$ with rank $r$ such that $\|\bB^{\lr}\|_{\max}\leq 3\bar{A}$ and $\Delta_{E,\bB^{\lr}} < \Delta_{E,\bX^{\lr}}$. Set
$
\bB^{\lr}=\frac{\bar{A}}{\|\bX^{\lr}\|_{\max}}\cdot \bX^{\lr}.
$
Clearly, $rank\{\bB^{\lr}\}=rank\{\bX^{\lr}\}$. By construction $\|\bB^{\lr}\|_{\max}\leq \bar{A}$, so
$
\Delta_{E,\bB^{\lr}}=\|\bB^{\lr}-\bX\|_{\max} \leq \|\bB^{\lr}\|_{\max}+\|\bX\|_{\max}\leq 2\bar{A}.
$
\end{proof}

\begin{proof}[Proof of Proposition~\ref{prop:preserve_approx}]
By definition, $r=rank\{\bX^{\lr}\}$. The first result follows directly from Proposition~\ref{prop:poly_dict}. To see the second result, consider the case where $d_{\max}=2$. Then any higher order entry of  $b(\bX)-b(\bX^{\lr})$ is of the form $
    |X_{ij}X_{ik}-X^{\lr}_{ij}X^{\lr}_{ik}|\leq |X_{ij}X_{ik}-X^{\lr}_{ij}X_{ik}|+ |X^{\lr}_{ij}X_{ik}-X^{\lr}_{ij}X^{\lr}_{ik}|\leq \bar{A}\Delta_E+3\bar{A}\Delta_E
$
by Lemma~\ref{lemma:wlog}. More generally, there are $d_{\max}$ such terms, and the largest is of the form $(3\bar{A})^{d_{\max}} \Delta_E$.
\end{proof}

%% file: G_matrix.tex
\section{Data cleaning supporting details}\label{sec:cleaning_proof}

Define the unit ball $\mathbb{B}^p=\{v\in\mathbb{R}^p:\|v\|_2\leq 1\}$ and sphere $\mathbb{S}^{p-1}=\{v\in\mathbb{R}^p:\|v\|_2= 1\}$.

\begin{proof}[Proof of Proposition~\ref{prop:fill}]
Immediate from the law of iterated expectations.
\end{proof}

\begin{proposition}[Bound on $\|\bhA\|_{\max}$]\label{prop:bhA_max}
Suppose $k=r$ and $\hat{s}_1,...,\hat{s}_r \leq C\sqrt{\frac{np}{r}}$. Assume the following incoherence conditions for the corrupted singular vectors: 
$
\|\bhU_r\|_{\max}\leq C n^{-1/2}$ and  $\|\bhV_r\|_{\max}\leq C p^{-1/2}.
$
Then
$
\|\bhA\|_{\max}\leq C r^{1/2}.
$
\end{proposition}

The condition $\hat{s}_1,...,\hat{s}_r \leq C\sqrt{\frac{np}{r}}$ holds with high probability under $s_1,...,s_r \leq C\sqrt{\frac{np}{r}}$ by Weyl's inequality, similar to Proposition~\ref{prop:gtr}. The condition $s_1,...,s_r \leq C\sqrt{\frac{np}{r}}$ complements Assumption~\ref{assumption:well_balanced_spectrum}. To interpret the incoherence conditions, note that $U_{\cdot,j}\in\R^n$ and $V_{\cdot,j}\in \R^p$.

\begin{proof}
Write
$
\hat{A}_{ij}=\sum_{\ell=1}^r  \hat{U}_{i\ell} \hat{s}_{\ell} \hat{V}_{j\ell}.
$
Hence
$
|\hat{A}_{ij}|\leq \sum_{\ell=1}^r  |\hat{U}_{i\ell}|\cdot |\hat{s}_{\ell}|\cdot |\hat{V}_{j\ell}|.
$
\end{proof}

\begin{lemma}\label{lemma:masked_noise_operator_norm}
Under Assumption~\ref{assumption:missing},
$
    \|\E[(\bZ-\bA\boldsymbol{\rho})^T(\bZ-\bA\boldsymbol{\rho})]\|
    \leq \rho_{\max}(1-\rho_{\min})\big(\max_{j\in[p]}\|A_{\cdot,j}\|_2^2 + {\|diag(\E[\bH^T\bH])\|}  \big)
    +\rho_{\max} \|\E[\bH^T\bH]\|,
$
where $\rho_{\max} \coloneqq \max_{j\in [p]} \rho_j \le 1$.
\end{lemma}

\begin{proof} 
Write $
    \E[(\bZ -  \bA \pmb{\rho})^T (\bZ - \bA \pmb{\rho})] = \sum_{\ell=1}^n \E[ (Z_{\ell,\cdot} - A_{\ell,\cdot} \pmb{\rho}) \otimes (Z_{\ell,\cdot} - A_{\ell,\cdot} \pmb{\rho}) ]$.
Let $\bX = \bA + \bH$. For any $(\ell, j) \in [n]\times [p]$, 
$
    \E[Z_{\ell j}] = \rho_j A_{\ell j}$ and $
    \E[Z_{\ell j}^2] = \rho_j \E[X_{\ell j}^2]. 
$
Fix a row $\ell \in [n]$ and denote $\bW^{(\ell)} = (Z_{\ell,\cdot} - A_{\ell,\cdot} \pmb{\rho}) \otimes (Z_{\ell,\cdot} - A_{\ell,\cdot} \pmb{\rho})$.  
By linearity of expectations, 
$
    \E[W_{ij}^{(\ell)}] = \E[Z_{\ell i} Z_{\ell j}] - \rho_j \E[ Z_{\ell i} ]A_{\ell j} - \rho_i \E[ Z_{\ell j} ]A_{\ell i} + \rho_i\rho_j  A_{\ell i} A_{\ell j}.
$
Suppose $i = j$, then 
$
    \E[W_{ii}^{(\ell)}] = \rho_i \E[X_{\ell i}^2] - \rho_i^2 A_{\ell i}^2
    = \rho_i (1-\rho_i) \E[X_{\ell i}^2] + \rho_i^2 \E[(X_{\ell i} - A_{\ell i})^2]. 
$
On the other hand, if $i \neq j$,  
$
    \E[W_{ij}^{(\ell)}] \leq \sqrt{\rho_i\rho_j} \E[ (X_{\ell i} - A_{\ell i}) (X_{\ell j} - A_{\ell j})]
$
since $\E[Z_{\ell i} Z_{\ell j}]=\E[\pi_{i\ell} \pi_{\ell j}] \E[X_{\ell i} X_{\ell j}]\leq \sqrt{\E[\pi^2_{\ell i}]}\sqrt{\E[\pi^2_{\ell j}]}\E[X_{\ell i} X_{\ell j}]=\sqrt{\rho_i\rho_j}\E[X_{\ell i} X_{\ell j}].$
Therefore, we can bound $\bW^{(\ell)}$ as the sum of two matrices and hence
$
    \E[\bW^{(\ell)}] \leq  \rho_{\max}(1-\rho_{\min}) \E[diag(X_{\ell,\cdot} \otimes X_{\ell,\cdot})] + \rho_{\max} \E[H_{\ell,\cdot} \otimes H_{\ell,\cdot}]. 
$
Summing over all rows $\ell \in [n]$ yields
$
    \E[ (\bZ - \bA \pmb{\rho})^T (\bZ - \bA \pmb{\rho})] \leq \rho_{\max} (1-\rho_{\min}) diag(\E[\bX^T \bX]) + \rho_{\max} \E[ \bH^T \bH]. 
$
To complete the proof, we apply triangle inequality: 
$
    \norm{ \E[(\bZ - \bA \pmb{\rho})^T (\bZ - \bA \pmb{\rho})]} \le 
    \rho_{\max} (1- \rho_{\min}) \norm{ diag(\E[\bX^T \bX])} + \rho_{\max} \norm{ \E[\bH^T \bH]}
$
and since $\bH$ is zero mean,
$
 \norm{ diag(\E[\bX^T \bX])} \le \norm{ diag(\bA^T \bA)} + \norm{ diag(\E[\bH^T \bH])}.
$
\end{proof}

\begin{lemma}[Lemma H.2 of \cite{agarwal2021robustness}]\label{lem:technical_Ka}
	Suppose that $X \in \mathbb{R}^n$ and $P \in \{0, 1\}^n$ are random vectors. Then for any $a \geq 1$,
	$\norm{ X \odot P }_{\psi_a} \leq \norm{ X }_{\psi_a}.$
\end{lemma}

\begin{lemma}\label{lemma:masked_noise_row_norm}
Under Assumptions~\ref{assumption:bounded},~\ref{assumption:measurement}, and~\ref{assumption:missing}, 
$
\|Z_{i,\cdot}-A_{i,\cdot}\boldsymbol{\rho}\|_{\psi_a}\leq K_a+\bar{A}\bar{K}.
$
\end{lemma}
\begin{proof}
To begin, write
$
\|Z_{i,\cdot}-A_{i,\cdot}\boldsymbol{\rho}\|_{\psi_a}\leq \|(X_{i,\cdot}-A_{i,\cdot})\odot \pi_{i,\cdot}\|_{\psi_a}+\| A_{i,\cdot}\odot \pi_{i,\cdot}-A_{i,\cdot}\boldsymbol{\rho}\|_{\psi_a}.
$
By Lemma \ref{lem:technical_Ka} and Assumption~\ref{assumption:measurement},
$
\|(X_{i,\cdot}-A_{i,\cdot})\odot \pi_{i,\cdot}\|_{\psi_a}\leq \|(X_{i,\cdot}-A_{i,\cdot})\|_{\psi_a}=\|H_{i,\cdot}\|_{\psi_a}\leq K_a.
$
By the definition of $\|\cdot\|_{\psi_a}$ and Assumption~\ref{assumption:bounded},
$
\| A_{i,\cdot}\odot \pi_{i,\cdot}-A_{i,\cdot}\boldsymbol{\rho}\|_{\psi_a} 
=\sup_{u\in\mathbb{B}^{p}}\left\|\sum_{j=1}^p u_j A_{ij}(\pi_{ij}-\rho_j)\right\|_{\psi_a}
$
$
=\bar{A}\sup_{u\in\mathbb{B}^{p}}\left\|\sum_{j=1}^p u_j \frac{A_{ij}}{\bar{A}}(\pi_{ij}-\rho_j)\right\|_{\psi_a}.
$
Let $v_j = u_j \frac{A_{ij}}{\bar{A}}$. Since $v\in \mathbb{B}^{p}$, we have
$$
   \sup_{u\in\mathbb{B}^{p}}\left\|\sum_{j=1}^p u_j \frac{A_{ij}}{\bar{A}}(\pi_{ij}-\rho_j)\right\|_{\psi_a}
    \le \sup_{v\in\mathbb{B}^{p}}\left\|\sum_{j=1}^p v_j(\pi_{ij}-\rho_j)\right\|_{\psi_a} 
    =\left\|\pi_{i,\cdot}-(\rho_1,...,\rho_p)\right\|_{\psi_a} 
    \leq \bar{K},
$$
using Assumption~\ref{assumption:missing}.
\end{proof}

\begin{lemma}[Proposition H.1 of \cite{agarwal2021robustness}]\label{lem:spectral_upper_bound}
	Let $\bW \in \mathbb{R}^{n \times p}$ be a random matrix whose rows $\bW_{i, \cdot}$ are independent 
	$\psi_a$-random vectors for some $a \geq 1$. Then for any $\tau > 0$, with probability at least $ 1 - \frac{2}{n^{1 + \tau}p^{\tau}}$, 
	$
		\norm{\bW} \leq \norm{ \E \bW^T \bW }^{1/2}
		+ \sqrt{(1+\tau) p } \max_{i \in [n]}\norm{ \bW_{i, \cdot} }_{\psi_a}
		\Big\{ 1 + \big(2 + \tau \big) \ln(np) \Big\}^{\frac{1}{a}} \sqrt{ \ln(np) } 
	$. 
\end{lemma}

\begin{proposition}\label{prop:E1_whp}
Under Assumptions~\ref{assumption:bounded},~\ref{assumption:measurement}, and~\ref{assumption:missing}
$\p(\Ec_1^c )\leq \frac{2}{n^{11}p^{10}}< \frac{2}{n^{10}p^{10}}.$
\end{proposition}

\begin{proof}
We show that for all $\tau>0$, with probability $1-\frac{2}{n^{1+\tau}p^{\tau}}$, 
$
    \|\bZ-\bA\boldsymbol{\rho}\|  \leq C\sqrt{n}\left(\bar{A} + \kappa + K_a \right) 
    +\sqrt{1+\tau}\sqrt{p}(K_a+\bar{A} \bar{K})\left\{1+(2+\tau)\ln(np)\right\}^{\frac{1}{a}}
\sqrt{\ln(np)}.
$
Setting $\tau=10$ and simplifying the bound yields the result. By Lemma~\ref{lemma:masked_noise_operator_norm}, $\|\E[(\bZ-\bA\boldsymbol{\rho})^T(\bZ-\bA\boldsymbol{\rho})]\| \leq \max_{j\in[p]}\|A_{\cdot,j}\|_2^2 + {\|diag(\E[\bH^T\bH])\|}  +\|\E[\bH^T\bH]\|.$ We bound these terms as $n\bar{A}^2$, $n C K_a$, and $n\kappa^2$, respectively, then plug them and Lemma~\ref{lemma:masked_noise_row_norm} into  
Lemma~\ref{lem:spectral_upper_bound}. 
\end{proof}

\begin{lemma}[Lemma H.4 of \cite{agarwal2021robustness}]\label{lem:ind_sum}
	Let $X_1, \ldots, X_n$ be independent random variables with mean zero.
	For $a \geq 1$, 
$\norm{\sum_{i=1}^n X_i}_{\psi_a} \leq C \left( \sum_{i=1}^n \norm{X_i}_{\psi_a}^2 \right)^{1/2}$.
\end{lemma}

\begin{lemma}\label{prop:H9}
Under Assumptions~\ref{assumption:bounded},~\ref{assumption:measurement}, and~\ref{assumption:missing}, 
$
\|Z_{\cdot,j}-\rho_jA_{\cdot,j}\|_{\psi_a}\leq C(K_a+\bar{A}\bar{K}). 
$
\end{lemma}

\begin{proof}
Write
$
\norm{ Z_{\cdot, j} - \rho_j A_{\cdot, j}}_{\psi_a}
=  \sup_{u \in \mathbb{S}^{n-1}} \norm{ u^T \big( \bZ -  \bA \pmb{\rho} \big) e_j }_{\psi_a}
=  \sup_{u \in \mathbb{S}^{n-1}} \big\| \sum_{i=1}^n u_i \big( Z_{i, \cdot} -  A_{i, \cdot} \pmb{\rho} \big) e_j \big\|_{\psi_a}
$. By Lemma \ref{lem:ind_sum}, its bound is $C \sup_{u \in \mathbb{S}^{n-1}} \big( \sum_{i=1}^n  u_i^2  \norm{ \big( Z_{i, \cdot} -  A_{i, \cdot} \pmb{\rho} \big)  e_j }_{\psi_a}^2 \big)^{1/2} \leq  C \max_{i \in [n]} \norm{ (Z_{i, \cdot} -  A_{i, \cdot}\pmb{\rho})e_j }_{\psi_a}$. The conclusion follows from Lemmas~\ref{lem:technical_Ka} and~\ref{lemma:masked_noise_row_norm}.
\end{proof}

\begin{lemma}[Lemma I.7 of \cite{agarwal2021robustness}]\label{lem:norm_psi_alpha}
Let $W_1, \ldots, W_n$ be a sequence of $\psi_a$-random variables for some $a \geq 1$. For any $t \geq 0$,
 $
\p\left( \sum_{i=1}^n W_i^2 > t \right) \leq 2 \sum_{i=1}^n \exp \left\{ - \left( \frac{t}{n \| W_i \|_{\psi_a}^2 }\right)^{a/2}  \right\}.
$
\end{lemma}

\begin{proposition}\label{prop:H4}
Under Assumptions~\ref{assumption:bounded},~\ref{assumption:measurement}, and~\ref{assumption:missing},
$\p(\Ec_2^c )\leq \frac{2}{n^{10}p^{10}}$
\end{proposition}

\begin{proof}
Fix $j$. 
Write
$
\|Z_{\cdot,j}-\rho_jA_{\cdot,j}\|_2^2=\sum_{i=1}^n W_i^2$, where $W_i=e_i^T(Z_{\cdot,j}-\rho_j A_{\cdot,j}).
$
By Lemmas~\ref{lem:technical_Ka} and~\ref{prop:H9},
$
\|W_i\|_{\psi_a}\leq \|Z_{\cdot,j}-\rho_j A_{\cdot,j}\|_{\psi_a} \leq C(K_a+ \bar{K}\bar{A})
$.
By Lemma \ref{lem:norm_psi_alpha} and the union bound, we arrive at the conclusion.
\end{proof}

\begin{proposition}\label{prop:E3_whp_new}
Under Assumptions~\ref{assumption:bounded},~\ref{assumption:measurement}, and~\ref{assumption:missing}, 
$\p(\Ec_3^c )\leq \frac{2}{n^{10}p^{10}}.$
\end{proposition}

\begin{proof}
The key equality is
$
\|\bU_k \bU_k^T (Z_{\cdot,j}-\rho_j A_{\cdot,j})\|_2^2=\sum_{i=1}^k W_i^2$, where  $W_i=u_{i}^T(Z_{\cdot,j}-\rho_j A_{\cdot,j}).
$
To see that it holds, set $v=Z_{\cdot,j}-\rho_j A_{\cdot,j}$. Then
$
\|\bU_k \bU_k^T v\|_2^2=v^T\bU_k \bU_k^T\bU_k \bU_k^T v=v^T\bU_k \bU_k^T v=W^TW.
$
The rest is analogous to Proposition~\ref{prop:H4}.
\end{proof}

\begin{proposition}\label{prop:E2_whp}
Under Assumption~\ref{assumption:missing},
$\p(\Ec_4^c)\leq \frac{2}{n^{10}p^{10}}.$
\end{proposition}

\begin{proof}
Fix $\delta>1$. Define the event
$
\Ec_{(j)}=\left\{\frac{1}{\delta}\rho_j\leq \hat{\rho}_j\leq \delta \rho_j\right\}.
$
By the Chernoff bound for binary variables,
$
\p(\Ec^c_{(j)}) \leq 2\exp\left(-\frac{(\delta-1)^2}{2\delta^2}n\rho_j\right)\leq 2\exp\left(-\frac{(\delta-1)^2}{2\delta^2}n\rho_{\min}\right).
$
Hence by De Morgan's law and the union bound 
$
\p(\Ec_4^c)=\p\left(\left\{\bigcap_{j\in[p]}\Ec_{(j)}\right\}^c\right)= \p\left(\bigcup_{j\in[p]}\Ec^c_{(j)}\right) \leq  2p\exp\left(-\frac{(\delta-1)^2}{2\delta^2}n\rho_{\min}\right).
$
Solve
$
\frac{2}{n^{10}p^{10}}\geq \frac{2}{n^{11}p^{10}}= 2p\exp\left(-\frac{(\delta-1)^2}{2\delta^2}n\rho_{\min}\right)
$
for $\delta$.
\end{proof}

\begin{proposition}\label{prop:E5_whp}
Under Assumption~\ref{assumption:missing},
$\p(\Ec_5^c)\leq \frac{2}{n^{10}p^{10}}.$
\end{proposition}

\begin{proof}
Define the event
$
\mathcal{E}_{(j)}=\{|\hat{\rho}_j-\rho_j|\leq t\}.
$
By Hoeffding's inequality for bounded variables, 
$
\p(\Ec^c_{(j)})\leq 2 \exp\left(-2nt^2\right).
$
By De Morgan's law and the union bound,  
$
\p(\Ec_5^c)=\p\left(\left\{\bigcap_{j\in[p]}\Ec_{(j)}\right\}^c\right)= \p\left(\bigcup_{j\in[p]}\Ec^c_{(j)}\right) \leq  2p\exp\left(-2nt^2\right).
$
To arrive at the desired result, solve
$
\frac{2}{n^{10}p^{10}}\geq \frac{2}{n^{11}p^{10}}= 2p\exp\left(-2nt^2\right)
$
for $t$.
\end{proof}

%% file: H_regression.tex
\section{Error-in-variable regression supporting details}\label{sec:regression_proof}

We propose a new norm for error-in-variable regression analysis $\mathcal{R}(\hat{\gamma})$ (Theorem~\ref{theorem:fast_rate}). Our norm is essentially an on-average generalization error, and we demonstrate its compatibility with semiparametric theory (Theorem~\ref{thm:dml_inid}). To prove the guarantee, we combine an analysis of \te\err\, (Proposition~\ref{prop:gen_test}) with a new theory of implicit data cleaning (Proposition~\ref{prop:gen_geom}). The former builds on an analysis of \tr\err\, (Proposition~\ref{prop:param_est}). 

In the PCR literature, the norm called \tr\err\, is standard. The norm called  \te\err\, is similar in spirit to \cite{agarwal2020principal}. Our norm $\mathcal{R}(\hat{\gamma})$ extends these ideas.

We also propose a new norm for error-in-variable balancing weight analysis $\mathcal{R}(\hat{\alpha})$ (Theorem~\ref{theorem:fast_rate_RR}) that is compatible with semiparametric theory (Theorem~\ref{thm:dml_inid}). Our definitions of balancing weight \tr\err\,(Proposition~\ref{prop:param_est_RR2}),  \te\err\, (Proposition~\ref{prop:gen_test_RR}), and implicit data cleaning (Proposition~\ref{prop:gen_geom_RR}) norms all appear to be conceptual innovations.

\begin{proof}[Proof of Proposition~\ref{prop:indep}]
For $i\in \te$,
$
\hgamma(D_i,Z_{i,\cdot})=b(D_i,Z_{i,\cdot}\bhrho^{-1})\hbeta,
$
which equals
$$
\begin{bmatrix}D_iZ_{i,\cdot}\bhrho^{-1} &  (1-D_i)Z_{i,\cdot}\bhrho^{-1} \end{bmatrix} \begin{bmatrix}\hbeta^{\treat} \\ \hbeta^{\untreat} \end{bmatrix} 
=
\begin{bmatrix}D_iZ_{i,\cdot}&  (1-D_i)Z_{i,\cdot} \end{bmatrix} \begin{bmatrix}\bhrho^{-1}\hbeta^{\treat} \\ \bhrho^{-1}\hbeta^{\untreat} \end{bmatrix}.
$$
Both $(\bhrho,\hbeta)$ are calculated from $\tr$, while $(D_i,Z_{i,\cdot})$ and $(D_j,Z_{j,\cdot})$ are i.n.i.d.
\end{proof}


\begin{lemma}[Theorem 4.6.1 of \cite{vershynin2018high}]\label{lemma:subG_matrix}
Let $\bU\in\R^{m\times r}$ whose rows are independent, mean zero, subGaussian, and isotropic with $\|U_{i,\cdot}\|_{\psi_2}\leq K_u$.  Then for any $t\geq 0$, with probability $1-2e^{-t^2}$,
$
\sqrt{m}-CK_u^2(\sqrt{r}+t)\leq s_r(\bU)\leq s_1(\bU)\leq \sqrt{m}+CK_u^2(\sqrt{r}+t).
$
\end{lemma}

\begin{proposition}[Verifying row space inclusion]\label{prop:inclusion_row}
By hypothesis, $rank\{\bA^{\lr}\}=r$, so it admits a representation $A^{\lr}_{ij}=\langle u_i,v_j, \rangle $ where $u_i,v_j\in\R^r$. Suppose $\{u_i\}$ are independent, mean zero, subGaussian, and isotropic with $\|u_i\|_{\psi_2}\leq K_u$. If $m \gg K_u^4 \cdot r\ln(mp)$ then with probability $1-O\{(mp)^{-10}\}$, $\row\{\bA^{\lr,\tr}\} =\row\{\bA^{\lr,\te}\}$.
\end{proposition}

\begin{proof}
Consider $\bA^{\lr,\tr}$. Let $\bU$ have rows $\{U_{i,\cdot}\}$. By Lemma~\ref{lemma:subG_matrix} with $t=\ln^{\frac{1}{2}}(mp)$,
$
s_r(\bU) \geq \sqrt{m}-CK_u^2\{\sqrt{r}+\ln^{\frac{1}{2}}(mp)\} \gg 0.
$
With high probability, $s_r(\bU)\gg0$, implying that $\{U_{i,\cdot}\}$ are full rank:  $\row(\bU)=\R^r$. Consider $\bA^{\lr,\te}$. Let $\bU'$ have rows $\{U'_{i,\cdot}\}$. Fix $i\in \te$. Since $U'_{i,\cdot}\in \R^r=\row(\bU)$, there exists some $\lambda \in  \R^r$ such that $U'_{i,\cdot} = \sum_{k=1}^r \lambda_k U_{k,\cdot}$. Therefore
$
A^{\lr,\te}_{ij}=\langle U'_{i,\cdot},V_{\cdot,j} \rangle=\left\langle \sum_{k=1}^r \lambda_k U_{k,\cdot},V_{\cdot,j} \right\rangle=\sum_{k=1}^r \lambda_k \langle U_{k,\cdot},V_{\cdot,j} \rangle=\sum_{k=1}^r \lambda_k A^{\lr,\tr}_{kj}.
$
Thus for any $i\in \te$, $A^{\lr,\te}_{i,\cdot}\in\row\{\bA^{\lr,\tr}\}$. Therefore $\row\{\bA^{\lr,\te}\}\subset\row\{\bA^{\lr,\tr}\}$. Likewise for the other direction.
\end{proof}

The results for $\tilde{\Ec}_1$, $\tilde{\Ec}_2$, $\tilde{\Ec}_3$ follow from the results for $\Ec$. We focus on $\tilde{\Ec}_4$ and $\tilde{\Ec}_5$.

\begin{proposition}\label{prop:gtr}
If Assumptions~\ref{assumption:bounded},~\ref{assumption:measurement},~\ref{assumption:missing}, and~\ref{assumption:well_balanced_spectrum} hold, $k=r$, and
$
    \rho_{\min} \gg \tilde{C}\sqrt{r}\ln^{\frac{3}{2}}(np)\left(\frac{1}{\sqrt{p}}\vee\frac{1}{\sqrt{n}} \vee \Delta_E \right)$, where $\tilde{C}:= C  \bar{A} \Big(\kappa + \bar{K} + K_a \Big),
$
then $\p(\tilde{\Ec}_4^c)\leq\frac{C}{n^{10}p^{10}}$.
\end{proposition}

\begin{proof}
By Lemma~\ref{lem:operator_norm_noise_new_bound}, with probability at least $1-O\{1/(np)^{10}\}$, $|\hat{s}_r-s_r| \leq \Delta$. Hence $\hat{s}_r \ge s_r - \Delta$. We want to show $\Delta = o(s_r)$, i.e. $\Delta\leq c_n s_r$ where $c_n\rightarrow 0$. It suffices to show $\frac{\Delta}{s_r}\rightarrow 0$. In such case, $\hat{s}_r \ge s_r - \Delta \geq s_r - c_n s_r = (1-c_n) s_r$, i.e. $\hat{s}_r \gtrsim  s_r$ as desired. We upper bound $\Delta$ using Lemma~\ref{lem:operator_norm_noise_new_bound} and lower bound $
s_r\geq C \sqrt{\frac{np}{r}}
$
using Assumption~\ref{assumption:well_balanced_spectrum} to derive the stated sufficient condition on $ \rho_{\min}$.
\end{proof}

\begin{lemma}\label{lem:cross_term_expectation}
If Assumption~\ref{assumption:noise} holds then
$
\E[\langle \bhA (\hbeta - \beta^*) , \varepsilon \rangle|\bA] \le \bar{\sigma}^2 k.
$
\end{lemma}

\begin{proof}
Note that
$
\hbeta=\bhA^{\dagger}Y=\bhA^{\dagger} \{\bA^{\lr} \beta^* + \varepsilon + \phi^{\lr}\}.
$
Since $\varepsilon$ is conditionally independent of $\bhA$, $\bA^{\lr}$, $\beta^*$, and $\phi^{\lr}$ we have
$
\E[\langle \bhA (\hbeta - \beta^*) , \varepsilon \rangle|\bA]
= \E[\langle \bhA \bhA^{\dagger}\varepsilon , \varepsilon \rangle\bA].
$
By properties of trace algebra, conditional independence of $\varepsilon$ from $\bhA$, Assumption \ref{assumption:noise}, and the fact that $rank(\bhA)=k$, 
\begin{multline*}
\E[\langle \bhA \bhA^{\dagger}\varepsilon , \varepsilon \rangle|\bA]
= \E\left[ trace\left(  \bhA \bhA^{\dagger} \varepsilon \varepsilon^T \right) \mid \bA \right]
= trace\left( \E\left[ \bhA \bhA^{\dagger} \mid \bA\right] \E\left[  \varepsilon \varepsilon^T \mid \bA \right] \right) \\
\le \bar{\sigma}^2 trace\left( \E\left[ \bhA \bhA^{\dagger} \mid \bA \right] \right)
= \bar{\sigma}^2 k.
\end{multline*}
\end{proof}

\begin{lemma}[Lemma A.3 of \cite{agarwal2020principal}] \label{lemma:hoeffding_random} 
Let $X \in \R^n$ be random vector with independent mean zero subGaussian random coordinates with $\norm{ X_i }_{\psi_2} \le K$.
Let $a \in \R^n$ be another random vector that satisfies $\norm{a}_2 \le b$ almost surely for some constant $b \ge 0$.
Then for all $t \ge 0$, 
$
	\p \Big( \Big| \sum_{i=1}^n a_i X_i\Big| \ge t \Big) \le 2 \exp\Big(-\frac{ct^2}{K^2 b^2} \Big),
$
where $c > 0$ is a universal constant. 
\end{lemma}

\begin{lemma}[Lemma A.4 of \cite{agarwal2020principal}] \label{lemma:hansonwright_random} 
Let $X \in \R^n$ be a random vector with independent mean zero subGaussian coordinates where $\norm{X_i}_{\psi_2} \le K$. 
Let $\bB \in \R^{n \times n}$ be a random matrix satisfying $\norm{\bB}  \le a$ and $\norm{\bB}_{Fr}^2 \, \le b$ almost surely for some $a, b \ge 0$.
Then for any $t \ge 0$,
$
	\p \left( | X^T \bB X - \E[X^T \bB X] | \ge t \right) \le 2 \cdot \exp \Big\{ -c \min\Big(\frac{t^2}{K^4 b}, \frac{t}{K^2 a} \Big) \Big\}. 
$
\end{lemma} 

\begin{proposition}
If Theorem~\ref{theorem:cov} conditions and Assumption~\ref{assumption:noise} hold, $\p(\tilde{\Ec}_5^c)\leq\frac{C}{n^{10}p^{10}}$.
\end{proposition}

\begin{proof}
    We show that under Assumptions~\ref{assumption:bounded} and~\ref{assumption:noise}, with $k=r$, the following holds with probability at least $1-O\{1/(np)^{10}\}$ with respect to randomness in $\varepsilon$: $
\langle \bhA (\hbeta - \beta^*), \varepsilon\rangle  \leq 
C \bar{\sigma}^2 \ln(np) \left\{r+\|\phi^{\lr}\|_2+\|\beta^*\|_1(\sqrt{n}\bar{A}   +  \| \bhA - \bA\|_{2,\infty} ) \right\}
 $. Simplifying yields the desired result.

 Recall that $\hbeta = \bhV_k \bhS_k^{-1} \bhU_k^T Y$, $\bhA =  \bhU_k\bhS_k \bhV_k^T$, and $Y = \bA^{\lr} \beta^* +\phi^{\lr}+ \varepsilon$. 
Thus, 
$
\bhA \hbeta = \bhU_k\bhS_k \bhV_k^T \bhV_k \bhS_k^{-1} \bhU_k^T Y  =  \bhU_k \bhU_k^T \bA^{\lr} \beta^* +  \bhU_k \bhU_k^T \phi^{\lr}+  \bhU_k \bhU_k^T \varepsilon.
$
Therefore, 
$\langle \bhA (\hbeta - \beta^*), \varepsilon\rangle = \langle  \bhU_k\bhU_k^T \bA^{\lr} \beta^*, \varepsilon\rangle+ \langle  \bhU_k\bhU_k^T \phi^{\lr}, \varepsilon\rangle + \langle  \bhU_k\bhU_k^T \varepsilon, \varepsilon\rangle - \langle  \bhA \beta^*, \varepsilon \rangle. 
$

For the first, second, and fourth terms, we use Lemma~\ref{lemma:hoeffding_random}. Note that
$
\| \bhU_k\bhU_k^T \bA^{\lr} \beta^*\|_2  
\le \|\bA^{\lr} \beta^*\|_2 
\leq 3\sqrt{n}\bar{A} \|\beta^*\|_1
$
since $\bhU_k\bhU_k^T$ is a projection matrix and $\|\bA^{\lr}\|_{2, \infty} \leq 3\sqrt{n}\bar{A}$ due to Lemma~\ref{lemma:wlog}. It follows that
$
\p\left(  \langle \bhU_k \bhU_k^T \bA^{\lr} \beta^*, \varepsilon \rangle  \geq t \right) \leq  \exp\Big( - \frac{c t^2}{n \bar{A}^2 \|\beta^*\|_1^2 \bar{\sigma}^2 } \Big).
$ Similarly, 
$
\| \bhU_k\bhU_k^T \phi^{\lr}\|_2 
\le \|\phi^{\lr}\|_2
$
since $\bhU_k\bhU_k^T$ is a projection matrix. Hence
$
\p\left(  \langle \bhU_k \bhU_k^T \phi^{\lr}, \varepsilon \rangle  \geq t \right)  \leq  \exp\Big( - \frac{c t^2}{ \|\phi^{\lr}\|_2^2 \bar{\sigma}^2 } \Big).
$ Finally, 
$
\| \bhA \beta^*\|_2 \leq \big(\| \bhA - \bA\|_{2,\infty} + \|\bA\|_{2,\infty} \big) \|\beta^*\|_1$. Therefore
$\p\left(  \langle \bhA \beta^*, \varepsilon \rangle  \geq t\right)
 \leq  \exp\left( - \frac{c t^2}{\bar{\sigma}^2 (n\bar{A}^2+\| \bhA - \bA\|^2_{2,\infty}) \|\beta^*\|_1^2  }  \right)$.

 For the third term, we use Lemma~\ref{lemma:hansonwright_random}. Recall that $\varepsilon$ is conditionally independent of $\bhU_k, \bhS_k,  \bhV_k$ since $ \bhA$ is determined by $\bZ$.
Hence
$
\E\big[\langle \bhU_k\bhU_k^T \varepsilon, \varepsilon\rangle | \bA \big]  = \E\big[ \varepsilon^T \bhU_k\bhU_k^T \varepsilon  | \bA \big] 
= \E\big[trace(\varepsilon \varepsilon^T \bhU_k\bhU_k^T)  | \bA \big] 
= trace(\E\big[\varepsilon \varepsilon^T   | \bA \big] \E\big[ \bhU_k\bhU_k^T | \bA \big] ) 
\leq \bar{\sigma}^2 trace(\E\big[ \bhU_k^T \bhU_k | \bA \big] )=\bar{\sigma}^2 k.$
Since $ \bhU_k \bhU_k^T $ is a projection matrix,
 $
    \| \bhU_k \bhU_k^T \|\leq 1 $ and $
    \|\bhU_k \bhU_k^T \|^2_{Fr}=trace(\bhU_k \bhU_k^T \bhU_k \bhU_k^T)=trace(\bhU_k^T\bhU_k)=k.
 $
    Using Lemma \ref{lemma:hansonwright_random}, it follows that for any $t > 0$
$
\p\left(  \langle \bhU_k\bhU_k^T \varepsilon, \varepsilon\rangle   \geq \bar{\sigma}^2 k + t  | \bA \right)  \leq  \exp\Big\{ - c \min\Big(\frac{t^2}{k \bar{\sigma}^4}, \frac{t}{\bar{\sigma}^2}\Big)\Big\}.
$ Hence it also holds unconditionally.

Set each probability equal to $1/(np)^{10}$, solve for $t$, then combine terms.
\end{proof}

%% file: I_riesz.tex
\section{Error-in-variable balancing weight supporting details}\label{sec:riesz_proof}

\textbf{Counterfactual moments}. 
We describe the counterfactual moments for general parameters and general dictionaries. In this appendix, we consider causal parameters of the form
$
\theta_0=\frac{1}{2n}\sum_{i\in \tr,\te} \theta_i$, where $\theta_i=\E[m(W_{i,\cdot},\gamma_0)]$, $ W_{i,\cdot}=(A_{i,\cdot},H_{i,\cdot},\pi_{i,\cdot})
$, and we slightly abuse sample size notation to avoid overloading $m$.
Given a dictionary $b:\R^p\rightarrow \R^{p'}$, define
$
b^{\signal}(W_{i,\cdot})=b(A_{i,\cdot})$ and $b^{\noise}(W_{i,\cdot})=b(Z_{i,\cdot}).
$
\begin{algorithm}[Counterfactual moment with data cleaning]
Given corrupted training covariates $\bZ^{\tr} \in\R^{n\times p}$, the dictionary $b:\R^p\rightarrow \R^{p'}$, and the formula $m:\Wc\times \mathbb{L}_2\rightarrow \R$: 
(i) perform data cleaning on $\bZ^{\tr}$ to obtain $\bhA^{\tr}\in\R^{n\times p}$;
 (ii) for $i\in \tr$ calculate $m_{i,\cdot}=m(W_{i,\cdot},b^{\noise})\in\R^{p'}$;
(iii) for $i\in \tr$, calculate $\hat{m}_{i,\cdot}$ from $m_{i,\cdot}$ by overwriting $Z_{i,\cdot}$ with $\hat{A}_{i,\cdot}$;
(iv) calculate $\hat{M}=\frac{1}{n}\sum_{i\in \tr} \hat{m}_{i,\cdot}$.
\end{algorithm}
To specialize this procedure, it suffices to describe $\hat{m}_i$. We provide the explicit expressions for $\hat{m}_i$ for each leading example in the proof of Proposition~\ref{prop:data_cont} below.

As theoretical devices, we introduce several related objects. First, we define the counterfactual vectors $\tilde{m}_{i,\cdot},\hat{m}_{i,\cdot} \in \R^{p'}$ for observation $i$. The former vector uses clean data, while the latter uses cleaned data. In particular,
$
   \tilde{m}_{i,\cdot}=m(W_{i,\cdot},b^{\signal})$, and $\hat{m}_{i,\cdot}=m(W_{i,\cdot},b^{\noise})
$
 overwriting $Z_{i,\cdot}$ with $\hat{A}_{i,\cdot}$.
We concatenate the vectors $\tilde{m}_{i,\cdot}$ as rows in the matrix $\btM$. We concatenate the vectors $\hat{m}_{i,\cdot}$ as rows in the matrix $\bhM$. We refer to these objects as the counterfactual matrices. We also use the counterfactual vectors to define the counterfactual moments $M^*,\hat{M}\in \R^{p'}$:
$
    M^*=\frac{1}{2n} \sum_{i\in \tr,\te} \alpha_0(W_{i,\cdot})b\{A^{\lr}_{i,\cdot}\}$ and $ \hat{M}=\frac{1}{n}\sum_{i\in \tr} \hat{m}_{i,\cdot}.
$
Finally, we introduce notation for the covariance matrices $\bG^*,\bhG\in \R^{p'\times p'}$:
$
    \bG^*=\frac{1}{2n}\sum_{i\in \tr,\te} b\{A^{\lr}_{i,\cdot}\}^T b\{A^{\lr}_{i,\cdot}\}$ and $ \bhG=\frac{1}{n} b(\bhA^{\tr})^T b(\bhA^{\tr}).
$

A desirable property is that data cleaning guarantees for the corrupted regressors imply data cleaning guarantees of the counterfactual moments. We refer to this property as data cleaning continuity, and verify that it holds for the leading examples.

\begin{assumption}[Data cleaning continuity]\label{assumption:data_cont}
There exist $C_m',C_m''<\infty$ such that
(i)
            $\|\bhM -\btM \|^2_{2, \infty} \le C'_{m} \|\bhA - \bA \|^2_{2, \infty}$;
         (ii) $
\max_{j\in [p']} |\tilde{m}_{ij}|\leq C''_{m}.
$
\end{assumption}


\begin{proposition}[Verifying data cleaning continuity]\label{prop:data_cont}
Suppose Assumption~\ref{assumption:bounded} holds. In Example~\ref{ex:ATE} with the interacted dictionary, $C_m'=1$ and $C_m''=\bar{A}$.
In Example~\ref{ex:LATE} with the interacted dictionary, $C_m'=1$ and $C_m''=\bar{A}$ in the numerator and denominator.
 In Example~\ref{ex:policy} with the identity dictionary, suppose the counterfactual policy is of the form $t:A_{i,\cdot}\mapsto t_1\odot A_{i,\cdot} +t_2$ where $t_1,t_2\in\mathbb{R}^p$. Then $C_m'=(\|t_1\|_{\max}+1)^2$ and $C_m''=(\|t_1\|_{\max}+1)\bar{A}+\|t_2\|_{\max}$.
In Example~\ref{ex:deriv} with the interacted quadratic dictionary, $C_m'=4\bar{A}^2$ and $C_m''=2\bar{A}^2$.\footnote{Likewise for any polynomial of $D_i$ interacted with $Z_{i,\cdot}$.}
In Example~\ref{ex:p_ols} with the partially linear dictionary, $C_m'=0$ and $C_m''=1$.\footnote{Recall that, to estimate a weighted balancing weight, we propose estimating an unweighted balancing weight then applying the weighting. The verification here is for the unweighted balancing weight that will be weighted.}
 In Example~\ref{ex:p_iv} with the partially linear dictionary, $C_m'=0$ and $C_m''=1$ in the numerator and denominator.
 In Example~\ref{ex:CATE} with the interacted dictionary, $C_m'=1$ and $C_m''=\bar{A}$.\footnote{Recall that, to estimate a local balancing weight, we propose estimating a global balancing weight then applying the localization. The verification here is for the global balancing weight that will be localized.}
\end{proposition}

\begin{proof}
In Example~\ref{ex:ATE}, write
    $
    m_{i,\cdot}=b(1,Z_{i,\cdot})-b(0,Z_{i,\cdot})=\{Z_{i,\cdot},0\}-\{0,Z_{i,\cdot}\}=(Z_{i,\cdot},-Z_{i,\cdot}).
    $
    Hence
    $
        \hat{m}_{i,\cdot}=(\hat{A}_{i,\cdot},-\hat{A}_{i,\cdot})$ and $ 
        \tilde{m}_{i,\cdot}=(A_{i,\cdot},-A_{i,\cdot}).
   $
    Example~\ref{ex:LATE} is analogous. In Example~\ref{ex:policy}, write
    $
    m_{i,\cdot}=b\{t(Z_{i,\cdot})\}-b(Z_{i,\cdot})=t(Z_{i,\cdot})-Z_{i,\cdot}=t_1\odot Z_{i,\cdot} + t_2-Z_{i,\cdot}=[(t_1-\1^T)\odot Z_{i,\cdot}] +t_2.
    $
    Hence
    $
        \hat{m}_{i,\cdot}=[(t_1-\1^T)\odot \hat{A}_{i,\cdot}] +t_2$ and $ 
        \tilde{m}_{i,\cdot}=[(t_1-\1^T)\odot A_{i,\cdot}] +t_2.
    $
   In Example~\ref{ex:deriv}, write
$
        m_{i,\cdot}
        =\nabla_d b(D_i,Z_{i,\cdot})
        =\nabla_d (1,D_i,D_i^2,Z_{i,\cdot},D_iZ_{i,\cdot},D_i^2Z_{i,\cdot})
        =(0,1,2D_i,0,Z_{i,\cdot},2D_iZ_{i,\cdot}).
   $
    Hence
   $
        \hat{m}_{i,\cdot}= (0,1,2D_i,0,\hat{A}_{i,\cdot},2D_i\hat{A}_{i,\cdot})$ and $
        \tilde{m}_{i,\cdot}=(0,1,2D_i,0,A_{i,\cdot},2D_iA_{i,\cdot}).
 $
    In Example~\ref{ex:p_ols}, let $b(D_i,Z_{i,\cdot})=\{D_i,\tilde{b}(Z_{i,\cdot})\}$. Write
    $
    m_{i,\cdot}=b(1,Z_{i,\cdot})-b(0,Z_{i,\cdot})=\{1,\tilde{b}(Z_{i,\cdot})\}-\{0,\tilde{b}(Z_{i,\cdot})\}=(1,0,...,0).
    $
    Hence
    $
        \hat{m}_{i,\cdot}=(1,0,...,0)$ and $
        \tilde{m}_{i,\cdot}=(1,0,...,0).
   $
 Example~\ref{ex:p_iv} is analogous to Example~\ref{ex:p_ols}.
Example~\ref{ex:CATE} is analogous to Example~\ref{ex:ATE}.
\end{proof}

\textbf{Properties}. The error-in-variable balancing weight confers balance and equivalence.

\begin{proposition}[Finite sample balance]\label{prop:balance}
For any finite training sample size $n$, and any dictionary $b$, the coefficient $\heta$ balances the cleaned actual regressors with the corresponding cleaned counterfactuals in the sense that
$
\frac{1}{n}\sum_{i\in \tr} b(\hat{A}_{i,\cdot})\cdot \hat{\omega}_i =\frac{1}{n}\sum_{i\in \tr} \hat{m}_{i,\cdot}
$
where $\hat{\omega}_i \in \R$ are balancing weights computed from $\heta$: for each $i\in \tr$,
$
\hat{\omega}_i=b(\hat{A}_{i,\cdot})\heta.
$
\end{proposition}

\begin{proof}
$
\frac{1}{n}\sum_{i \in \tr} b(\hat{A}_{i,\cdot})^Tb(\hat{A}_{i,\cdot})\heta=\bhG\heta=\hat{M}^T=\frac{1}{n}\sum_{i\in \tr}(\hat{m}_{i,\cdot})^T.
$
\end{proof}

In words, $\heta$ serves to balance actual observations with counterfactual queries. 


\begin{proof}[Proof of Proposition~\ref{prop:balance_ATE}]
By Proposition~\ref{prop:balance},
$
\frac{1}{n}\sum_{i\in \tr} b(D_i,\hat{A}_{i,\cdot})\cdot \hat{\omega}_i
    =\frac{1}{n}\sum_{i\in \tr} \hat{m}_{i,\cdot}.
$
By Proposition~\ref{prop:data_cont},
$
\hat{m}_{i,\cdot}=(\hat{A}_{i,\cdot},-\hat{A}_{i,\cdot}).
$
Notice that
$
b(D_i,\hat{A}_{i,\cdot})=\{D_i\hat{A}_{i,\cdot},(1-D_i)\hat{A}_{i,\cdot}\}
$
and
$
\hat{\omega}_i
=b(D_i,\hat{A}_{i,\cdot})\heta
=D_i \cdot \hat{\omega}_i^{\treat} -(1-D_i)\hat{\omega}_i^{\untreat}.
$
Therefore
   $ b(D_i,\hat{A}_{i,\cdot})\cdot \hat{\omega}_i
    =\{D_i\hat{A}_{i,\cdot}\cdot \hat{\omega}_i^{\treat},(1-D_i)\hat{A}_{i,\cdot}\cdot (-\hat{\omega}_i^{\untreat})\}.
    $
In summary,
$
    \frac{1}{n}\sum_{i\in \tr}D_i\hat{A}_{i,\cdot}\cdot \hat{\omega}_i^{\treat} =\frac{1}{n}\sum_{i\in \tr}\hat{A}_{i,\cdot}$ and $
    \frac{1}{n}\sum_{i\in \tr}(1-D_i)\hat{A}_{i,\cdot}\cdot (-\hat{\omega}_i^{\untreat}) =\frac{1}{n}\sum_{i\in \tr} (-\hat{A}_{i,\cdot})$.
\end{proof}

A well-known equivalence holds for treatment effects when using OLS with the interacted dictionary (without data cleaning).\footnote{We thank David Bruns-Smith and Avi Feller for suggesting this connection. See e.g. \cite{ben2021balancing} for a recent summary, and references therein.} We generalize it in three ways: (i) for our entire class of semiparametric and nonparametric estimands, (ii) for any square integrable dictionary, (iii) for estimation with or without data cleaning. We define, for $i\in\tr$,
$
\tgamma(D_i,Z_{i,\cdot})=b(D,\hat{X}_{i,\cdot})\hbeta$ and $\talpha(D_i,Z_{i,\cdot})=b(D,\hat{X}_{i,\cdot})\heta.
$
We also define $\E_{\tr}[\cdot]=\frac{1}{m}\sum_{i\in \tr} [\cdot]$.

\begin{proposition}[Equivalence in \tr]
If Assumption~\ref{assumption:mean_square_cont} holds and $b$ is square integrable, then the outcome, balancing weight, and doubly robust estimators coincide on the training set:
$
\E_{\tr}[m(W_{i,\cdot},\tgamma)]=\E_{\tr}[Y_i\talpha(D_i,Z_{i,\cdot})]=\E_{\tr}[m(W_{i,\cdot},\tgamma)+\talpha(D_i,Z_{i,\cdot})\{Y_i-\tgamma(D_i,Z_{i,\cdot})\}].
$
The same result holds without data cleaning.
\end{proposition}

\begin{proof}
To prove the second equality, we appeal to the first order condition for $\heta$: $\heta^T\bhG =\hat{M}$. Multiplying by $\hbeta$, we have
    $
    \heta^T\bhG\hbeta=\heta^T\E_{\tr}[b(D,\hat{X}_{i,\cdot})^Tb(D,\hat{X}_{i,\cdot})]\hbeta=\E_{\tr}[\talpha(D_i,Z_{i,\cdot})\tgamma(D_i,Z_{i,\cdot})]
    $
    and
    $
    \hat{M}\hbeta=\E_{\tr}[\hat{m}_{i,\cdot}]\hbeta=\E_{\tr}[m(W_{i,\cdot},\tgamma)].
    $
    In summary,
    $
    \E_{\tr}[\talpha(D_i,Z_{i,\cdot})\tgamma(D_i,Z_{i,\cdot})]=\E_{\tr}[m(W_{i,\cdot},\tgamma)]
    $
    which implies the result.
    To prove the first equality, we appeal to the first order condition for $\hbeta$: $\hbeta^T\bhG =\E_{\tr}[Y_ib(D_i,\hat{X}_{i,\cdot})]$. Multiplying by $\heta$, we have
    $
    \E_{\tr}[Y_ib(D_i,\hat{X}_{i,\cdot})]\heta=\E_{\tr}[Y_i \talpha(D_i,Z_{i,\cdot})]
    $
    and, appealing to the previous result,
    $
    \hbeta^T\bhG\heta=\E_{\tr}[\talpha(D_i,Z_{i,\cdot})\tgamma(D_i,Z_{i,\cdot})]=\E_{\tr}[m(W_{i,\cdot},\tgamma)].
    $
\end{proof}

However, our estimator involves sample splitting and implicit data cleaning to break dependence, motivated by our goal of inference after data cleaning. 

\begin{proposition}[Non-equivalence in \te]
If Assumption~\ref{assumption:mean_square_cont} holds and $b$ is square integrable, then the outcome, balancing weight, and doubly robust estimators generically do not coincide on the test set:
$
\E_{\te}[m(Z_{i,\cdot},\hgamma)]\neq \E_{\te}[Y_i\halpha(D_i,Z_{i,\cdot})] \neq \E_{\te}[m(Z_{i,\cdot},\hgamma)+\halpha(D_i,Z_{i,\cdot})\{Y_i-\hgamma(D_i,Z_{i,\cdot})\}].
$
\end{proposition}

\begin{proof}
The first order conditions for $(\hbeta,\heta)$ hold for $\tr$ after data cleaning. They do not hold for $\te$, especially since we do not clean the test covariates.
\end{proof}

\textbf{High probability events}. Define the events:
$$
\Ec_6=\left\{\max_{j\in[p]}\left|\frac{1}{n}\sum_{i\in \tr} \left\{\tilde{m}_{ij} - \E[\tilde{m}_{ij}]\right\}\right| \leq C \cdot C_m''  \sqrt{\frac{\log(np)}{n}} \right\};
$$
$$
\Ec_7=\left\{\max_{j\in[p]}\left|\frac{1}{2n}\sum_{i \in \tr,\te}\left\{\alpha_0(W_{i,\cdot})A_{ij} - \E[\alpha_0(W_{i,\cdot})A_{ij}]\right\} \right| \leq C \cdot \bar{\alpha}\bar{A}  \sqrt{\frac{\log(np)}{n}} \right\};
$$
$$
\Ec_8=\left\{\max_{j,k\in[p]}\left|\frac{1}{n}\sum_{i\in \tr} \left\{A_{ij} A_{ik}-\E[A_{ij} A_{ik}]\right\}\right| \leq C\cdot \bar{A}^2 \sqrt{\frac{\log(np)}{n}}  \right\};
$$
$$
\Ec_9=\left\{\max_{j,k\in[p]}\left|\frac{1}{2n} \sum_{i\in \tr,\te} \left\{A_{ij} A_{ik}-\E[A_{ij} A_{ik}]\right\}\right| \leq C\cdot \bar{A}^2 \sqrt{\frac{\log(np)}{n}} \right\}.
$$

\begin{lemma}\label{lem:R2}
Under Assumption~\ref{assumption:data_cont}, 
$
\p(\Ec_6^c) \leq \frac{2}{n^{10}p^{10}}.
$ Under Assumption~\ref{assumption:bounded} and $\|\alpha_0\|_{\infty}\leq \bar{\alpha}$,
$
\p(\Ec_7^c) \leq \frac{2}{n^{10}p^{10}}.
$ Under Assumption~\ref{assumption:bounded}, 
$
\p(\Ec_8^c) \leq \frac{2}{n^{10}p^{10}}
$ and
$
\p(\Ec_9^c) \leq \frac{2}{n^{10}p^{10}}.
$
\end{lemma}

\begin{proof}
By Assumption~\ref{assumption:data_cont}, $\tilde{m}_{ij} \leq C_m''$. By Assumption~\ref{assumption:bounded}, $|\alpha_0(W_{i,\cdot})A_{ij}|\leq \bar{\alpha}\bar{A}$ and $|A_{ij} A_{ik}|\leq \bar{A}^2$. For $\Ec_6,\Ec_7$ we appeal to Hoeffding for any $j\in[p]$, then take the union bound. For $\Ec_8,\Ec_9$ we appeal to Hoeffding for any $j,k\in[p]$, then take the union bound.
\end{proof}

\begin{lemma}\label{lem:Delta_M}
Suppose Assumptions~\ref{assumption:bounded},~\ref{assumption:data_cont},~\ref{assumption:mean_square_cont}, and~\ref{assumption:invariance} hold, and $\|\alpha_0\|_{\infty}\leq \bar{\alpha}$. Then
$
\|\hat{M}-M^*\|_{\max} | \{\Ec_6,\Ec_7\}\leq \Delta_M=\sqrt{C_m'} \frac{1}{\sqrt{n}}\|\bhA-\bA\|_{2,\infty} +C\cdot (C_m''+\bar{\alpha}\bar{A}) \sqrt{\frac{\ln(np)}{n}}+\bar{\alpha}\cdot \Delta_E.
$
\end{lemma}

\begin{proof}
    Write $  \hat{M}-M^*=\sum_{k=1}^5 R^{(k)}$, where $\{R^{(k)}\}$ are below. By triangle inequality, it suffices to bound $R_j^{(k)}$. Write
    $$
         \{R^{(1)}_j\}^2
         =\left\{\frac{1}{n}\sum_{i\in \tr} (\hat{m}_{ij}-\tilde{m}_{ij})\right\}^2 
         \leq \frac{1}{n}\sum_{i\in \tr} (\hat{m}_{ij}-\tilde{m}_{ij})^2 
         \leq \frac{1}{n} \| \bhM-\btM \|^2_{2,\infty} 
         \leq \frac{1}{n} C_m' \|\bhA-\bA\|^2_{2,\infty},
    $$
    appealing to Assumption~\ref{assumption:data_cont}. Write
    $
    R_j^{(2)}=\frac{1}{n}\sum_{i\in \tr} \{\tilde{m}_{ij} - \E[\tilde{m}_{ij}]\},
    $
    then appeal to $\Ec_6$. Write
    $
    R_j^{(3)}=\frac{1}{n}\sum_{i\in \tr} \E[\tilde{m}_{ij}]-\frac{1}{2n} \sum_{i\in \tr,\te}\E[\alpha_0(W_{i,\cdot})A_{ij}]=0
    $
    by Riesz representation and ex ante identical distribution of folds in the random partition $(\tr,\te)$. In particular, since $b_j^{\signal}\in\mathbb{L}_2(\Wc)$,
    $
    \E[\tilde{m}_{ij}]=\E[m(W_{i,\cdot},b_j^{\signal})]=\E[\alpha_0(W_{i,\cdot})b_j^{\signal}(W_{i,\cdot})]=\E[\alpha_0(W_{i,\cdot})b_j(A_{i,\cdot})].
    $ 
    Write  $
     -R_j^{(4)}=\frac{1}{2n} \sum_{i\in \tr,\te} \{\alpha_0(W_{i,\cdot})A_{ij}-\E[\alpha_0(W_{i,\cdot})A_{ij}]\}
    $
    then appeal to $\Ec_7$. Write
    $
    |R_j^{(5)}|=\left|\frac{1}{2n} \sum_{i\in \tr,\te} \alpha_0(W_{i,\cdot})E^{\lr}_{ij}\right|\leq \bar{\alpha}\cdot  \Delta_E
    $
    where $\alpha_0(W_{i,\cdot}) \leq \bar{\alpha}$.
\end{proof}

\begin{lemma}\label{lem:Delta_G2}
Suppose Assumptions~\ref{assumption:bounded} holds. Then
$
\|\bhG-\bG^*\|_{\max}| \{\Ec_8,\Ec_9\}\leq \Delta_G
$
where
$
\Delta_G=(\bar{A}+\|\bhA-\bA\|_{2,\infty}) \frac{1}{\sqrt{n}}\|\bhA-\bA\|_{2,\infty}+C \cdot \bar{A}^2 \sqrt{\frac{\ln(np)}{n}}+C\cdot \bar{A} \Delta_E.
$
\end{lemma}

\begin{proof}
     Write $  \bhG-\bG^*=\sum_{\ell=1}^7 S^{(\ell)}$, where $\{S^{(\ell)}\}$ are below. By triangle inequality, it suffices to bound $S_{jk}^{(\ell)}$. Write
    $
    S^{(1)}_{jk}=\frac{1}{n}\sum_{i\in \tr} \hat{A}_{ij} (\hat{A}_{ik}-A_{ik})\leq \|\bhA\|_{\max} \cdot \frac{1}{n}\sum_{i\in \tr} (\hat{A}_{ik}-A_{ik})
    $
    hence
   $
    \{ S^{(1)}_{jk}\}^2
    \leq \|\bhA\|^2_{\max} \cdot \left\{\frac{1}{n}\sum_{i\in \tr} (\hat{A}_{ik}-A_{ik})\right\}^2 
    \leq \|\bhA\|^2_{\max} \cdot \frac{1}{n}\sum_{i\in \tr} (\hat{A}_{ik}-A_{ik})^2 
    \leq \|\bhA\|^2_{\max} \frac{1}{n} \|\bhA-\bA\|^2_{2,\infty}.
   $ Then use $
     \|\bhA\|_{\max}\leq \|\bhA-\bA\|_{\max}+\|\bA\|_{\max}.
    $ Write
    $
    S^{(2)}_{jk}=\frac{1}{n}\sum_{i\in \tr} (\hat{A}_{ij}-A_{ij})A_{ik} \leq \bar{A} \cdot \frac{1}{n}\sum_{i\in \tr} (\hat{A}_{ij}-A_{ij}).
    $
    By a similar argument,
   $
    \{ S^{(2)}_{jk}\}^2
    \leq \bar{A}^2 \frac{1}{n} \|\bhA-\bA\|^2_{2,\infty}.
    $ 
    Write $
     S^{(3)}_{jk}=\frac{1}{n}\sum_{i\in \tr} \{A_{ij} A_{ik}-\E[A_{ij} A_{ik}]\}
    $
    then appeal to $\Ec_8$.  
    Write
$
 S^{(4)}_{jk}=\frac{1}{n}\sum_{i\in \tr} \E[A_{ij} A_{ik}] -\frac{1}{2n} \sum_{i\in \tr,\te} \E[A_{ij}A_{ik}]=0
$
by ex ante identical distribution of folds in the random partition $(\tr,\te)$. 
  Write
   $
     -S^{(5)}_{jk}=\frac{1}{2n} \sum_{i\in \tr,\te} \{A_{ij} A_{ik}-\E[A_{ij} A_{ik}]\}
    $
    then appeal to $\Ec_9$.
     By Assumption~\ref{assumption:bounded},
    $
    S^{(6)}_{jk}=\frac{1}{2n} \sum_{i\in \tr,\te} A_{ij} E^{\lr}_{ik} \leq \bar{A} \Delta_E.
    $
     By Assumption~\ref{assumption:bounded} and Lemma~\ref{lemma:wlog}.
    $
    S^{(7)}_{jk}=\frac{1}{2n} \sum_{i\in \tr,\te} E^{\lr} A^{\lr}_{ik}\leq \|\bE^{\lr}\|_{\max} \|\bA^{\lr}\|_{\max}\leq 3\bar{A}\Delta_E.
    $
\end{proof}

\begin{proposition}\label{prop:whp_tilde_e5}
    If the conditions of Proposition~\ref{prop:param_est_RR2} hold, then $\p(\tilde{\Ec}_5^c)\leq\frac{C}{n^{10}p^{10}}$.
\end{proposition}

\begin{proof}
    The result follows from Lemmas~\ref{lem:Delta_M} and~\ref{lem:Delta_G2}.
\end{proof}

%% file: J_target.tex
\section{Data cleaning-adjusted confidence intervals details}\label{sec:target_proof}

\textbf{Riesz representation}. We generalize Assumptions~\ref{assumption:invariance_mod}  and~\ref{assumption:mean_square_cont_mod} from the ATE example to the general case. In doing so, we also generalize the balancing weight to a Riesz representer.

\begin{assumption}[Linearity and mean square continuity]\label{assumption:mean_square_cont}
(i) The functional $\gamma\mapsto \E[m(W_{i,\cdot},\gamma)]$ is linear. 
   (ii) There exists $\bar{Q}<\infty$ and $\bar{q}\in(0,1]$ such that for all $\gamma\in\Gamma,$
    $
    \E[m(W_{i,\cdot},\gamma)^2]\leq \bar{Q}\cdot \{\E[\gamma(W_{i,\cdot})^2]\}^{\bar{q}}.
    $
\end{assumption}

These restrictions generalize the usual propensity score assumptions; Assumption~\ref{assumption:mean_square_cont_mod} is a special case of Assumption~\ref{assumption:mean_square_cont}. Assumption~\ref{assumption:mean_square_cont} implies that the balancing weight exists.

\begin{lemma}[\cite{chernozhukov2018global}]\label{lem:RR}
Suppose Assumption~\ref{assumption:mean_square_cont} holds and $\gamma_0\in \Gamma$, which may be imposed in estimation. Then there exists a Riesz representer $\alpha_0\in \mathbb{L}_2(\Wc)$ such that for all $\gamma\in \Gamma$, 
$
\E[m(W_{i,\cdot},\gamma)]=\E[\alpha_0(W_{i,\cdot})\gamma(W_{i,\cdot})].
$
There exists a unique minimal Riesz representer $\alpha_0^{\min}\in \Gamma$ satisfying this equation. Moreover, denoting by $\bar{M}$ the operator norm of $\gamma\mapsto \E[m(W_{i,\cdot},\gamma)]$, 
$
\{\E[\alpha_0^{\min}(W_{i,\cdot})^2]\}^{\frac{1}{2}}=\bar{M} \leq \bar{Q}^{\frac{1}{2}}<\infty.
$
\end{lemma}

The balancing weight is a special case of a Riesz representer. Hereafter, we refer to the Riesz representer as a balancing weight nonetheless, since our estimator $\halpha$ achieves balance across examples; see Proposition~\ref{prop:balance}, which generalizes Proposition~\ref{prop:balance_ATE}. To lighten notation, we will typically consider the case where $\Gamma=\mathbb{L}_2(\Wc)$ and $\alpha_0^{\min}=\alpha_0$. When we consider the more general case, as in Example~\ref{ex:deriv}, we will use the richer notation.

We impose that $(\gamma_0,\alpha_0)$ do not vary across observations, generalizing familiar distribution shift assumptions. Assumption~\ref{assumption:invariance_mod} is a special case of Assumption~\ref{assumption:invariance}, which we now state.

\begin{assumption}[Marginal distribution shift]\label{assumption:invariance}
For all observations $i\in [n]$, (i) the regression $\gamma_0$ does not vary: $\E[\gamma_0(W_{i,\cdot}) v(W_{i,\cdot})]=\E[Y_iv(W_{i,\cdot})]$ for all $v \in\mathbb{L}_2(\Wc)$; 
(ii) the Riesz representer $\alpha_0$ does not vary: $\E[\alpha_0(W_{i,\cdot}) u(W_{i,\cdot})]=\E[ m(W_{i,\cdot},u)]$ for all $u \in\mathbb{L}_2(\Wc)$.
\end{assumption}

\begin{proposition}[Verifying Assumptions~\ref{assumption:mean_square_cont} and~\ref{assumption:invariance}]\label{prop:mean_square_cont}
Assumptions~\ref{assumption:mean_square_cont} and~\ref{assumption:invariance} hold under simple and interpretable conditions for the leading examples. Recall that $\|\alpha_0\|_{\infty} \leq \bar{\alpha}$, while $(\bar{Q},\bar{q})$ are defined in Assumption~\ref{assumption:mean_square_cont}.
In Example~\ref{ex:ATE}, $\alpha_0(W_{i,\cdot})=\frac{D_i}{\phi_0(X_{i,\cdot},H_{i,\cdot},\pi_{i,\cdot})}-\frac{1-D_i}{1-\phi_0(X_{i,\cdot},H_{i,\cdot},\pi_{i,\cdot})}$, where $ \phi_0(X_{i,\cdot},H_{i,\cdot},\pi_{i,\cdot}):=\E[D_i|X_{i,\cdot},H_{i,\cdot},\pi_{i,\cdot}].$ 
    Suppose
    $
    0<\underline \phi \leq \phi_0(X_{i,\cdot},H_{i,\cdot},\phi_{i,\cdot})\leq \bar{\phi}< 1
    $. Then 
    $
    \bar{\alpha}=\frac{1}{\underline \phi}+\frac{1}{1-\bar{\phi}}$, $\bar{Q}=\frac{2}{\underline \phi}+\frac{2}{1-\bar{\phi}}$, and $\bar{q}=1
    $
   . We impose that the outcome regression and treatment propensity score do not vary.
In Example~\ref{ex:LATE}, $\alpha_0(W_{i,\cdot})=\frac{U_i}{\phi_0(X_{i,\cdot},H_{i,\cdot},\pi_{i,\cdot})}-\frac{1-U_i}{1-\phi_0(X_{i,\cdot},H_{i,\cdot},\pi_{i,\cdot})}$, where $\phi_0(X_{i,\cdot},H_{i,\cdot},\pi_{i,\cdot}):=\E[U_i|X_{i,\cdot},H_{i,\cdot},\pi_{i,\cdot}]$. Suppose 
    $
    0<\underline \phi \leq \phi_0(X_{i,\cdot},H_{i,\cdot},\pi_{i,\cdot})\leq \bar{\phi}< 1
    $. Then 
    $
    \bar{\alpha}=\frac{1}{\underline \phi}+\frac{1}{1-\bar{\phi}}$, $\bar{Q}=\frac{2}{\underline \phi}+\frac{2}{1-\bar{\phi}}$, and $\bar{q}=1
    $
    . We impose that the outcome regression and instrument propensity score do not vary.
   In Example~\ref{ex:policy},
    $
    \alpha_0(W_{i,\cdot})=\omega(X_{i,\cdot},H_{i,\cdot},\pi_{i,\cdot})-1$, where $\omega(X_{i,\cdot},H_{i,\cdot},\pi_{i,\cdot})=\frac{f\{t(X_{i,\cdot},H_{i,\cdot},\pi_{i,\cdot})\}}{f(X_{i,\cdot},H_{i,\cdot},\pi_{i,\cdot})}.
    $ Suppose
    $
  \omega(X_{i,\cdot},H_{i,\cdot},\pi_{i,\cdot})\leq \bar{\omega}<\infty
    $. Then 
    $
    \bar{\alpha}=\bar{\omega}+1$, $\bar{Q}=2\bar{\omega}+2$, and $\bar{q}=1
    $
    . We impose that the outcome regression and covariate density ratio do not vary.
 In Example~\ref{ex:deriv},
    $
    \alpha_0(W_{i,\cdot})=-\nabla_d \ln f(D_i \mid X_{i,\cdot},H_{i,\cdot},\pi_{i,\cdot}).
    $ 
    Suppose $-\nabla_d \ln f(D_i \mid X_{i,\cdot},H_{i,\cdot},\pi_{i,\cdot}) \leq \bar{f}<\infty$. Then 
    $
    \bar{\alpha}=\bar{f}$, $\bar{Q}=\bar{f}(\bar{\gamma}+\bar{\gamma}')$, and $\bar{q}=1/2
    $
    for $\Gamma$ that satisfies a Sobolev condition: $\E[ \{\nabla_d \gamma(D_i,X_{i,\cdot},H_{i,\cdot},\pi_{i,\cdot})\}^2]\leq \bar{\gamma}^2<\infty$ and $\E[ \{\partial^2_d \gamma(D_i,X_{i,\cdot},H_{i,\cdot},\pi_{i,\cdot})\}^2]\leq (\bar{\gamma}')^2<\infty$. We impose that the outcome regression and conditional density of goods do not vary.
   In Example~\ref{ex:p_ols}, 
    $
    \alpha_0(W_{i,\cdot})=\ell_i \frac{D_i-\phi_0(X_{i,\cdot},H_{i,\cdot},\pi_{i,\cdot})}{\E[\{D_i-\phi_0(X_{i,\cdot},H_{i,\cdot},\pi_{i,\cdot})\}^2]}$, where $\phi_0(X_{i,\cdot},H_{i,\cdot},\pi_{i,\cdot}):=\E[D_i|X_{i,\cdot},H_{i,\cdot},\pi_{i,\cdot}].
    $
    Suppose $\E[\{D_i-\phi_0(X_{i,\cdot},H_{i,\cdot},\pi_{i,\cdot})\}^2]>\underline \phi$ and $|\ell_i|\leq\bar{\ell}$. Then
        $
    \bar{\alpha}=\frac{2\bar{\ell}\bar{A}}{\underline \phi}$, $\bar{Q}=\frac{4\bar{\ell}^2\bar{A}^2}{\underline \phi^2}$, and $\bar{q}=1
    $
    . We impose that the outcome regression and treatment regression do not vary.
  In Example~\ref{ex:p_iv}, 
    $
    \alpha_0(W_{i,\cdot})=\ell_i \frac{U_i-\phi_0(X_{i,\cdot},H_{i,\cdot},\pi_{i,\cdot})}{\E[\{U_i-\phi_0(X_{i,\cdot},H_{i,\cdot},\pi_{i,\cdot})\}^2]}$, where $ \phi_0(X_{i,\cdot},H_{i,\cdot},\pi_{i,\cdot}):=\E[U_i|X_{i,\cdot},H_{i,\cdot},\pi_{i,\cdot}]
    $
    . Suppose $\E[\{D_i-\phi_0(X_{i,\cdot},H_{i,\cdot},\pi_{i,\cdot})\}^2]>\underline \phi$ and $|\ell_i|\leq\bar{\ell}$. Then
        $
    \bar{\alpha}=\frac{2\bar{\ell}\bar{A}}{\underline \phi}$, $\bar{Q}=\frac{4\bar{\ell}^2\bar{A}^2}{\underline \phi^2}$, and $\bar{q}=1
    $
    . We impose that the outcome regression and instrument regression do not vary.
    In Example~\ref{ex:CATE},
     $
     \alpha_0(W_{i,\cdot})=\ell_h(V_i)\left\{\frac{D_i}{\phi_0(V_i,X_{i,\cdot},H_{i,\cdot},\pi_{i,\cdot})}-\frac{1-D_i}{1-\phi_0(V_i,X_{i,\cdot},H_{i,\cdot},\pi_{i,\cdot})}\right\}.
     $
     Suppose
    $
    0<\underline \phi \leq \phi_0(V_i,X_{i,\cdot},H_{i,\cdot},\pi_{i,\cdot})\leq \bar{\phi}< 1
    $ and other regularity conditions hold given in Lemma~\ref{lemma:local} below. Then 
    $
    \bar{\alpha}_h\leq C \cdot \frac{1}{h}\left( \frac{1}{\underline \phi}+\frac{1}{1-\bar{\phi}}\right)$, $\bar{Q}_h\leq C \cdot \frac{1}{h^2}\left(\frac{2}{\underline \phi}+\frac{2}{1-\bar{\phi}}\right)$, $\bar{q}=1
    $
   .  We impose that the outcome regression and treatment propensity score do not vary.
\end{proposition}

\begin{proof}[Proof of Proposition~\ref{prop:mean_square_cont}]
The results follow from the law of iterated expectations and integration by parts. See \cite[Lemmas S3 and S4]{chernozhukov2021simple} for mean square continuity of Example~\ref{ex:deriv} and \cite[Theorem 2]{chernozhukov2021simple} for the characterization of $(\bar{\alpha}_h,\bar{Q}_h)$ with localization.
\end{proof}

\textbf{Nonparametrics}. A local functional $\theta_0^{\lim}\in \R$ is a scalar that takes the form
$
\theta^{\lim}_{0}
=\lim_{h\rightarrow 0} \theta_0^h$, where $\theta_0^h=\frac{1}{n}\sum_{i=1}^n\theta^h_i$ and $ \theta_i^h=\E[m_h(W_{i,\cdot},\gamma_0)]
=\E[\ell_h(W_{ij}) m(W_{i,\cdot},\gamma_0)]
$.
Here, $\ell_h$ is a Nadaraya Watson weighting with bandwidth $h$ and $W_{ij}$ is a scalar component of $W_{i,\cdot}$. $\theta^{\lim}_{0}$ is a nonparametric quantity that we approximate by the sequence $\{\theta_0^h\}$, incurring the nonparametric approximation error
$\Delta_h=
n^{1/2} \sigma^{-1}|\theta_0^h-\theta_0^{\lim}|$.
Each $\theta_0^h$ can be analyzed like $\theta_0$ above as long as we keep track of how certain quantities depend on $h$

\begin{lemma}[Theorem 2 of \cite{chernozhukov2021simple}]\label{lemma:local}
If response noise has finite variance then $\bar{\sigma}^2<\infty$. Suppose bounded moment and heteroscedasticity conditions hold. Then for global functionals
$
\xi/\sigma \lesssim \sigma \asymp \bar{M} < \infty$, $\xi, \chi\lesssim \bar{M}^2\leq   \bar{Q}<\infty$, and $\bar{\alpha}<\infty.$
Suppose bounded moment, heteroscedasticity, density, and derivative conditions hold. Then for local functionals
$
\xi_h/\sigma_h \lesssim h^{-1/6}$, $\sigma_h \asymp \bar{M}_h \asymp h^{-1/2}$, $\xi_h\lesssim h^{-2/3}$, $\chi_h\lesssim h^{-3/4}$, $\bar{\alpha}_h\lesssim h^{-1}$, $\bar{Q}_h\lesssim h^{-2}
$,
and $\Delta_h \lesssim n^{1/2} h^{\mathsf{v}+1/2}$ where $\mathsf{v}$ is the order of differentiability of the kernel $K$.
\end{lemma}

Equipped with this lemma, we prove validity of the data cleaning-adjusted confidence interval for nonparametric quantities.

\begin{corollary}[Confidence interval coverage]\label{cor:nonparametric}
Suppose the conditions of Corollary~\ref{cor:CI_inid} and Lemma~\ref{lemma:local} hold. Update the rate conditions to be (i) bandwidth regularity: $n^{-1/2}h^{-3/2}\rightarrow 0$ and $\Delta_h\rightarrow0$;
(ii) error-in-variable regression rate: $\left(h^{-1}+\bar{\alpha}'\right)\{\Rc(\hgamma)\}^{\bar{q}/2}\rightarrow 0$;
(iii) error-in-variable balancing weight rate: $\bar{\sigma}h^{-1}\{\Rc(\halpha)\}^{1/2}\rightarrow 0$;
(iv) product of rates is fast: $h^{-1/2}\{n \Rc(\hgamma) \Rc(\halpha)\}^{1/2}\rightarrow 0$.
Then the conclusions of Corollary~\ref{cor:CI_inid} hold, replacing $(\hat{\theta},\theta_0)$ with $(\hat{\theta}^h,\theta_0^{\lim})$.
\end{corollary}

\begin{proof}[Proof of Corollary~\ref{cor:nonparametric}]
By Lemma~\ref{lemma:local}, the regularity condition on moments is
$
\left\{\left(\kappa/\sigma\right)^3+\zeta^2\right\}n^{-1/2}\lesssim \left\{\left(h^{-1/6}\right)^3+(h^{-3/4})^2\right\}n^{-1/2}\lesssim h^{-3/2} n^{-1/2}.
$
By Lemma~\ref{lemma:local}, the first learning rate condition is
$
\left(\bar{Q}^{1/2}+\bar{\alpha}/\sigma+\bar{\alpha}'\right)\{\mathcal{R}(\hgamma)\}^{1/2} 
\lesssim \left(h^{-1}+h^{-1}/h^{-1/2}+\bar{\alpha}'\right)\{\mathcal{R}(\hgamma)\}^{1/2} 
\lesssim \left(h^{-1}+\bar{\alpha}'\right)\{\mathcal{R}(\hgamma)\}^{1/2}.
$
By \cite[Lemma S.9]{chernozhukov2021simple}, the second learning rate condition is
$
\bar{\sigma}\{\mathcal{R}(\halpha^h)\}^{1/2} \lesssim \bar{\sigma}h^{-1}\{\mathcal{R}(\halpha)\}^{1/2}.
$
By Lemma~\ref{lemma:local} and \cite[Lemma S.9]{chernozhukov2021simple}, the third learning rate condition is
$
\{n \mathcal{R}(\hgamma) \mathcal{R}(\halpha^h)\}^{1/2}  /\sigma 
\lesssim \{n \mathcal{R}(\hgamma) \mathcal{R}(\halpha)\}^{1/2}  h^{-1}/h^{-1/2}
=h^{-1/2}\{n \mathcal{R}(\hgamma) \mathcal{R}(\halpha)\}^{1/2}.
$
\end{proof}

%% file: K_factor.tex
\section{Nonlinear factor model}\label{sec:factor}

We denote $\mathcal{R}_{\gamma}=\mathcal{R}(\hat{\gamma}_{\ell})$ and $\mathcal{R}_{\alpha}=\mathcal{R}(\hat{\alpha_{\ell}})$ to lighten notation. The distinction between $n$ and $m=\frac{n}{2}$ is irrelevant in the context of $(\mathcal{R}_{\gamma},\mathcal{R}_{\alpha})$ due to the absolute constant $C$.

\begin{lemma}[\cite{agarwal2021robustness}]\label{prop:factor_model_approx}
Suppose Assumption~\ref{assumption:factor_model} holds for some fixed $\mathcal{H}(q,S,C_H)$. Then for any small $\delta>0$, there exists $\bA^{(lr)}$ such that
$
r=rank(\bA^{(lr)})\leq C \cdot \delta^{-q}$ and $\Delta_E=\|\bA-\bA^{(lr)}\|_{\max} \leq C_H \cdot \delta^S
$, 
where $C$ may depend on $(q,S)$.
\end{lemma}

\begin{proof}[Proof of Corollary~\ref{cor:GFM}]
From Lemma~\ref{prop:factor_model_approx},
$
r \leq C \cdot \delta^{-q}$ and $\Delta_E\leq C \cdot  \delta^{S}.
$
The conditions of Corollary~\ref{cor:GFM} imply $(\sigma,\bar{\sigma},\bar{\alpha},\bar{\alpha}',\bar{Q})$ are irrelevant. We verify simplified rate conditions from Corollary~\ref{cor:CI_inid}:
$
\mathcal{R}_{\gamma} \rightarrow 0$, $\mathcal{R}_{\alpha} \rightarrow 0$, $\sqrt{n \mathcal{R}_{\gamma} \mathcal{R}_{\alpha}}\rightarrow 0.
$
The relevant terms in $(\mathcal{R}_{\gamma},\mathcal{R}_{\alpha})$ simplify as well. From Theorem~\ref{theorem:fast_rate}, these are
$
    \mathcal{R}_{\gamma}
    \leq 
    C r^3\big\{
    \frac{1}{n}+\frac{p}{n^2}+\frac{1}{p}+\big(1+\frac{p}{n}\big)\Delta_E^2+p\Delta_E^4\big\}.
$
From Theorem~\ref{theorem:fast_rate_RR}, these are 
$
    \mathcal{R}_{\alpha}\leq C r^5
    \left\{  \frac{1}{n}+\frac{1}{p}+\frac{p}{n^2}+\frac{n}{p^2}+\left(1+\frac{p}{n}+\frac{n}{p}\right)\Delta^{2}_E+(n+p)\Delta_E^4+np \Delta^{6}_E
    \right\}.
    $ 

Suppose $n=p^{\upsilon}$ with $\upsilon\geq 1$. Then
$
     \mathcal{R}_{\gamma}
    \leq C r^3\left(
    \frac{1}{p}+\Delta_E^2+p\Delta_E^4\right) 
    \leq C \delta^{-3q}\left(
    \frac{1}{p}+\delta^{2S}+p\delta^{4S}\right).
$ The three terms are equalized with $\delta^{2S}=p^{-1}$. Hence $  \mathcal{R}_{\gamma}
    \leq C \delta^{-3q}
    \frac{1}{p} 
    = C p^{\frac{3q}{2S}} \frac{1}{p} 
    =C p^{\frac{3q}{2S}-1}.$
Similarly
$
    \mathcal{R}_{\alpha}\leq  C r^5
    \left( \frac{n}{p^2}+\frac{n}{p}\Delta^{2}_E+n\Delta_E^4+np \Delta^{6}_E
\right)
\leq  C  \delta^{-5q}  \left(  \frac{n}{p^2}+\frac{n}{p}\delta^{2S}+n\delta^{4S}+np \delta^{6S}
\right).
$
The four terms are equalized with $\delta^{2S}=p^{-1}$. Hence
$
     \mathcal{R}_{\alpha}
    \leq C \delta^{-5q}
    \frac{n}{p^2} 
    = C p^{\frac{5q}{2S}} \frac{n}{p^2} 
    =C p^{\frac{5q}{2S}-2} n.
$ 
To satisfy $\mathcal{R}_{\gamma}\leq \mathcal{R}_{\alpha}\rightarrow 0$, it suffices that
$\frac{q}{S} < \frac{2}{5}(2-\upsilon)$.
To satisfy $\sqrt{n \mathcal{R}_{\gamma}\mathcal{R}_{\alpha}}\rightarrow 0$, it suffices that
$
\frac{q}{S}<\frac{1}{2}\left(\frac{3}{2}-\upsilon \right).
$
In summary, a sufficient generalized factor model is one in which
$
\frac{q}{S}<\frac{2}{5}(2-\upsilon) \wedge \frac{1}{2}\left(\frac{3}{2}-\upsilon\right)$ where $ \upsilon \leq \frac{3}{2}.
$
The latter condition binds for $1\leq \upsilon \leq \frac{3}{2}$. 

If instead $n=p^{\upsilon}$ with $\upsilon\geq 1$, then a similar argument arrives at the same condition.
\end{proof}

When using a polynomial dictionary, the relevant terms in $(\mathcal{R}_{\gamma},\mathcal{R}_{\alpha})$ are as above, instead using $(r',\Delta_E')$. Let $q'=d_{\max}\cdot q$. Then
$
r'\leq C\cdot r^{d_{\max}}\leq C\cdot \delta^{-q d_{\max}}=C\cdot \delta^{-q'}
$
and
$
\Delta_E'\leq C \bar{A}^{d_{\max}} \cdot d_{\max} \Delta_E \leq C\cdot \Delta_E \leq C \cdot  \delta^{S}.
$
Hence the proof of Corollary~\ref{cor:GFM} remains the same.

%% file: L_sim.tex
\section{Simulation and application}\label{sec:sim}

\textbf{Simulation design}. We focus on average treatment effect (ATE) with corrupted covariates (Example~\ref{ex:ATE}). A single observation is a triple $(Y_i,D_i,Z_{i,\cdot})$ for outcome, treatment, and corrupted covariates where $Y\in\R$, $D_i\in\{0,1\}$, and $Z_{i,\cdot}\in\R^{p}$ are generated as follows. 

First, we generate signal from a factor model. Sample $\bU \sim \mathcal{N}(0,\bI_{n\times r})$ and $\bV\sim \mathcal{N}(0,\bI_{p\times r})$. Then set $\bX=\bU \bV^{T}$. By construction, 
$$
\E[X_{ij}]=\E\left[\sum_{s=1}^r U_{is}V_{sj}\right]=\sum_{s=1}^r \E\left[U_{is}\right]\E\left[V_{sj}\right]=0
$$ and $
    \V[X_{ij}]=\V\left[\sum_{s=1}^r U_{is}V_{sj}\right]=\sum_{s=1}^r \V\left[U_{is}\right]\V\left[V_{sj}\right]=r$.

Draw response noise as $\varepsilon_i \overset{i.i.d.}{\sim}\mathcal{N}(0,1)$. Define the vector $\beta\in\R^{p}$ by $\beta_j=j^{-2}$. Then set $D_i\sim \text{Bernoulli}\{ \Lambda (0.25 X^T\beta) \}$ and $ Y_i=2.2D_i+1.2X_{i,\cdot}\beta+D_iX_{i1}+\varepsilon_i$ where $\Lambda(t)=(0.95-0.05)\frac{\exp(t)}{1+\exp(t)}+0.05 $ is the truncated logistic function. The ATE is $\theta_0=2.2$. 

We observe the corrupted covariate
$
Z_{i,\cdot}=[X_{i,\cdot}+H_{i,\cdot}]\odot \pi_{i,\cdot}.
$
$H_{ij}\overset{i.i.d.}{\sim} F_H$ is drawn i.i.d. with mean zero and variance $\sigma_H^2$. $\pi_{ij}$ is $1$ with probability $\rho$ and $\NA$ with probability $1-\rho$. We consider different choices of the measurement error distribution $F_H$ to corresponding to classical measurement error, discretization, and differential privacy. In summary, the three data corruption parameters are $(F_H,\sigma_H,\rho)$. The remaining design parameters are $(n,p,r)$ corresponding to the sample size, dimension of covariates, and rank of the signal. 

For classical measurement error, $F_H=\mathcal{N}(0,\sigma^2_H)$. For discretization, we generate $Z_{ij}=sign(X_{ij})\cdot Poisson(|X_{ij}|)$ and implicitly define $F_H$ by $H_{ij}=Z_{ij}-X_{ij}$. Note that
$
    \E[Z_{ij}|X_{ij}]
    = sign(X_{ij}) \E[Poisson(|X_{ij}|) |X_{ij}]
    =X_{ij}
$
as desired. Below, we show that $\sigma^2_H=\V[H_{ij}]=1.7$ in this construction. For differential privacy, $F_H=Laplace(0,\frac{\sigma_H}{\sqrt{2}})$. 

\begin{proposition}[Discretization variance]
Given some random variable $X$, define $P=\text{Poisson}(|X|)$, $Z=sign(X)\cdot P$, and $H=Z-X$. Then $\E[H]=0$ and $\V[H]=\E[|X|]$.
\end{proposition}

\begin{proof}
To begin, write
 $
    \E[Z|X]
    = sign(X) \cdot \E[P|X]
    =X.
$
By the law of total variance,
$
\V[H]=\E[\V[H|X]]+\V[\E[H|X]].
$
In the latter term,
 $
    \E[H|X]=\E[Z-X|X]=0.
    $
In the former term,
    $
    \V[H|X]=\V[Z|X]=\E[Z^2|X]-\{\E[Z|X]\}^2.
    $
    Moreover,
    $
    \E[Z^2|X]=\E[P^2|X]=\V[P|X]+\{\E[P|X]\}^2=|X|+X^2.
    $
    In summary,
    $
    \V[H|X]=|X|. 
    $
\end{proof}

\textbf{Robustness to data dimensions}. In the main text, each sample from the simulated data generating process produces a matrix of covariates $\bX \in\R^{100\times 100}$ with rank $r=5$. How robust is our end-to-end procedure across realistic dimensions of economic data? We consider the following variations: $\bX \in \R^{50\times 200}$, $\R^{100 \times 100}$, $\R^{200\times 50}$, $\R^{500\times 20}$, and $\R^{1000\times 10}$. For each choice of $(n,p)$, we set the rank $r=\{\min(n,p)\}^{1/3}$. Across data dimensions, we introduce measurement error with the fixed noise-to-signal ratio of $20\%$. We consider the oracle tuning of the PCA hyperparameter $k=r$. 

\begin{wraptable}{r}{0.5\textwidth}
    \vspace{-13pt}
    \centering
    \resizebox{0.5\textwidth}{!}{%
    \begin{tabular}{ccccccc}
        \hline
        Meas. Err. & $n$ & $p$ & ATE & SE & 80\% CI & 95\% CI\tabularnewline
        \hline
        20\% & 50 & 200 & 2.21 & 0.56 & 0.82 & 0.96\tabularnewline
        20\% & 100 & 100 & 2.21 & 0.35 & 0.82 & 0.95\tabularnewline
        20\% & 200 & 50 & 2.22 & 0.21 & 0.81 & 0.95\tabularnewline
        20\% & 500 & 20 & 2.23 & 0.12 & 0.78 & 0.94\tabularnewline
        20\% & 1000 & 10 & 2.23 & 0.08 & 0.76 & 0.93\tabularnewline
        \hline
        \hline
        20\% & 722 & 30 & 2.26 & 0.12 & 0.78 & 0.92\tabularnewline
        \hline
    \end{tabular}
    }
    \vspace{-10pt}
    \caption{Our approach adapts to data shape}\label{fig:sizes}
    \vspace{-20pt}
\end{wraptable}

Table~\ref{fig:sizes} quantifies coverage performance. Different rows correspond to different data dimensions. We record the average point estimates, which are close to $\theta_0=2.2$. Next, we record the average standard errors, which adaptively decrease in length for larger sample sizes.  These confidence intervals are the correct length, since coverage is close to the nominal level. 

We repeat this exercise for the simulated data generating process with $\bX \in \R^{722 \times 30}$ and rank $r=5$. Table~\ref{fig:sizes} confirms that our procedure attains nearly nominal coverage.

\textbf{Can data corruption flip signs}? In the main text, we show that for the simulated data generating process with $\bX \in\R^{100\times 100}$ and rank $r=5$, OLS performs well with clean data and performs poorly with corrupted data. We investigate two follow-up questions. First, can data corruption flip the sign of OLS estimates, i.e. can it lead to negative point estimates when the average treatment effect is $\theta_0=2.2>0$? Second, can can data corruption flip the sign of OLS and 2SLS estimates in scenarios more similar to our real world example?

We find that data corruption can flip the sign of OLS estimates some of the time. In particular, measurement error with a $20\%$ noise-to-signal ratio is enough to flip the sign roughly one quarter of the time.  We repeat this exercise for the simulated data generating process with $\bX \in \R^{722 \times 30}$ and rank $r=5$. Flipping signs requires not only 20\% measurement error but also 10\% missingness. A similar fraction of OLS estimates have flipped signs.

Finally, we conduct a semi-synthetic sign flipping exercise. We consider the covariates of \cite{autor2013china} at the commuting zone level. Rather than a synthetic ATE, the estimand is the actual effect of import competition on manufacturing employment in a partially linear instrumental variable model. Flipping signs requires not only 20\% measurement error but also 20\% missingness. In this thought experiment, we take the reported effect from \cite{autor2013china} as the ground truth, we take the data set from \cite{autor2013china} as clean data, and we generate synthetic measurement error and missingness. 

\begin{wraptable}{r}{0.5\textwidth}
    \centering
    \resizebox{0.5\textwidth}{!}{%
    \begin{tabular}{cccc}
        \hline
        Data & Meas. Err. & Miss. Val. & Sign Flip\tabularnewline
        \hline
        $100\times100$ & 20\% & 0\% & 27\%\tabularnewline
        $722\times30$ & 20\% & 10\% & 22\%\tabularnewline
        Census & 20\% & 20\% & 9\%\tabularnewline
        \hline
    \end{tabular}
    }
    \vspace{-10pt}
    \caption{Data corruption can flip signs}\label{fig:flips}
    \vspace{-20pt}
\end{wraptable}

We summarize the results of these sign flipping exercises in Table~\ref{fig:flips}. The rows correspond to (i) synthetic data with $\bX \in \R^{100\times 100}$; (ii) synthetic data with $\bX \in \R^{722 \times 30}$; and (iii) semi-synthetic data from \cite{autor2013china}. We interpret the OLS and TSLS results as motivation for data cleaning before data analysis. Our procedure may be viewed as an extension of OLS and TSLS with simple data cleaning that we subsequently account for in our confidence intervals.

\textbf{Empirical application}. The variable definitions follow \cite{autor2013china}. In the authors' original specification \cite[Table 3, column 6]{autor2013china}, $X_{i,\cdot}\in\R^{14}$ consists of: a constant, an indicator for the 2000-2007 period, percentage of employment in manufacturing, percentage of college educated population, percentage of foreign-born population, percentage of employment among women, percentage of employment in routine occupations, average offshorability index of occupations, and Census division dummies.

In our augmented specification, $X_{i,\cdot}\in\R^{30}$ consists of variables from the original specification as well as additional variables in \cite[Appendix Table 2]{autor2013china}. These include percentages of the working age population: employed in manufacturing, employed in non-manufacturing, unemployed, not in the labor force, receiving disability benefits; average log weekly wages: manufacturing, non-manufactuing; average benefits per capita: individual transfers, retirement, disability, medical, federal income assistance, unemployment, TAA; and average household income per working age adult: total, wage and salary.

Figure~\ref{fig:semi_all_alt} provides analogous results to Figure~\ref{fig:semi_all}, where now we center and scale the covariates $X_{i,\cdot}\in\R^{30}$ before conducting the exercise. We arrive at similar conclusions.

\begin{figure}[H]
\begin{centering}
    \begin{subfigure}[b]{0.48\textwidth}
        \centering
        \includegraphics[width=0.8\textwidth]{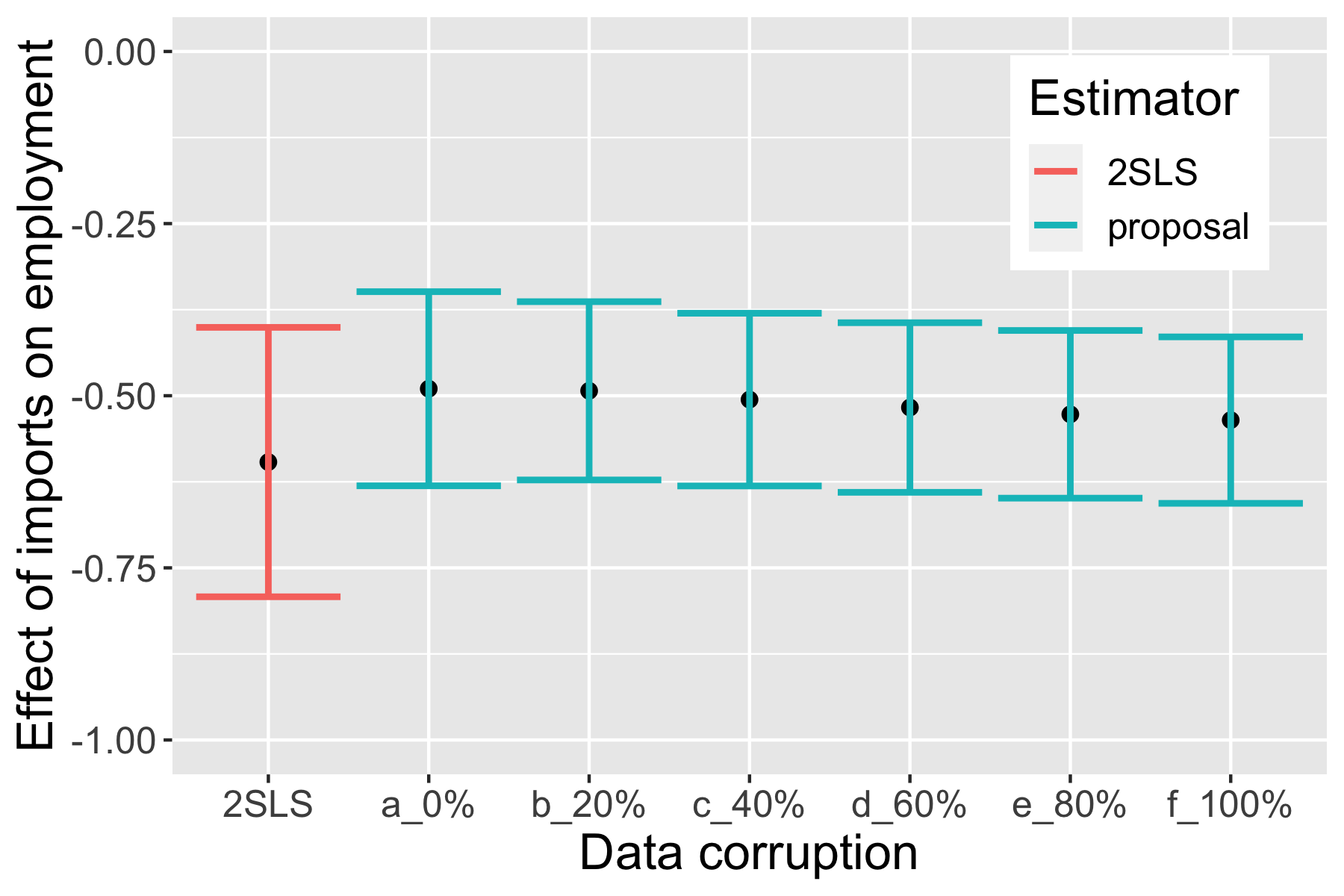}
        \vspace{-3pt}
        \caption{\label{fig:semi_noise_alt} Measurement error}
    \end{subfigure}
    \begin{subfigure}[b]{0.48\textwidth}
        \centering
        \includegraphics[width=0.8\textwidth]{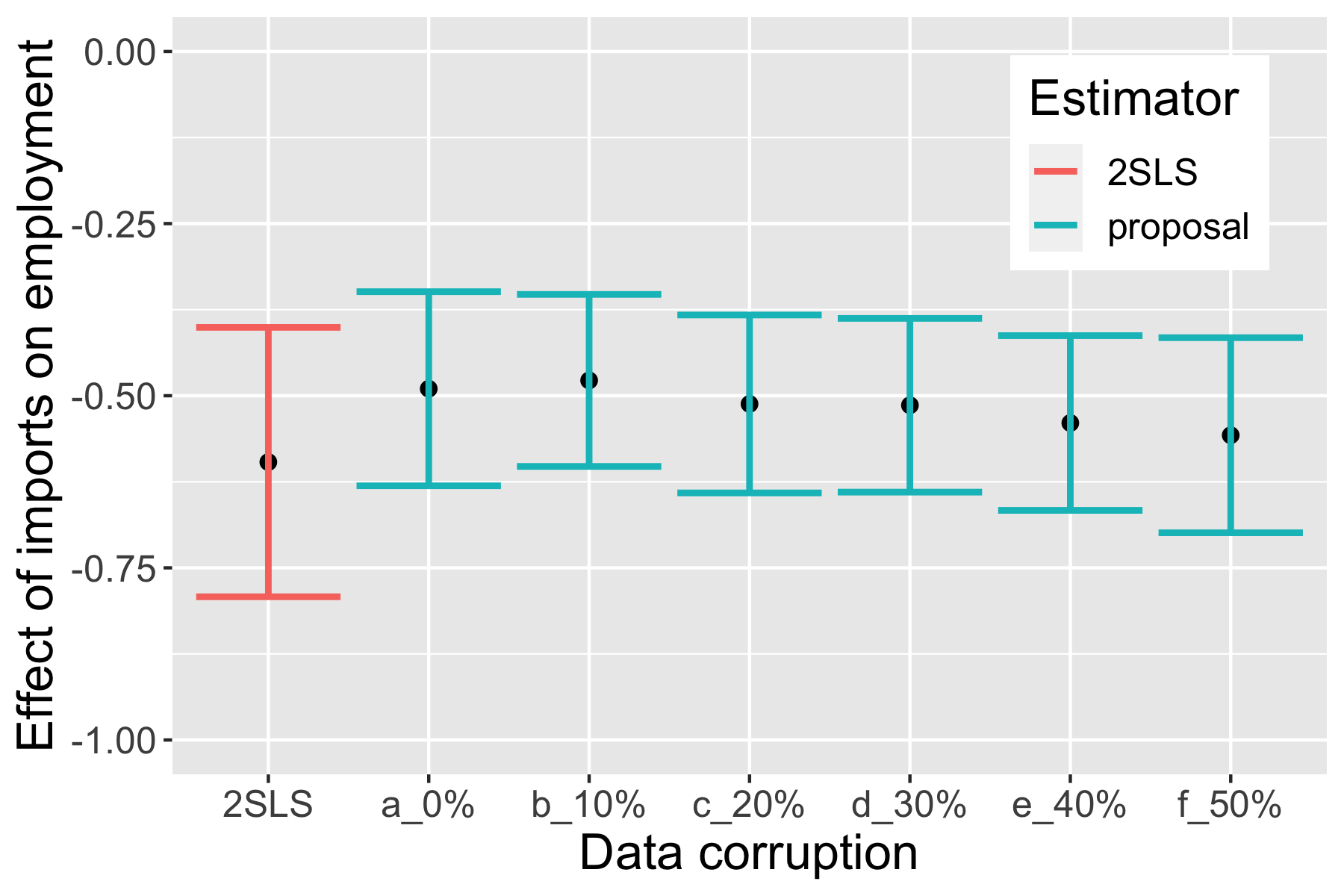}
        \vspace{-3pt}
        \caption{\label{fig:semi_missing_alt} Missing values}
    \end{subfigure}

    \begin{subfigure}[b]{0.48\textwidth}
        \centering
        \includegraphics[width=0.8\textwidth]{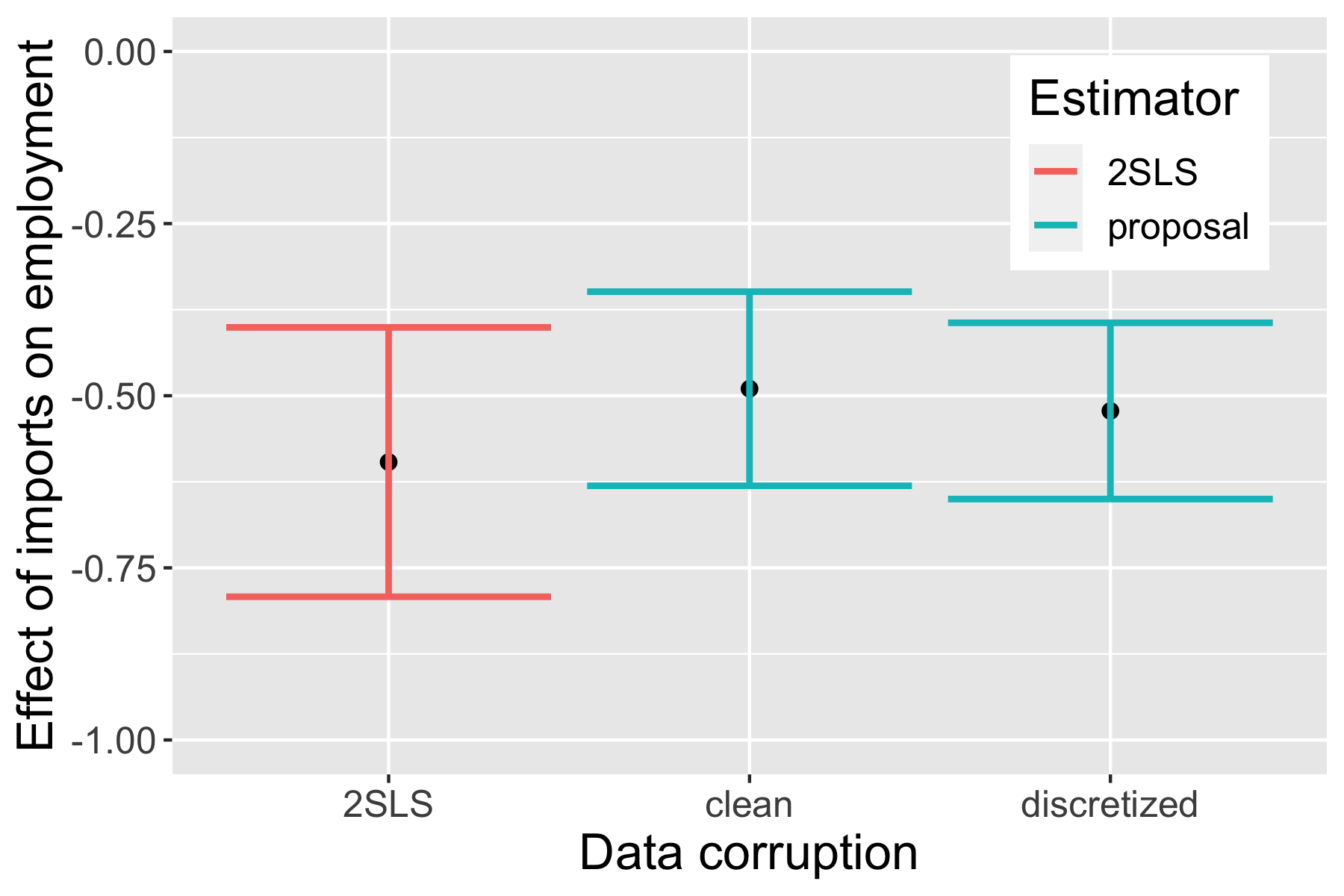}
        \vspace{-3pt}
        \caption{\label{fig:semi_discrete_alt} Discretization}
    \end{subfigure}
    \begin{subfigure}[b]{0.48\textwidth}
        \centering
        \includegraphics[width=0.8\textwidth]{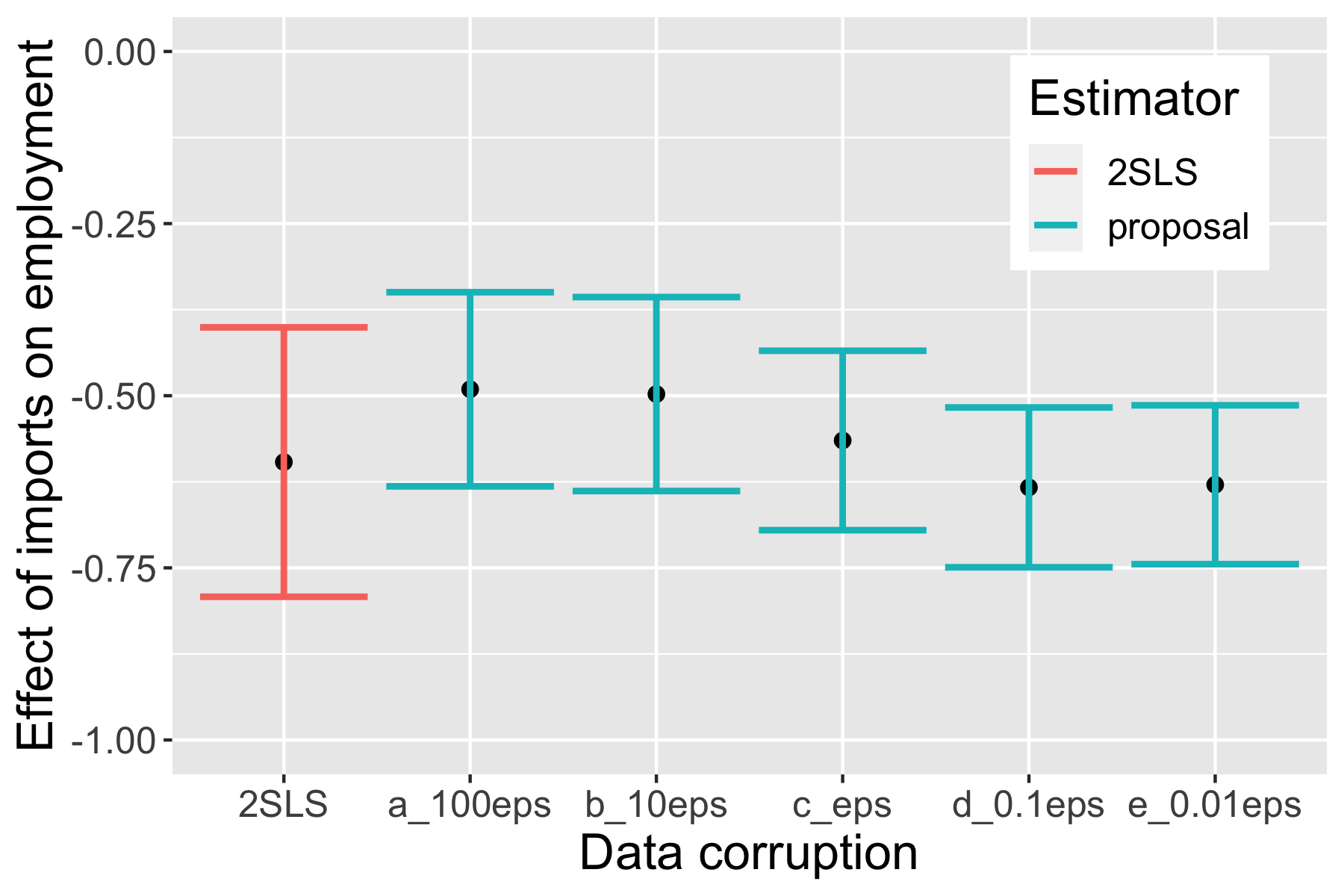}
        \vspace{-3pt}
        \caption{\label{fig:semi_calibrate_alt} Differential privacy (calibrated)}
    \end{subfigure}
\vspace{-10pt}
\caption{Standardizing covariates before synthetic corruption}\label{fig:semi_all_alt}
\vspace{-15pt}
\end{centering}
\end{figure}

%% file: M_privacy.tex
\section{Formalizing privacy}\label{sec:privacy}

\begin{proof}[Proof of Proposition~\ref{prop:macro}]
Fix the commuting zone $i\in[n]$. We refer to the construction of the summary statistic
$
X_{ij}=f_j(\bM^{(i)})=\frac{1}{L_i}\sum_{\ell=1}^{L_i} M^{(i)}_{\ell j}
$
as the $j$-th query $f_j$ about $\bM^{(i)}$, where $j\in[p]$. To ensure privacy level $\epsilon_j$ for query $f_j$, a possible mechanism is, according to \cite[Proposition 3.3]{dwork2006calibrating},
$
Z_{ij}=X_{ij}+H_{ij}$, where $X_{ij}=f_j(\bM^{(i)})$ and $H_{ij}\overset{i.i.d.}{\sim} \text{Laplace}(S(f_j)/\epsilon_j).
$
$S(f_j)$ is the sensitivity of the query, to which we return below. If no individual appears in two commuting zones, the Bureau can achieve privacy level $\epsilon$ while publishing all $j\in [p]$ variables for this commuting zone by setting $\epsilon_j=\epsilon/p$.

We wish to characterize the resulting subexponential parameters. They are, by independence of the Laplacians,
$
    K_a=\|H_{i,\cdot}\|_{\psi_a}=\max_{j\in[p]}\|H_{ij}\|_{\psi_a}=\max_j \sqrt{2}\cdot S(f_j)/\epsilon_j= \sqrt{2}/\epsilon \cdot p \max_j S(f_j)$ and $
\kappa^2=\|\E[H^T_{i,\cdot} H_{i,\cdot}]\|_{op}=\max_{ij}\V(H_{ij})=2\max_j S(f_j)^2/\epsilon_j^2= 2/\epsilon^2 \cdot p^2 \max_j S(f_j)^2.
$
What remains is to define and characterize the the sensitivity $S(f_j)$. The sensitivity of the query $f_j$ is the most that the query may vary if one individual in the microdata were replaced. Formally,
$
\max_{\bM^{(i)},\bM^{(i')}} |f_j(\bM^{(i)})-f_j(\bM^{(i')})|\leq S(f_j)
$
where $\bM^{(i)}$ and $\bM^{(i')}$ are two possible data sets of $L_i$ individuals that differ in one individual.

In what follows, we suppress indexing by $i$ to lighten notation. By hypothesis, each entry of microdata is bounded: $|M_{\ell j}|\leq\bar{A}$. This fact, together with the fact that the query $f_j$ is a sample mean, provides a bound on the sensitivity $S(f_j)$. To begin, write
$
      f_j(\bM)=\frac{1}{L}\left\{\sum_{\ell=1}^{L} M_{\ell j}\right\}
      =\frac{1}{L}\left\{\sum_{\ell=1}^{L-1} M_{\ell j}+M_{\ell L} \right\}.
$
   Therefore without loss of generality
$
     f_j(\bM)-f_j(\bM')=\frac{1}{L}(M_{\ell L}-M'_{\ell L})
 $
 and hence
$
   S(f_j)=\max_{\bM,\bM'} |f_j(\bM)-f_j(\bM')| =\max_{\bM,\bM'}\left|\frac{1}{L}(M_{\ell L}-M'_{\ell L})\right| 
   \leq \frac{2\bar{A}}{L}.
$ Lemma~\ref{lemma:post} ensures that privacy is preserved.
\end{proof}

\begin{lemma}\label{lemma:post}
    Suppose the conditions of Proposition~\ref{prop:macro} hold. If $Z_{i,\cdot}=X_{i,\cdot}+H_{i,\cdot}$ confers $\epsilon$ differential privacy, then $\hat{X}_{i,\cdot}$ remains $\epsilon$ differentially private.
\end{lemma}

\begin{proof}
  Extending the notation from the previous proof, let $Z_{i,\cdot}=\mathcal{M}(\bM^{(i)})=X_{i,\cdot}+H_{i,\cdot}$ and $Z_{i',\cdot}=\mathcal{M}(\bM^{(i')})=X_{i',\cdot}+H_{i,\cdot}$. Recall that $\hat{X}_{i,\cdot}$ is a function of $Z_{i,\cdot}$ and $\{Z_{j,\cdot}\}_{j\neq i}$, i.e. $\hat{X}_{i,\cdot}=\textrm{CLEAN}[Z_{i,\cdot};\{Z_{j,\cdot}\}_{j\neq i}]$. Analogously,  $\hat{X}_{i',\cdot}=\textrm{CLEAN}[Z_{i',\cdot};\{Z_{j,\cdot}\}_{j\neq i}]$.
  
  By hypothesis, for any event $E$, $
\frac{\p_{H_{i,\cdot}}(Z_{i,\cdot}\in E )}{\p_{H_{i,\cdot}}(Z_{i',\cdot}\in E)}\leq e^{\epsilon}
$ where the subscript emphasizes the source of randomness. We wish to show that, for any event $F$, $
\frac{\p_{H_{i,\cdot},\{H_{j,\cdot}\}_{j\neq i}}(\hat{X}_{i,\cdot}\in F )}{\p_{H_{i,\cdot},\{H_{j,\cdot}\}_{j\neq i}}(\hat{X}_{i',\cdot}\in F)}\leq e^{\epsilon}
$. Fix $F$ and define $G:=\{z\in\mathbb{R}:\textrm{CLEAN}[z;\{Z_{j,\cdot}\}_{j\neq i}]\in F\}$. Then
\begin{align*}
    &\p_{H_{i,\cdot},\{H_{j,\cdot}\}_{j\neq i}}(\hat{X}_{i,\cdot}\in F )
    =\p_{H_{i,\cdot},\{H_{j,\cdot}\}_{j\neq i}}(\textrm{CLEAN}[Z_{i,\cdot};\{Z_{j,\cdot}\}_{j\neq i}]\in F ) \\
    &=\p_{H_{i,\cdot},\{H_{j,\cdot}\}_{j\neq i}}(Z_{i,\cdot}\in G) 
    =\E_{\{H_{j,\cdot}\}_{j\neq i}}[\E_{H_{i,\cdot}}[\1(Z_{i,\cdot}\in G)|\{H_{j,\cdot}\}_{j\neq i}]].
\end{align*}
Moreover, $\E_{H_{i,\cdot}}[\1(Z_{i,\cdot}\in G)|\{H_{j,\cdot}\}_{j\neq i}]=\p_{H_{i,\cdot}}(Z_{i,\cdot}\in G)\leq e^{\epsilon}\cdot \p_{H_{i,\cdot}}(Z_{i',\cdot}\in G)=e^{\epsilon}\cdot \E_{H_{i,\cdot}}[\1(Z_{i',\cdot}\in G)|\{H_{j,\cdot}\}_{j\neq i}]$. Reversing the steps above yields the conclusion.
\end{proof}

\begin{lemma}[\cite{bun2016concentrated}]\label{lemma:conversion}
    If $\mathcal{M}$ is differentially private with parameter $\epsilon$, then it is zero concentrated differentially private with parameter $\rho=\epsilon^2/2$.
\end{lemma}

The Bureau's global privacy loss budget for people, in terms of zero concentrated differential privacy, is $2.56$ in 2020 Census redistricting data (P.L 94-171) \cite{abowd2022the}. Of this budget, $447/4,099$ is for counties. We use these numbers to calibrate a realistic privacy budget of $\rho=2.56\cdot 447/4,099$ for a hypothetical data release concerning commuting zones. Lemma~\ref{lemma:conversion} demonstrates that a sufficient degree of differential privacy is $\epsilon=(2\rho)^{1/2}$.